\documentclass[acmsmall,screen]{acmart}
\usepackage[T1]{fontenc}
\usepackage{graphicx}
\usepackage{lineno}
\usepackage{hyperref}
\usepackage{color}
    
\usepackage[dvipsnames, table]{xcolor}
\usepackage[capitalise]{cleveref}
\usepackage{hyperref}
\usepackage{thmtools}
\usepackage{thm-restate}
\usepackage{cancel}
\usepackage{enumitem}
\usepackage{mathrsfs}  
\usepackage[all]{xy}
\usepackage{multirow}
\usepackage{caption} 
\usepackage{float}
\usepackage{multicol}
\usepackage{soul}
\usepackage{hhline}
\usepackage{bbm}
\usepackage{mathtools}
\usepackage{bm}
\usepackage{ dsfont }
\usepackage{pifont}
\usepackage{array}
\usepackage{pdflscape}
\usepackage{mdframed}
\usepackage{tikz-cd}
\usepackage{tikz}
\usetikzlibrary{patterns.meta}
\usepackage{pgfplots}
    \pgfplotsset{compat=newest}
\usepackage[ttdefault=true]{AnonymousPro}
\usepackage{adjustbox}
\usepackage{makecell}
\usepackage{tabstackengine}
\usepackage{xspace}
\usepackage{soul}
\usepackage{wrapfig}
\usepackage[misc]{ifsym}
\usepackage{orcidlink}
\usepackage{scalerel}
\usepackage{xcolor}
    \hypersetup{
        colorlinks,
        linkcolor={green!50!black},
        citecolor={red!85!black},
        urlcolor={blue!85!black}
    }
\usepackage{comment}
\usepackage[ruled]{algorithm}
\usepackage[noend]{algpseudocode}
\algnewcommand{\TRUE}{\textbf{true}}
\algnewcommand{\FALSE}{\textbf{false}}

\newcommand{\wit}{\textit{wit}}

\newcommand{\overarrow}{\overrightarrow}

\newcommand{\Exists}[1]{\exists\,#1.\:}

\newcommand{\Forall}[1]{\forall\,#1.\:}

\newcommand{\A}{\ensuremath{\mathcal{A}}}

\newcommand{\B}{\ensuremath{\mathcal{B}}}

\newcommand{\sA}{\ensuremath{\mathbb{A}}}

\newcommand{\T}{\ensuremath{\mathcal{T}}}

\newcommand{\vars}{\textit{vars}}

\newcommand{\N}{\ensuremath{\mathbb{N}_{\omega}}}

\newcommand{\qf}[1]{\ensuremath{QF(#1)}}

\newcommand{\minmod}{\ensuremath{\mathrm{minmod}}}

\newcommand{\Tgeqn}{\ensuremath{\T_{\geq n}}}

\newcommand{\Tinfty}{\ensuremath{\T_{\infty}}}

\newcommand{\Tleqn}{\ensuremath{\T_{\leq n}}}

\newcommand{\NEQ}{\ensuremath{\neq}}

\newcommand{\NNEQ}[1]{\ensuremath{\neq(#1_{1},\ldots,#1_{n})}}

\newcommand{\NNNEQ}[2]{\ensuremath{\neq(#1_{1},\ldots,#1_{#2})}}

\newcommand{\distinct}[1]{\NNEQ{x}}

\renewcommand{\int}[2]{\mathcal{#1}/\mathcal{#2}}

\newcommand{\Tle}{\ensuremath{\T_{\leq}}}

\newcommand{\Tgr}{\ensuremath{\T_{\geq}}}

\newcommand{\Sp}{\ensuremath{\Sigma_{P}}}

\newcommand{\Spn}{\ensuremath{\Sigma_{P}^{\mathbb{N}}}}

\newcommand{\reference}{}

\newcommand{\No}{\ensuremath{\mathbb{N}^{*}}}

\newcommand{\Inf}{\ensuremath{\aleph_{0}}}

\newcommand{\spec}{\ensuremath{\textit{Spec}}}

\newcommand{\dom}[1]{\ensuremath{dom(#1)}}

\newcommand{\quagen}{$\mathfrak{F}$-QG}

\newcommand{\coquagen}{co-$\mathfrak{F}$-QG}

\newcommand{\Th}[1]{\ensuremath{\T_{\langle\scriptscriptstyle #1\rangle}}}

\newcommand{\algofont}[1]{\textnormal{\textsc \selectfont\sffamily  #1}}
\newcommand{\procform}{\algofont{ProcessFormula}}
\newcommand{\procnum}{\algofont{ProcessNumber}}

\newcommand{\Teq}{\ensuremath{\T_{\scriptscriptstyle\textbf{eq}}}}

\newcommand{\class}[1]{\ensuremath{\mathfrak{T}_{\text{#1}}}}

\newcommand{\Testsi}{\ensuremath{\T_{>n}^{P}}}

\newcommand{\Testcfs}{\ensuremath{\T_{=}^{P}}}

\newcommand{\Testsm}{\ensuremath{\T^{n}_{m}}}

\newcommand{\Testg}{\ensuremath{\T_{\leq}^{S}}}

\newcommand{\Teqn}{\ensuremath{\T_{=n}}}

\newcommand{\Sigmasigma}{\ensuremath{Sig^{P}_{\Sigma}}}

\newcommand{\Gal}{\ensuremath{G}}

\newcommand{\unc}{\ensuremath{U}}

\tikzset{
    rotated halfcircle/.style={%
        mark=halfcircle*,
        mark color=black,
        fill=red,
        every mark/.append style={rotate=#1}
    }
}

\newcolumntype{P}[1]{>{\centering\arraybackslash}p{#1}}
\Crefname{theorem}{Theorem}{Theorems}
\Crefname{lemma}{Lemma}{Lemmas}
\Crefname{corollary}{Corollary}{Corollaries}
\Crefname{example}{Example}{Examples}
\Crefname{proposition}{Proposition}{Propositions}
\Crefname{definition}{Definition}{Definitions}
\Crefname{conjecture}{Conjecture}{Conjectures}

\newtheorem{remark}{Remark}[section]

\let\tp\texorpdfstring

\setcopyright{cc}
\setcctype{by}
\acmDOI{10.1145/3776652}
\acmYear{2026}
\acmJournal{PACMPL}
\acmVolume{10}
\acmNumber{POPL}
\acmArticle{10}
\acmMonth{1}
\received{2025-07-09}
\received[accepted]{2025-11-06}

\begin{document}

\title{Characterizing Sets of Theories That Can Be Disjointly Combined}

\author{Benjamin Przybocki}
\orcid{0009-0007-5489-1733}
\affiliation{%
  \institution{Carnegie Mellon University}
  \city{Pittsburgh}
  \country{USA}
}
\email{bprzyboc@andrew.cmu.edu}

\author{Guilherme V. Toledo}
\orcid{0000-0002-6539-398X}
\affiliation{%
  \institution{Bar-Ilan University}
  \city{Ramat Gan}
  \country{Israel}
}
\email{guivtoledo@gmail.com}

\author{Yoni Zohar}
\orcid{0000-0002-2972-6695}
\affiliation{%
  \institution{Bar-Ilan University}
  \city{Ramat Gan}
  \country{Israel}
}
\email{yoni206@gmail.com}


\begin{abstract}
    We study properties that allow first-order theories to be disjointly combined, including stable infiniteness, shininess, strong politeness, and gentleness. Specifically, we describe a Galois connection between sets of decidable theories, which picks out the largest set of decidable theories that can be combined with a given set of decidable theories. Using this, we exactly characterize the sets of decidable theories that can be combined with those satisfying well-known theory combination properties. This strengthens previous results and answers in the negative several long-standing open questions about the possibility of improving existing theory combination methods to apply to larger sets of theories. Additionally, the Galois connection gives rise to a complete lattice of theory combination properties, which allows one to generate new theory combination methods by taking meets and joins of elements of this lattice. We provide examples of this process, introducing new combination theorems. We situate both new and old combination methods within this lattice.
\end{abstract}



\maketitle

\section{Introduction}

Given decidable first-order theories $\T_1$ and $\T_2$ with disjoint signatures, when is their combination $\T_1 \sqcup \T_2$ decidable? This question has spawned an entire subfield of logic called \emph{theory combination} (see \cite{DBLP:conf/birthday/BonacinaFRT19} for a survey). Theory combination goes back to Nelson and Oppen's early work on decision procedures~\cite{NelsonOppen} and has since become a central topic in Satisfiability Modulo Theories (SMT)~\cite{BSST21}, a field with applications ranging from hardware and software verification to automated theorem proving. Since SMT solvers often need to reason about multiple theories at once, they make essential use of theory combination methods.

The applicability of SMT solvers in software verification stems from their support for multiple theories, and in particular, their combinations.
Important examples include
the theories of 
bit-vectors~\cite{DBLP:series/txtcs/KroeningS16}, floating points~\cite{brain2019building}, and
strings~\cite{chen2022solving}, which are used to model primitive
types in programming languages; 
as well as the theories of 
uninterpreted functions~\cite{DBLP:series/txtcs/KroeningS16},
data types~\cite{kovacs2017coming,barrett2007abstract}, and
arrays~\cite{de2009generalized}, which
are often used to model container types.
While reasoning about certain programs only requires
a single theory (e.g., purely numerical functions or generic functions on lists),
software verification often relies
on combinations of the aforementioned theories,
as well as other theories (e.g., for reasoning about lists of machine integers).

Despite the importance of theory combination for SMT solving and the serious attention it has received from researchers, our theoretical understanding of theory combination remains unsatisfactory. Several conditions on $\T_1$ and $\T_2$ that are \emph{sufficient} for their combination to be decidable have been identified. On the other hand, some examples are known of decidable theories whose combination is undecidable, which have been used to show that certain conditions do not suffice for theory combination~\cite{Bonacina,CADE30}. However, prior to the present work, and to the best of our knowledge, no necessary and sufficient condition for a theory to be combinable with some set of theories has been proven. In particular, it has remained open whether existing theory combination methods can be improved to apply to larger sets of theories. In this paper, we provide a general framework for studying theory combination properties, which allows us to prove necessary and sufficient conditions for a theory to be combinable with every theory satisfying certain properties. Our results apply to all of the well-known disjoint theory combination methods: Nelson--Oppen~\cite{NelsonOppen}, shiny~\cite{shiny}, polite~\cite{polite,JB10-LPAR}, and gentle~\cite{gentle}.\footnote{One combination method we do not discuss is Shostak's method~\cite{shostak}, which can be viewed as a refinement of the Nelson--Oppen method that is 
less generally applicable~\cite{DBLP:conf/frocos/BarrettDS02,shostak-light}.} Our techniques also facilitate the discovery of new combination methods, which we demonstrate by introducing several new methods.

To motivate and contextualize our results, we start with a non-exhaustive history of theory combination
and situate in it some of the implications of our results.

\subsection{History of Theory Combination}

Theory combination as a field began in 1979 when Nelson and Oppen~\cite{NelsonOppen} proposed their combination method, which was proved correct by Oppen~\cite{OppenSI} subject to a technical condition he called \emph{stable infiniteness}. Specifically, the Nelson--Oppen theorem states that two disjoint and decidable theories can be combined so long as both are stably infinite.

The Nelson--Oppen method was the state-of-the-art for two decades and continues to be widely used, but the requirement of stable infiniteness has proven to be a hindrance in some applications, as many practical theories are not stably infinite (e.g., the theory of fixed-width bit-vectors). To address this limitation, in 2003, Tinelli and Zarba~\cite{shiny} proved a new combination theorem using a property they called \emph{shininess}. Shininess is a stronger property than stable infiniteness, but shiny theories can be combined with any other decidable theory, including those that are not stably infinite.

The ability to combine a theory with any other decidable theory is very convenient, so researchers naturally began looking for other sets of theories with this property. In 2005, Ranise, Ringeissen, and Zarba~\cite{polite} introduced \emph{polite theory combination}. They claimed that so-called \emph{polite} theories can also be combined with any other decidable theory. But, in 2010, Jovanovi{\'c} and Barrett~\cite{JB10-LPAR} discovered a mistake in the proof of their combination theorem. They were however able to recover the result by imposing a stronger property than politeness, what later came to be called \emph{strong politeness}~\cite{CasalRasga2}. Strongly polite theory combination is very practical and is used in the SMT solver cvc5~\cite{cvc5}.

There were still several loose ends to tie up regarding politeness. First, it was unknown whether strong politeness was actually strictly stronger than politeness; that is, whether there are any theories that are polite but not strongly polite. A negative answer to this question would have been desirable, 
because establishing politeness of a theory is easier than establishing strong politeness.
There were some promising initial results by Casal and Rasga~\cite{CasalRasga} in 2013 and later by Sheng et al.~\cite{polite-alg} in 2020 proving conditions under which politeness implies strong politeness. But then, in 2021, Sheng et al.~\cite{sheng-si-polite} found an example of a polite theory that is not strongly polite. Still, one could have hoped that the polite combination criterion might still be correct, albeit with a different proof. Alas, we recently dashed this hope by finding a pair of decidable theories, one of which is polite, whose combination is undecidable~\cite{CADE30}. This refuted the polite combination criterion, but it left open other possibilities that could provide other possible applications of politeness. Even if polite theories cannot be combined with \emph{every} decidable theory (as originally claimed), is there some set of theories with which they can be combined (besides the stably infinite theories, given that politeness implies stable infiniteness)? Or is there some property intermediate in strength between politeness and strong politeness that suffices for combinability with every decidable theory? We answer both questions negatively.

The other loose end was investigating the relations between strong politeness and other theory combination properties. In 2018, Casal and Rasga~\cite{CasalRasga2} showed, perhaps surprisingly, that shininess and strong politeness coincide for decidable theories. In 2023, Toledo, Zohar, and Barrett~\cite{CADE} proved that theories satisfying a property seemingly weaker than strong politeness can be combined with any other decidable theory; however, they conjectured that their property is in fact equivalent to strong politeness, which was later confirmed in 2024 by Przybocki et al.~\cite{nounicorns}. These equivalences are nontrivial, and it has been somewhat mysterious why every attempt to improve the shiny combination theorem has resulted in a property equivalent to shininess. Casal and Rasga~\cite{CasalRasga2} expressed a desire to find ``a class of theories strictly containing the shiny/strongly polite theories for which there exists an indiscriminate Nelson--Oppen method, in the sense that they can be combined with an arbitrary theory with a decidable quantifier-free satisfiability problem.'' We prove that there is no such class.

Another combination method intended to address the limitations of the Nelson--Oppen method is \emph{gentle theory combination}, introduced by Fontaine~\cite{gentle} in 2009. Gentleness is weaker than shininess and incomparable with stable infiniteness. Fontaine proved that gentle theories can be combined with any theory satisfying a mild technical condition. Recently, \cite{CADE30} proved that gentle theories can be combined with any theory having \emph{computable finite spectra}, which is weaker than the condition Fontaine gave. Most theories of practical interest have computable finite spectra, but it is a theoretically interesting question whether this condition can be weakened even further. We prove that it cannot.

All of the above combination theorems that came after Nelson--Oppen require one theory to satisfy a property strictly stronger than or incomparable with stable infiniteness. Over four decades after Nelson and Oppen's seminal paper, we still have no combination method that is strictly more general than theirs, and not for a lack of interest or effort. In 1996, Tinelli and Harandi~\cite{tinelli-new} conjectured ``that there might exist weaker requirements on the component theories which are sufficient for the [Nelson--Oppen] procedure's correctness''. We finally disprove this conjecture, assuming we interpret it as stating that the requirements for \emph{both} component theories can be strictly weaker than stable infiniteness.

\subsection{Our Contributions}

First, in \Cref{sec-galois}, we introduce a Galois connection between sets of decidable theories,
in which every set of decidable theories corresponds to the largest set of decidable theories with which they can be combined. This Galois connection forms the core of our abstract framework for studying theory combination. It induces a complete lattice of sets of decidable theories, each of which is involved in some combination theorem. It also induces a closure operator on sets of decidable theories, which can be seen as a tool to improve a combination theorem to a version that applies to as many theories as possible.

Second, in \Cref{sec-known}, we show that well-known disjoint theory combination methods cannot be further improved. Specifically, we show that shiny theories are the largest set of decidable theories that can be combined with all decidable theories (\Cref{thm-shiny-tight}), stably infinite theories are the largest set of decidable theories that can be combined with all stably infinite theories (\Cref{thm-f-si}), and gentle theories are the largest set of decidable theories that can be combined with all theories with computable finite spectra (\Cref{thm-f-cfs}) and vice versa (\Cref{thm-f-gentle}). This shows that the following combination theorems cannot be improved to apply to larger sets of theories: the shiny combination theorem~\cite{shiny}, the strongly polite combination theorem~\cite{JB10-LPAR}, the Nelson--Oppen combination theorem~\cite{OppenSI}, and the gentle combination theorem~\cite{gentle} after incorporating an improvement from~\cite{CADE30}. 
Our proof of \Cref{thm-f-si} also shows that the original polite combination theorem~\cite{polite} is unsalvageable: the largest set of theories with which polite theories can be combined is the set of stably infinite theories. This concludes a line of research about the possibility of repairing this claimed theorem~\cite{CasalRasga,polite-alg,sheng-si-polite,CADE30}. 
As a byproduct of our proofs, we also obtain results about ``maximally difficult'' theories from the standpoint of theory combination (e.g., \Cref{thm-shiny-uniform}), which may be of independent interest.

One of the virtues of our abstract framework is that it facilitates the discovery of new combination theorems, which we demonstrate in \Cref{sec-new},
where we also prove their sharpness.
Specifically, \Cref{sec-sm-cs} sharpens a combination theorem from \cite{CADE30} using our technique. 
In \Cref{sec-cs}, we show that by taking meets and joins of elements of the lattice induced by the Galois connection, we can use existing combination theorems to discover a new combination theorem (\Cref{thm-cs}), which is sharp. This method turns out to be an improvement of a combination theorem from \cite{Bonacina}. In \Cref{sec-n-decidable}, we show yet another way to use our abstract framework to discover new combination theorems: start with a natural set of decidable theories and take its closure with respect to the Galois connection. We use this to discover a countably infinite family of combination theorems, which are sharp by construction and substantially improve a family of combination theorems from \cite{manna-zarba-survey}. In \Cref{sec-quasi-gentle}, we go further by exhibiting an uncountably infinite family of combination theorems, each of which is sharp. These combination theorems can be seen as natural variants of gentleness, and their cardinality demonstrates how rich the space of combination theorems is.

Finally, in \Cref{sec-lattice}, we situate all of the above combination theorems within the lattice induced by the Galois connection, allowing one to visualize the relationships between these theorems. We also make a few remarks about the properties of this lattice.

All of our results deal with one-sorted theories.
We leave their generalization to many-sorted theories for future work.

To keep the paper concise and easy to read, numerous technical proofs are deferred to the appendix.

\subsection{Related Work}

To the best of our knowledge, only two papers have studied \emph{negative} results in theory combination, that is, results demonstrating that theories cannot be combined in certain cases. Since negative results are necessary for characterizing when theories can be combined, these papers have a close affinity to our work.

In 2006, Bonacina et al.~\cite{Bonacina} gave an example of a pair of decidable theories, only one of which is stably infinite, whose combination is undecidable. This is notable because it shows that the Nelson--Oppen theorem cannot be weakened to only assume that one of the component theories is stably infinite.

Recently, in 2025, we extended the previous work to prove that other combination methods cannot be extended in certain ways~\cite{CADE30}. That is, for each well-known combination method (namely, Nelson--Oppen, shiny, strongly polite, and gentle), we constructed a pair of decidable theories meeting all but one of the requirements for the method such that their combination is undecidable. This shows that none of the requirements for these theory combination methods can be omitted, but it does not tell us whether the requirements can be weakened. Indeed, we showed in the same paper that one of the requirements for gentle theory combination can be weakened.

The present paper can be seen as a continuation of this line of work given that we make essential use of (variants of) the theories constructed in those papers. The main innovation of the present paper is to realize that these theories can be used as what we call \emph{test theories} to deduce properties of theories with which they can be combined. The idea is that, given an arbitrary decidable theory $\T_1$, if we assume it can be combined with a test theory $\T_2$, then that forces $\T_1$ to satisfy some desired property. Probing theories with enough carefully constructed test theories allows us to prove upper bounds on the set of decidable theories that can be combined with a given set of theories. If we have a matching lower bound coming from a theory combination method, this gives a sharp characterization.

\section{Preliminaries}
We go over the relevant concepts from first-order logic and theory combination in \Cref{sec:FOL,sec:spec,sec:thcomb}.
A special result in theory combination that we use throughout the paper, which we call
{\em Fontaine's lemma}, is described in \Cref{sec:fontainelemma}.
Finally, since we arrange combination methods in a lattice
induced by a Galois connection, we review
basic concepts of lattice theory in \Cref{sec:lattice}.

In this paper, $\mathbb{N}$ is the set of natural numbers (including 0), $\N = \mathbb{N}\cup\{\aleph_{0}\}$, and $\No = \mathbb{N}\setminus\{0\}$. Let $|X|$ denote the cardinality of the set $X$, let $\aleph_{0}=|\mathbb{N}|$, and say that $S\subseteq \No$ is \emph{cofinite} if $\No\setminus S$ is finite. A set $S\subseteq\N$ is \emph{decidable} when there exists an algorithm that takes an $n\in\N$ and returns whether or not $n\in S$. Given $m,n\in \N$, $[m,n]$ denotes the set of elements $m\leq i\leq n$, which may include $\aleph_{0}$; $[n]$ is an abbreviation for $[1,n]$. $\vars(\varphi)$ is the set of free variables in $\varphi$.
Given a finite set of variables $V$ and an equivalence relation $E$ on it (if $x$ and $y$ are related according to $E$ we write $xEy$, and otherwise $x\cancel{E}y$), the \emph{arrangement} $\delta_{V}$ induced by $E$ on $V$ is a conjunction of the literals $x=y$, if $xEy$, or $\neg(x=y)$, if $x\cancel{E}y$.

\subsection{First-Order Logic}
\label{sec:FOL}
A \emph{signature} $\Sigma$ is a collection of symbols for functions and predicates, each coupled with an arity, having at least a binary predicate $=$ to be called equality.
A signature is \emph{empty} when it has no functions and no predicates other than equality;
two signatures are called \emph{disjoint} if they only share equality.
All signatures in this paper are assumed to be countable, i.e. have at most $\aleph_{0}$ symbols.
We define terms, atomic formulas, literals, cubes (conjunctions of literals), formulas, and sentences (formulas with no free variables) in the standard way.
$QF(\Sigma)$ is the set of quantifier-free formulas over the signature $\Sigma$.

A $\Sigma$-\emph{structure} $\sA$ is a nonempty set $\dom{\sA}$ (called the \emph{domain} of the structure) equipped with functions $f^{\sA}:\dom{\sA}^{m}\rightarrow\dom{\sA}$ and predicates $P^{\sA}\subseteq \dom{\sA}^{n}$, for all $f$ and $P$ in $\Sigma$ (where $m$ is the \emph{arity} of $f$ and $n$ that of $P$).
A $\Sigma$-\emph{interpretation} $\A$ is a $\Sigma$-structure where we have one value $x^{\A}\in\dom{\A}$ for every variable $x$; we sometimes call interpretations \emph{models}.
The value assigned to a term $\tau$ in $\A$ is
defined as usual, and is denoted by $\tau^{\A}$;
if $\Gamma$ is a set of terms, $\Gamma^{\A}=\{\tau^{\A} : \tau\in\Gamma\}$.
When $\A$ satisfies a formula $\varphi$ we may write $\A\vDash\varphi$.
Some formulas dealing with cardinalities are found in \Cref{card-formulas}: $\A$ satisfies $\psi_{\geq n}$, $\psi_{\leq n}$, or $\psi_{=n}$ if $\dom{\A}$ has at least, at most, or exactly $n$ elements.

\begin{figure}[t]
\begin{mdframed}
\begin{equation*}
\begin{aligned}
    \neq(x_{1},\ldots,x_{n}) \coloneqq \bigwedge_{i=1}^{n-1}\bigwedge_{j=i+1}^{n}\neg(x_{i}=x_{j})\\
    \\\vspace{-4mm}
    \psi_{\geq n} \coloneqq \Exists{{x_1,\ldots, x_n}}\NNEQ{x}
\end{aligned}
\quad
\begin{aligned}
\psi_{\leq n} \coloneqq \Exists{x_1,\ldots,x_n}\Forall{y}\bigvee_{i=1}^{n}y=x_{i}\\
\\\vspace{-4mm}
\psi_{=n} \coloneqq \psi_{\geq n}\wedge\psi_{\leq n}
\end{aligned}
\end{equation*}
\end{mdframed}
\caption{Cardinality formulas}
\label{card-formulas}
\end{figure}

A \emph{$\Sigma$-theory} $\T$ is a set of sentences over $\Sigma$, which are to be thought of as axioms for the theory. We refer to interpretations satisfying $\T$ as \emph{$\T$-interpretations}. A formula is said to be $\T$-\emph{satisfiable} when there is a $\T$-interpretation where it is satisfied, and a set of formulas is \emph{$\T$-satisfiable} if there is a $\T$-interpretation that satisfies all formulas in the set simultaneously;
two formulas are $\T$-\emph{equivalent} when a $\T$-interpretation satisfies one if and only if it satisfies the other.

$\T$ is \emph{decidable} when there is an algorithm that takes a quantifier-free formula and returns whether it is $\T$-satisfiable or not. 

Given two signatures $\Sigma_1$ and $\Sigma_2$, let $\Sigma_1 \sqcup \Sigma_2$ be their disjoint union; that is, $\Sigma_1 \sqcup \Sigma_2$ is the signature containing all symbols of $\Sigma_1$ and $\Sigma_2$, renaming the elements of $\Sigma_1$ and $\Sigma_2$ if necessary to make the signatures disjoint. Then, for decidable theories $\T_1$ and $\T_2$ defined over signatures $\Sigma_1$ and $\Sigma_2$ respectively, let $\T_1 \sqcup \T_2$ be the $\Sigma_1 \sqcup \Sigma_2$-theory axiomatized by the union of $\T_1$ and $\T_2$ after renaming elements of the signatures.
We say that $\T_1$ and $\T_2$ are \emph{combinable} if $\T_1 \sqcup \T_2$ is decidable.

\subsection{Spectra, L{\"{o}}wenheim--Skolem, and Compactness}
\label{sec:spec}
The \emph{spectrum} of a formula $\varphi$ in a theory $\T$, denoted $\spec_\T(\varphi)\subseteq \N$, is the set of $n\in\N$ such that there is a $\T$-interpretation $\A$ with $|\dom{\A}|=n$ and $\A\vDash\varphi$. The restriction to $\N$ as opposed to the class of all cardinals is justified by the L{\"{o}}wenheim--Skolem theorem:
\begin{theorem}[{\cite[Theorems~2.3.4 and 2.3.7]{marker2002}}] \label{LowenheimSkolem}
    Let $\Delta$ be a set of formulas over a countable signature, and let $\kappa \ge \aleph_0$. Then, $\Delta$ is satisfied by an interpretation of size $\kappa$ if and only if it is satisfied by an interpretation of size $\aleph_0$.
\end{theorem}

Another standard fact we will use is the compactness theorem:
\begin{theorem}[{\cite[Theorem~2.1.4]{marker2002}}] \label{compactness}
    Let $\Delta$ be a set of formulas over a countable signature. Then $\Delta$ is satisfiable if and only if every finite subset of $\Delta$ is satisfiable.
\end{theorem}
The compactness theorem applied to the set $\{\psi_{\ge n} : n \in \No\}$ has the following consequence, which we will use repeatedly:
\begin{corollary}
    If $n \in \spec_\T(\varphi)$ for infinitely many $n \in \No$, then $\aleph_0 \in \spec_\T(\varphi)$.
\end{corollary}

\subsection{Theory Combination Properties}
\label{sec:thcomb}
We now survey various properties of first-order theories,
and provide useful notations for the ones we focus on.
The notations are summarized in \Cref{tab:notation-prop-th}.

\begin{table}[t]
\caption{Notations for sets of theories. Each set only contains \emph{decidable} theories with the stated properties.}
\begin{tabular}{|c|c|}\hline
Symbol & Description \\\hline
$\class{}$ & --- \\\hline
$\class{ID}$ &  Infinitely decidable\\\hline
$\class{CFS}$ &  Computable finite spectra\\\hline
$\class{CS}$ &  Computable spectra\\\hline
$\class{gentle}$ &  Gentle\\\hline
$\class{SI}$ &  Stably infinite \\\hline
$\class{SM+CS}$ & Smooth + computable spectra \\\hline
$\class{shiny}$ &  Shiny\\\hline
\end{tabular}
\label{tab:notation-prop-th}
\end{table}

The set of decidable theories (over countable signatures) is denoted $\class{}$.\footnote{Technically, this definition of $\class{}$ makes it a proper class rather than a set, since we have imposed no restriction on which symbols are allowed to be used in a signature. 
We can avoid this problem by requiring that every signature only contain symbols drawn from some fixed set of allowed symbols.}

$\T$ is \emph{infinitely decidable} when there is an algorithm that takes a quantifier-free formula and returns whether $\aleph_{0}\in \spec_\T(\varphi)$ \cite{CADE30};
the set of theories that are both decidable and infinitely decidable
is denoted $\class{ID}$.
$\T$ has \emph{computable finite spectra} when there is an algorithm that takes a quantifier-free formula and an $n\in\No$ and returns whether $n\in\spec_\T(\varphi)$ \cite{CADE30};
the set of decidable theories with computable finite spectra is denoted $\class{CFS}$.
A theory has \emph{computable spectra} (a notion first defined in the present paper) when it is infinitely decidable and has computable finite spectra;
the set of decidable theories with computable spectra is denoted $\class{CS}$, and of course $\class{CS}=\class{CFS}\cap\class{ID}$.

$\T$ is \emph{gentle} when there is an algorithm that takes a quantifier-free formula $\varphi$ and returns a pair $(S,b)$ where: 
$S\subset \No$ is finite, and $b$ is a Boolean;
if $b$ is true, $\spec_\T(\varphi)=S$;
and if $b$ is false $\spec_\T(\varphi)=\N \setminus S$ \cite{gentle}.
The set of decidable gentle theories is denoted $\class{gentle}$.\footnote{In fact, every gentle theory is decidable, since we can compute whether $\spec_\T(\varphi) = \emptyset$.}
The \emph{minimal model function} of a decidable theory $\T$ is the function $\minmod_{\T}:QF(\Sigma)\rightarrow \N$ such that $\minmod_{\T}(\varphi)=\min\spec_\T(\varphi)$ whenever $\varphi$ is $\T$-satisfiable.\footnote{When $\T$ has the finite model property, this definition becomes the same as the original one found in \cite{shiny}.}

A theory $\T$ is said to be \emph{stably infinite} when, for every quantifier-free $\T$-satisfiable formula $\varphi$, there is a $\T$-interpretation $\A$ with $|\dom{\A}|\geq\aleph_{0}$ and $\A\vDash\varphi$ \cite{NelsonOppen};
equivalently (\Cref{LowenheimSkolem}): $\T$ is stably infinite when $\aleph_{0}\in\spec_\T(\varphi)$, for all $\T$-satisfiable quantifier-free formulas $\varphi$.
The set of decidable stably infinite theories is denoted $\class{SI}$.
A theory is said to have the \emph{finite model property} when, for every quantifier-free $\T$-satisfiable formula $\varphi$, there is a $\T$-interpretation $\A$ with $|\dom{\A}|<\aleph_{0}$ and $\A\vDash\varphi$;
to put it differently, if $\varphi$ is $\T$-satisfiable then $\spec_\T(\varphi)\cap\No\neq\emptyset$.
A theory $\T$ is \emph{smooth} when for every quantifier-free formula $\varphi$, $\T$-interpretation $\A$ that satisfies $\varphi$, and cardinal $m\geq|\dom{\A}|$, there exists a $\T$-interpretation $\B$ with $|\dom{\B}|=m$ and $\B\vDash\varphi$;
equivalently (using \Cref{LowenheimSkolem}): 
if $n\in\spec_\T(\varphi)$ and $m \geq n$, then $m\in\spec_\T(\varphi)$.

The set of decidable smooth theories with computable spectra is denoted $\class{SM+CS}$.\footnote{Notice that a decidable smooth theory is immediately infinitely decidable, so $\class{SM+CS}$ can be equivalently defined as the set of decidable smooth theories with computable \emph{finite} spectra. Also, every smooth theory $\T$ with computable spectra is decidable, since $\varphi$ is $\T$-satisfiable if and only if $\aleph_0 \in \spec_\T(\varphi)$.}
A theory is \emph{shiny} when it is smooth, has the finite model property, and has a computable minimal model function \cite{shiny};
the set of decidable shiny theories is denoted $\class{shiny}$.\footnote{It turns out that every shiny theory is decidable, a fact that was proven in~\cite{FroCoS2025}.}

\subsection{Fontaine's Lemma}
\label{sec:fontainelemma}
The possibility of disjointly combining two theories is captured by a result that we call \emph{Fontaine's lemma}, since a version of it was proven by Fontaine in \cite[Corollary~1]{gentle}.\footnote{A similar result was proven earlier by Tinelli and Harandi~\cite[Proposition~3.1]{tinelli-new}, and the essential idea goes back to Robinson~\cite{robinson-consistency}.}
\begin{restatable}[Fontaine's lemma]{lemma}{lemfontaine} \label{lem-fontaine}
    Let $\T_1$ and $\T_2$ be theories over the signatures $\Sigma_1$ and $\Sigma_2$, respectively. Then, the disjoint combination $\T_1 \sqcup \T_2$ is decidable if and only if the following problem is decidable: given conjunctions of literals $\varphi_1$ and $\varphi_2$ over $\Sigma_1$ and $\Sigma_2$ respectively, determine whether $\spec_{\T_1}(\varphi_1) \cap \spec_{\T_2}(\varphi_2) = \emptyset$.
\end{restatable}

In \cite[Corollary~1]{gentle}, Fontaine only stated and proved the ``if'' direction of this lemma; we state and prove the ``only if'' direction for completeness.

In one sense, Fontaine's lemma characterizes when theory combination is possible. But the lemma is seldom applied directly to a pair of theories to prove that their combination is decidable. This is because devising an algorithm to check whether $\spec_{\T_1}(\varphi_1) \cap \spec_{\T_2}(\varphi_2) = \emptyset$ is a difficult problem in general and requires a separate, ad hoc analysis for each pair of theories we wish to combine. Thus, the field of theory combination is chiefly concerned with which \emph{sets} of theories can be combined (equivalently, which properties of theories suffice for combinability), where Fontaine's lemma can be used as a tool to that end.

It is interesting to note that Fontaine's lemma does not require $\T_1$ and $\T_2$ to be decidable. And indeed, there are pairs of undecidable theories whose combination is decidable (e.g., if their combination is inconsistent). Nevertheless, we restrict our attention to decidable theories in this paper, as is typical in theory combination.

\subsection{Lattice Theory}
\label{sec:lattice}

A partially ordered set, or poset, $(A,\leq)$ is said to be a \emph{lattice} when the supremum (denoted $a\vee b$,
also called the {\em join} of $a$ and $b$) and infimum (denoted $a\wedge b$ and called the {\em meet}) of any two of its elements $a$ and $b$ exists; a lattice is \emph{complete} if the supremum and infimum of every subset $S\subseteq A$ exists.
An \emph{isomorphism} between posets $(A,\leq_{A})$ and $(B,\leq_{B})$ is a bijective function $f:A\rightarrow B$ such that $a\leq_{A}b$ if and only if $f(a)\leq_{B}f(b)$.
If $(A,\leq)$ is a poset, we define its \emph{dual} $(A,\leq^{\prime})$, over the same set $A$, by making $a\leq^{\prime}b$ if and only if $b\leq a$;
if $(A,\leq)$ is also a lattice then the supremum in its dual is its infimum and vice versa.

Given posets $(A,\le_A)$ and $(B,\le_B)$, a function $f:A\rightarrow B$ is said to be \emph{antitone} if $a\leq b$ implies $f(b)\leq f(a)$ for all $a,b\in A$.
If $f$ is a function from $(A,\leq)$ into itself, it is said to be an \emph{involution} if $f\circ f$ is the identity on $A$.
Given antitone functions $f : A \rightarrow B$ and $g : B \rightarrow A$, we say that $(f,g)$ is an \emph{antitone Galois connection} if
\[
    b \le_B f(a) \Longleftrightarrow a \le_A g(b)
\]
for all $a \in A$ and $b \in B$. In the sequel, we will use basic properties of Galois connections, which can be found, for example, in \cite{galois} and \cite[Chapter~V]{lattice-thy}.

A \emph{closure operator} on a poset $(A,\leq)$ is a function $Cl: A\rightarrow A$ such that for all $a \in A$:
$(i)$ $a\leq Cl(a)$;
$(ii)$ $a\leq b$ implies $Cl(a)\leq Cl(b)$;
and $(iii)$ $Cl(Cl(a))=Cl(a)$.
Given a closure operator, an element $a$ is said to be \emph{closed} when $Cl(a)=a$.

A \emph{chain} in a poset $(A,\leq)$ is a subset $X\subseteq A$ such that, for all $a,b\in X$, either $a\leq b$ or $b\leq a$;
in an \emph{antichain} $Y\subseteq A$, for any two $a,b\in Y$, neither $a\leq b$ nor $b\leq a$ unless $a=b$.
The \emph{height} of a poset is the supremum of the cardinalities of its chains;
the \emph{width} is the supremum of the cardinalities of its antichains.

\section{The Galois Connection} \label{sec-galois}
In this section we introduce a key ingredient for establishing the results of this paper.
Namely, we define a Galois connection between
sets of theories.
This Galois connection not only allows us to concisely and clearly describe sharpness of combination theorems
in terms of closures in \Cref{sec-known},
but will also give rise to 
the discovery of 
new combination theorems
in \Cref{sec-new}.

Antitone Galois connections formalize an intuitive concept. They are relationships between objects whose sizes vary inversely. The usual example from logic is that of formulas and models:
if you increase a set of formulas, the collection of models that satisfy them shrinks;
and if you increase a collection of models, the set of formulas satisfied by all of them shrinks.
We observe that theory combination also has an underlying Galois connection, as requiring more properties of the first component theory allows one more freedom when choosing the second.
For example, stably infinite theories are combinable with other stably infinite theories (by Nelson--Oppen);
but if we want one theory to be less than stably infinite, say only decidable, then we need more from the other, namely being shiny. The Galois connection we describe formalizes this trade-off.
 
Formally, let $\mathbb{P}(\class{})$ be the power set of $\class{}$, which forms a complete lattice ordered by the $\subseteq$ relation. We define a function $\Gal : \mathbb{P}(\class{}) \to \mathbb{P}(\class{})$ as follows. Given $X \subseteq \class{}$, let $\Gal(X)$ be the set of decidable theories that are combinable with every theory in $X$; that is, 
\[
    \Gal(X)=\{\T\in\class{} : \Forall{\T^{\prime}\in X}\T\sqcup\T^{\prime}\in\class{}\}.
\]

\begin{proposition}
$\Gal$ is antitone and for all $X, Y \subseteq \class{}$, we have
\[
    X \subseteq \Gal(Y) \Longleftrightarrow Y \subseteq \Gal(X).
\]
Hence, $(\Gal,\Gal)$ is an antitone Galois connection. 
\end{proposition}

\begin{proof}
    First, suppose $X\subseteq Y$, and let $\T\in\Gal(Y)$. Then, for each $\T'\in Y$, $\T\sqcup\T'$ is decidable. Since $X\subseteq Y$, $\T\sqcup\T'$ is decidable for each $\T'\in X$, and thus $\T\in\Gal(X)$. This proves that $\Gal(Y)\subseteq\Gal(X)$.

    Suppose now that $X\subseteq \Gal(Y)$ and $\T\in Y$.
    For each $\T'\in X$, we have that $\T\sqcup\T'$ is decidable, as $X\subseteq\Gal(Y)$. Thus, $\T\in\Gal(X)$, and so $Y\subseteq \Gal(X)$. This proves that $X \subseteq \Gal(Y) \Longrightarrow Y \subseteq \Gal(X)$, and the converse follows by symmetry.
\end{proof}

From the fact that $(G,G)$ is an antitone Galois connection, many nice properties immediately follow. The following is an instantiation of a general fact about Galois connections \cite[p.~496]{galois}.

\begin{proposition}
For all $X, Y \subseteq \class{}$, we have
\begin{itemize}
    \item $X \subseteq \Gal(\Gal(X))$;
    \item $X \subseteq Y \Longrightarrow \Gal(\Gal(X)) \subseteq \Gal(\Gal(Y))$; and
    \item $\Gal(\Gal(\Gal(X))) = \Gal(X)$.
\end{itemize}
In particular, $Cl \coloneqq \Gal \circ \Gal$ is a closure operator.
\end{proposition}

Let $\mathcal{C}$ be the set of closed elements of $\mathbb{P}(\class{})$ (equivalently, the image of $\Gal$). The Knaster--Tarski theorem \cite{knaster-fixed-point,tarski-fixed-point} implies the following.

\begin{proposition}
    $\mathcal{C}$ is a complete lattice with meets and joins given by
    \[
    \bigwedge_{i \in I} X_i = \bigcap_{i \in I} X_i \qquad \bigvee_{i \in I} X_i = Cl\left( \bigcup_{i \in I} X_i \right).
    \]
\end{proposition}

Note that $\Gal$ is an antitone involution on $\mathcal{C}$, which makes $\mathcal{C}$ isomorphic to its dual. In particular, for all $X, Y \in \mathcal{C}$, we have $\Gal(X \vee Y) = \Gal(X) \wedge \Gal(Y)$ and $\Gal(X \wedge Y) = \Gal(X) \vee \Gal(Y)$.

Concretely, if $X \subseteq \class{}$, then there is a combination theorem for the pair of sets $(X,\Gal(X))$:
every theory in $X$ is combinable with every theory from
$\Gal(X)$. If $X$ is closed, then this combination theorem is \emph{sharp}: we cannot enlarge $X$ to a larger set of theories that are all combinable with those in $\Gal(X)$.
This fact is used in \Cref{sec-known}.
Otherwise, it can be sharpened to a new combination theorem for $(Cl(X),\Gal(X))$, which applies to more pairs of theories. Another way to discover new combination theorems goes through the meet and join operators.
We provide examples for both techniques in \Cref{sec-new}.

\begin{remark}
    We defined our Galois connection over sets of decidable theories. One could define a similar Galois connection over all theories, since there are pairs of undecidable theories whose combination is decidable. We choose to focus on decidable theories, both because this is more relevant in practice and because this assumption is theoretically convenient.
\end{remark}

The results to be found throughout the rest of the paper may be summarized in the Hasse diagram found in \Cref{fig:lattice}:
we establish the sharpness of known combination methods (\Cref{sec-known}), find new ones (\Cref{sec-new}), and describe how they relate to one another (\Cref{sec-lattice}).

\begin{figure}[h!]
    \centering
    \adjustbox{scale=1.0,center}{
    \begin{tikzcd}[sep=small,column sep=0.5em]
& & & \mathfrak{T} (\ref{sec:shinydec}) \arrow[dash]{ddr}{} \arrow[dash]{dl}{} & & &\\%
& & \class{$n$-decidable} (\ref{sec-n-decidable}) \arrow[dash]{dl}{} & & & &\\
\phantom{O}\arrow[bend right=90,swap,red,<->]{dddd}{\Gal} & \class{CFS} (\ref{sec:gentlespec})\arrow[dash]{ddrr}{}\arrow[dash]{d}{} & & & \class{ID} (\ref{sec-sm-cs})\arrow[dash]{ddr}{}\arrow[dash]{ddl}{} & &\\
& \class{\coquagen} (\ref{sec-quasi-gentle})\arrow[dash]{dd}{} & & & & &\\
\phantom{O}\arrow[dash,dashed]{r}{}& \phantom{O}\arrow[dashed,dash]{rr}{} & & \class{CS} (\ref{sec-cs})\arrow[dash]{ddr}{}\arrow[dash]{ddll}{}\arrow[dashed,dash]{rr}{} & & \class{SI} (\ref{sec:stablelimit})\arrow[dash]{ddl}{}\arrow[dash,dashed]{r}{} &\phantom{O}\\
& \class{\quagen{}} (\ref{sec-quasi-gentle})\arrow[dash]{d}{} & & & & &\\
\phantom{O} & \class{gentle} (\ref{sec:gentlespec})\arrow[dash]{dr}{} & & & \class{SM+CS} (\ref{sec-sm-cs})\arrow[dash]{ddl}{} & &\\
& & \class{$n$-shiny} (\ref{sec-n-decidable})\arrow[dash]{dr}{} & & & &\\
& & & \class{shiny} (\ref{sec:shinydec}) & & &\\%
\end{tikzcd}}
    \caption{The lattice of theory combination properties we will examine:
    next to each property one finds the section which studies it; 
    the structure of the diagram is explained in \Cref{sec-lattice}, where a precise statement of its correctness is also given.
    The Galois connection reflects the lattice over the dashed line.}
    \label{fig:lattice}
\end{figure}
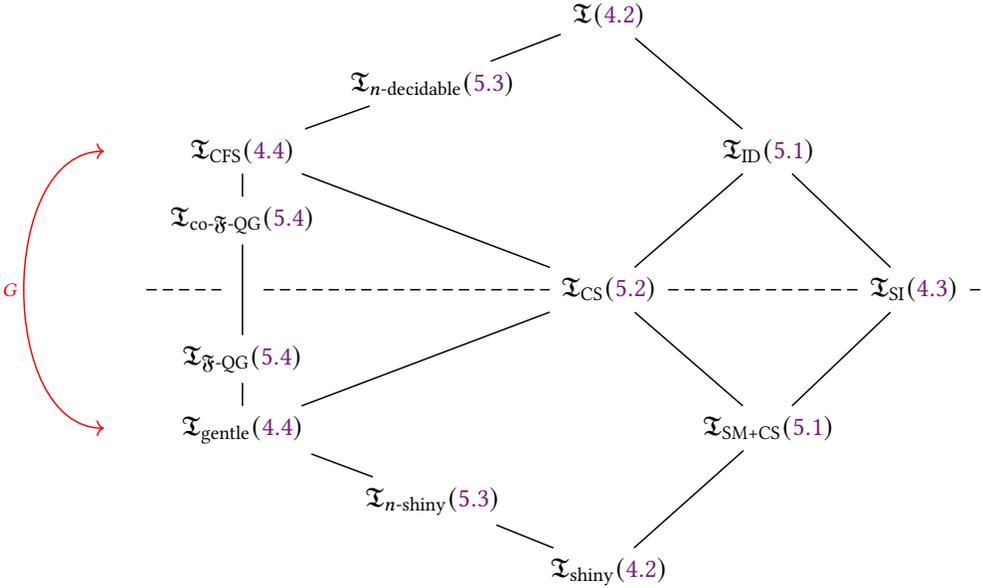

\section{The Sharpness of Known Combination Theorems} \label{sec-known}

In this section, we show that well-known theory combination methods cannot be improved, by investigating the requirements that they impose on theories.
For each requirement that arises from a combination theorem, we precisely identify the set of theories that can be combined with theories that admit these requirements.
This provides both positive and negative boundaries for all well-known combination methods.

Our technique can be sketched as follows:
for each set $X$ of theories, we identify a restricted set of theories
such that any other theory can be combined with all
theories from $X$ if and only if it can be combined with this set. We call these theories {\em test theories}.
We then prove properties of this restricted set of theories
in order to get the desired sharpness results.

Specifically: in \Cref{sec:sometests} we introduce 
test theories. 
Then, in \Cref{sec:shinydec}, we prove that the shiny combination theorem is sharp; that is, shiny theories are the only decidable theories that can be combined with all decidable theories. Next, in \Cref{sec:stablelimit}, we prove that the Nelson--Oppen combination theorem is sharp; that is, stably infinite theories are the only decidable theories that can be combined with all decidable stably infinite theories. Finally, in \Cref{sec:gentlespec}, we prove that the gentle combination theorem as improved by \cite{CADE30} is sharp; that is, gentle theories are the only decidable theories that can be combined with all decidable theories with computable finite spectra (and vice versa).

\subsection{Test Theories}
\label{sec:sometests}
Our sharpness proofs will make use of several specific theories that we define here. The high-level strategy of each of our proofs is the same: assuming that a theory $\T$ can be combined with every theory in some set $\class{X}$, we choose a particular $\T' \in \class{X}$; then, we argue that the assumption that $\T \sqcup \T'$ is decidable forces $\T$ to satisfy some desired property. We call such a theory $\T'$ a \emph{test theory}.
The core of all the proofs of this section is the careful design
of the test theories, so that they correctly encode the required properties of theories.
Notice that the axiomatizations of these test theories all impose some restriction over which cardinalities are possible for a model. This is because Fontaine's lemma (\Cref{lem-fontaine}) tells us that the decidability of a combination of theories is determined by a decidability condition related to the theories' spectra. For the same reason, several of the axiomatizations make use of undecidable sets.

We start by defining the signatures of our test theories
in \Cref{tab:signatures}.
$\Sigma_1$ is the empty signature.
$\Spn$ has infinitely many nullary predicates.
Finally, given a countable signature $\Sigma$, $\Sigmasigma$ is defined with a nullary predicate $P_{\varphi,n}$ for each quantifier-free formula $\varphi$ on $\Sigma$ and each $n\in\No$ (and no other predicates or functions).
Notice that, because $\Sigma$ is countable, so is $QF(\Sigma)$ and thus $\Sigmasigma$.

\begin{table}
\centering
\caption{Signatures}
\renewcommand{\arraystretch}{1.15}
\centering
\begin{tabular}{c|c|c}
Sig. & Functions & Predicates\\
\hline
$\Sigma_{1}$ & $\emptyset$ & $\emptyset$\\
$\Spn$ & $\emptyset$ & $\{P_{n} : n\in \No\}$\\
$\Sigmasigma$ & $\emptyset$ & $\{P_{\varphi,n} : \varphi\in QF(\Sigma), n\in\No\}$
\end{tabular}
\label{tab:signatures}
\centering
\end{table}

Let $F:\No\rightarrow\No\cup\{\Inf\}$ be a function such that:
$(i)$ $\{(m,n) \in {\No}^2 : F(m)\geq n\}$ is decidable; and $(ii)$ $\{n : F(n)=\aleph_{0}\}$ is undecidable.
Such functions do exist: 
one example takes an enumeration of Turing machines and returns, for the $n$th Turing machine, the number of steps it takes to halt (or $\aleph_0$ if it does not halt). Also, let $\unc \subset \No$ be an undecidable set.

$\Testg$ is a $\Spn$-theory 
parameterized by an infinite decidable set $S\subseteq\No$.
In it,
the validity of $P_{n}$ implies a model has at most $F(n)$ elements, whenever $F(n)$ is a finite quantity; 
we demand as well that the cardinality of its finite models must belong to a decidable set $S$. 
We sometimes work with $\Tle^{\No}$, the special case where $S = \No$. For example, this theory is used when testing whether a theory has the finite model property, as there is an infinite interpretation that satisfies $P_{n}$ if and only if $F(n)=\aleph_{0}$, an undecidable problem. 

$\Testcfs$ is also a $\Spn$-theory, 
where the validity of $P_{n}$ implies a model has exactly $n$ elements:
it is designed to test if a given $n\in\No$ is in the spectrum of a formula.
The similar $\Sigma_{1}$-theory $\Teqn$ has in turn only interpretations of cardinality $n$ and serves a similar purpose.

Given $m < n$, the theory $\Testsm$, defined over $\Spn$ as well, has only models of cardinalities $m$ and $n$;
and, when $P_{k}$ is satisfied, for an $k \in \unc$, then it has only models of cardinality $n$.
This allows one to test for smoothness, since a formula on a smooth theory with a model of size $m$ must have models of every size $n>m$.
Yet another $\Spn$-theory is $\Testsi$, where the validity of a $P_{m}$, for $m\in\unc$, implies the model has more than $n$ elements:
it tests for whether $\aleph_{0}$ is in the spectrum of a formula, since otherwise the formula has a maximal finite model whose cardinality we can use as our $n$.
The $\Sigma_{1}$-theory $\Tinfty$ contains only infinite interpretations, and it is used to test if a formula has infinite models.

\begin{table}
\centering
\caption{Test theories. 
$F:\No\rightarrow\No\cup\{\Inf\}$ is a function such that:
$(i)$ $\{(m,n) \in {\No}^2 : F(m)\geq n\}$ is decidable; and $(ii)$ $\{n : F(n)=\aleph_{0}\}$ is undecidable.
$\unc \subset \No$ is an undecidable set.
}
\renewcommand{\arraystretch}{1.5}
\centering
\begin{tabular}{c|c|c}
Sig. & Theory & Axiomatization\\
\hline
$\Spn$ &  $\T_{>n}^{P}$ & $\{P_{m}\rightarrow \psi_{\geq n+1} : m\in \unc\}\cup\{P_{i}\rightarrow\neg P_{j} : i\neq j\}$\\
$\Spn$ & $\T_{=}^{P}$ & $\{P_{n}\rightarrow\psi_{=n} : n\in\No\}$\\
$\Spn$ & $\T^{n}_{m}$ & $\{\psi_{=m}\vee\psi_{=n}\}\cup\{P_{k}\rightarrow\psi_{=n} : k\in \unc\}\cup\{P_{i}\rightarrow\neg P_{j} : i\neq j\}$\\
$\Spn$ & $\T_{\leq}^{S}$ & $\{P_{n}\rightarrow\psi_{\leq F(n)} : F(n)\in\mathbb{N}\}\cup\{\neg\psi_{=n} : n\notin S\}\cup\{P_{i}\rightarrow\neg P_{j} : i\neq j\}$\\
$\Sigmasigma$ & $Th_{\T}$ & \Cref{def:theorytheory}\\
$\Sigma_{1}$ & $\Tinfty$ & $\{\psi_{\geq n} : n\in\No\}$\\
$\Sigma_{1}$ & $\T_{=n}$ & $\{\psi_{=n}\}$
\end{tabular}
\label{tab:testtheories}
\centering
\end{table}

One more test theory, $Th_{\T}$, is defined below, separately, given its intricacy:
it is used in the study of gentleness to compute the complement of the spectrum of a formula when said spectrum is cofinite.
It takes a $\Sigma$-theory $\T$ and considers, for each of its quantifier-free formulas $\varphi$, the set $\No\setminus\spec_{\T}(\varphi)$, the finite cardinalities that a model of $\varphi$ does not have.
Let us enumerate this set: 
$s_{\varphi,1}$, $s_{\varphi,2}$ and so on;
if this set is finite, i.e. the spectrum of $\varphi$ is cofinite, then we make $s_{\varphi,n}=\aleph_{0}$ for all elements in the enumeration after the last actual one.
Then, if $P_{\varphi,n}$ is true in an interpretation $\A$ of $Th_{\T}$, $\dom{\A}$ has cardinality $s_{\varphi,n}$, which may be infinite;
that is, $|\dom{\A}|$ is either the $n$th number not in $\spec_{\T}(\varphi)$, or infinite if there is none.
This way, notice that the union of the spectra of $P_{\varphi,i}$, for $i\in\No$, is exactly $(\N\setminus\spec_{\T}(\varphi))\cup\{\aleph_{0}\}$.

\begin{definition}\label{def:theorytheory}
Let $\T$ be an arbitrary $\Sigma$-theory with computable finite spectra. Given a $\Sigma$-formula $\varphi$, let $S_\varphi = \No \setminus \spec_\T(\varphi)$. 
For each $i\in\No$, let 
$s_{\varphi,i}$ be its $i$th element in increasing order,
or $\Inf$ if $S_\varphi$ has fewer than $i$ elements. Now, let $Th_\T$ be the $\Sigmasigma$-theory axiomatized by
\begin{align*}
    &\{P_{\varphi,n} \rightarrow \psi_{=s_{\varphi,n}} : \varphi \in \qf{\Sigma},\; n \in \No,\; s_{\varphi,n} < \Inf\} \cup \\
    &\{P_{\varphi,n} \rightarrow \psi_{>m} : \varphi \in \qf{\Sigma},\; m,n \in \No,\; s_{\varphi,n} = \Inf\} \cup \\
    &\{P_{\varphi,n} \rightarrow \lnot P_{\varphi',n'} : (\varphi,n) \neq (\varphi',n')\}.
\end{align*}
\end{definition}

\begin{example}
    Consider, as a toy example, the $\Sigma_{Q}$-theory $\T$ axiomatized by $\{Q\rightarrow\psi_{=4}\}\cup\{\neg Q\rightarrow\psi_{\geq 3}\}$, where $\Sigma_{Q}$ is the signature with a single nullary predicate $Q$.
    This way, $\No\setminus\spec_{\T}(Q)=\No\setminus\{4\}$, and $\No\setminus\spec_{\T}(\neg Q)=\{1,2\}$;
    so $s_{Q,1}=1$, $s_{Q,2}=2$ and $s_{Q,3}=3$, but $s_{Q,4}=5$, $s_{Q,5}=6$ and so on;
    and $s_{\neg Q,1}=1$ and $s_{\neg Q,2}=2$, but $s_{\neg Q,3}=s_{\neg Q,4}=\cdots=\aleph_{0}$.
    So any model of $Th_{\T}$ satisfying $P_{Q,2}$ has cardinality $2$, and any satisfying $P_{Q,4}$ has cardinality $5$ (and such models do exist in both cases).
    Any model of $Th_{\T}$ satisfying $P_{\neg Q,2}$ has cardinality $2$, but any models satisfying $P_{\neg Q,4}$ must have infinitely many elements, and again such models exist in both cases.
\end{example}

The following theorem states some properties that our test theories satisfy. We will make essential use of these in our proofs.

\begin{theorem}\label{theo:testtheories}
    If a check mark appears where the column of a test theory $\T$ intersects the row of a set $\class{X}$ in \Cref{tab:testtheoprop}, then $\T \in \class{X}$.
\end{theorem}

We omit $\class{CS}$ from \Cref{tab:testtheoprop}, since $\class{CS} = \class{CFS} \cap \class{ID}$.

In all the proofs that follow, whenever we claim that a theory from \Cref{tab:testtheories} has a property, it is justified by \Cref{theo:testtheories}.

\subsection{Combination of Shiny Theories With Decidable Theories}
\label{sec:shinydec}

Tinelli and Zarba~\cite{shiny} proved that shiny theories are combinable with every decidable theory. In terms of our Galois connection, this means $\class{}\subseteq\Gal(\class{shiny})$ and $\class{shiny}\subseteq\Gal({\class{}})$; since $\Gal(\class{shiny})\subseteq \class{}$ by definition, we also get $\Gal(\class{shiny})=\class{}$. We prove that shiny theories are the \emph{only} decidable theories that can be combined with every decidable theory; that is, $\Gal(\mathfrak{T}) \subseteq \class{shiny}$, and therefore $\Gal(\mathfrak{T}) = \class{shiny}$ in virtue of the shiny combination theorem. In particular, $\class{shiny}$ is closed with respect to our Galois connection.

\begin{theorem} \label{thm-shiny-tight}
    If $\T$ is a decidable theory that can be combined with every decidable theory, then $\T$ is shiny, i.e. $\Gal(\class{})\subseteq\class{shiny}$.
\end{theorem}
\begin{proof}
    Suppose $\T$ can be combined with every decidable theory.

    First, we show that $\T$ has the finite model property. Suppose not, and let $\varphi$ be a formula such that $\spec_\T(\varphi) = \{\Inf\}$. Since $\Tle^{\No}$ is decidable\reference, $\T \sqcup \Tle^{\No}$ is decidable. But $\varphi \land P_n$ is $\T \sqcup \Tle^{\No}$-satisfiable if and only if $F(n) = \Inf$, which is undecidable by our assumption on $F$. This contradicts the decidability of $\T \sqcup \Tle^{\No}$.

    Second, we show that $\T$ has a computable minimal model function. Given a formula $\varphi$, we know that it has some finite $\T$-interpretation by the previous paragraph. Since $\Testcfs$ is decidable\reference, $\T \sqcup \Testcfs$ is decidable. Hence, we can compute the size of the minimal $\T$-interpretation of $\varphi$ by searching for the smallest $n \in \No$ such that $\varphi \land P_n$ is $\T \sqcup \Testcfs$-satisfiable.

    Finally, it remains to show that $\T$ is smooth. If not, then for some formula $\varphi$, we have $m \in \spec_\T(\varphi)$ and $n \notin \spec_\T(\varphi)$ for some $m,n \in \N$ with $m < n$. In fact, by compactness, we can take $m,n \in \No$. Since $\Testsm$ is decidable\reference, $\T \sqcup \Testsm$ is decidable. But $\varphi \land P_{k}$ is $\T \sqcup \Testsm$-satisfiable if and only if $k \notin \unc$, which is undecidable since $\unc$ is undecidable. This contradicts our assumption that $\T$ can be combined with every decidable theory.
\end{proof}

The proof shows that a decidable theory $\T$ is shiny if and only if it can be combined with the test theories $\{\Tle^{\No},\Testcfs\} \cup \{\Testsm : m,n \in \No,\; m<n\}$; that is, $\class{shiny} = \Gal(\{\Tle^{\No},\Testcfs\} \cup \{\Testsm : m,n \in \No,\; m<n\})$. We can improve this characterization to only use a \emph{single} test theory.
\begin{restatable}{theorem}{thmshinyuniform} \label{thm-shiny-uniform}
    There is a decidable theory $\T'$ such that, for every decidable theory $\T$, if $\T$ can be combined with $\T'$ then $\T$ is shiny (and therefore combinable with every decidable theory).
\end{restatable}

The theory $\T'$ is maximally difficult from the standpoint of theory combination, in the sense that combinability with $\T'$ implies combinability with every other decidable theory. This is analogous to the phenomenon of complete problems in computability theory and complexity theory (i.e., completeness in the sense of $\mathsf{NP}$-completeness). Thus, we say that $\T'$ is \emph{$\class{}$-complete} with the following meaning.

\begin{definition} 
A theory $\T$ is said to be \emph{$\class{X}$-complete} when $\class{X} = Cl(\{\T\})$.
\end{definition}

When a $\class{X}$-complete theory $\T$ exists, we can think of $\class{X}$ as being ``generated'' by the singleton $\{\T\}$. Perhaps surprisingly, most, but not all, of the combination properties we study turn out to have a corresponding $\class{X}$-complete theory. 

Note that the results of this section also apply to strongly polite theories, because Casal and Rasga~\cite{CasalRasga2} showed that $\class{shiny} = \class{SP}$, where $\class{SP}$ is the set of decidable strongly polite theories. Incidentally, given the combination theorem for strongly polite theories~\cite{JB10-LPAR}, \Cref{thm-shiny-tight} gives a new proof that $\class{SP} \subseteq \class{shiny}$.

\begin{table}
\centering
\caption{Test theories' membership}
\renewcommand{\arraystretch}{1.15}
\setlength{\tabcolsep}{10pt}
\centering
\begin{tabular}{c|ccccccccc}
 & $\Testsi$ & $\Testcfs$ & $\Testsm$ & $\Testg$ &  $\Tinfty$ & $Th_{\T}$ & $\Teqn$\\
\hline
$\class{}$ & \checkmark & \checkmark & \checkmark & \checkmark & \checkmark & \checkmark & \checkmark \\
$\class{CFS}$ & & \checkmark & & \checkmark & \checkmark & \checkmark & \\
$\class{ID}$ & & \checkmark & \checkmark & & \checkmark & & \\
$\class{SI}$ & \checkmark & & & & & & \\
$\class{gentle}$ & & \checkmark & & & & & \\
$\class{SM+CS}$ & & & & & \checkmark & & \\
$\class{polite}$ & \checkmark & & & & & &
\end{tabular}
\label{tab:testtheoprop}
\centering
\end{table}

\subsection{Combination of Stably Infinite Theories}
\label{sec:stablelimit}
Oppen~\cite{OppenSI} proved that any two decidable stably infinite theories are combinable; that is, $\class{SI} \subseteq \Gal(\class{SI})$. We prove that $\Gal(\class{SI}) = \class{SI}$, which implies that $\class{SI}$ is closed.
If a theory is not stably infinite, then it cannot be combined with every decidable stably infinite theory.

\begin{theorem} \label{thm-f-si}
    If $\T$ is a decidable theory that can be combined with every decidable stably infinite theory, then $\T$ is stably infinite, i.e. $\Gal(\class{SI})\subseteq\class{SI}$.
\end{theorem}
\begin{proof}
    Suppose $\T$ can be combined with every decidable stably infinite theory. 
    If $\T$ is not stably infinite, then there is a $\T$-satisfiable formula $\varphi$ such that $\Inf \notin \spec_\T(\varphi)$. 
    By compactness, there is a maximal $n \in \spec_\T(\varphi)$. 
    Then, the formula $\varphi \land P_{k}$ is $\T \sqcup \Testsi$-satisfiable if and only if $k \notin \unc$, which is undecidable: for the converse, notice that if $k \notin \unc$, then the axioms of $\Testsi$ are trivially satisfied by a model in which $P_{k}$ is true and $P_{k'}$ is false for $k'\neq k$. 
    But $\Testsi$ is decidable and stably infinite\reference, so $\T \sqcup \Testsi$ is decidable by our assumption on $\T$, a contradiction.
\end{proof}

The proof shows that a decidable theory $\T$ is stably infinite if and only if it can be combined with the test theories $\{\Testsi : n \in \No\}$; that is $\class{SI} = \Gal(\{\Testsi : n \in \No\})$. We can again improve this characterization to only use a single test theory, exhibiting a $\class{SI}$-complete theory.
\begin{restatable}{theorem}{thmsicomplete} \label{thm-si-complete}
    There is a decidable stably infinite theory $\T'$ such that, for every decidable theory $\T$, if $\T$ can be combined with $\T'$ then $\T$ is stably infinite.
\end{restatable}

Since $\Testsi$ is polite, the proof of \Cref{thm-f-si} shows that $\Gal(\class{polite}) = \class{SI}$, where $\class{polite}$ is the set of decidable polite theories. Hence, $Cl(\class{polite}) = \class{SI}$. Intuitively, this means that, for the purpose of theory combination, politeness is no more valuable than stable infiniteness.

\subsection{Combination of Gentle Theories With Theories That Have Computable Finite Spectra}
\label{sec:gentlespec}

Fontaine~\cite{gentle} proved that gentle theories are combinable with theories satisfying a disjunction of three properties. 
This was recently improved in \cite{CADE30}, where we proved that gentle theories are combinable with every decidable theory having computable finite spectra; again in terms of our Galois connection, that means $\class{CFS} \subseteq \Gal(\class{gentle})$ and $\class{gentle}\subseteq\Gal(\class{CFS})$. 
We prove that $\Gal(\class{gentle}) = \class{CFS}$ and that $\Gal(\class{CFS}) = \class{gentle}$, which together imply that $\class{CFS}$ and $\class{gentle}$ are closed.
If a theory is not gentle, then it cannot be combined with every decidable theory with computable finite spectra (and vice versa).

First, we prove that $\Gal(\class{gentle}) = \class{CFS}$ with the following theorem.
\begin{theorem} \label{thm-f-gentle}
    If $\T$ is a decidable theory that can be combined with every gentle theory, then $\T$ has computable finite spectra, i.e. $\Gal(\class{gentle})\subseteq \class{CFS}$.
\end{theorem}
\begin{proof}
    Suppose $\T$ can be combined with every gentle theory. Let $\varphi$ be a quantifier-free formula and $k \in \No$. Then, we have $k \in \spec_\T(\varphi)$ if and only if $\varphi \land P_k$ is $\T \sqcup \Testcfs$-satisfiable. Since $\Testcfs$ is gentle\reference, $\T \sqcup \Testcfs$ is decidable, so we have an algorithm for checking whether $k \in \spec_\T(\varphi)$, as desired.
\end{proof}

The proof shows that a decidable theory $\T$ has computable finite spectra if and only if it can be combined with the test theory $\Testcfs$. Thus, $\Testcfs$ is $\class{gentle}$-complete.

Next, we prove that $\Gal(\class{CFS}) = \class{gentle}$.
\begin{theorem} \label{thm-f-cfs}
    If $\T$ is a decidable theory that can be combined with every decidable theory with computable finite spectra, then $\T$ is gentle, i.e. $\Gal(\class{CFS})\subseteq\class{gentle}$.
\end{theorem}
\begin{proof}
    Suppose $\T$ can be combined with every decidable theory with computable finite spectra. Let $\varphi$ be a quantifier-free formula.

    First, we show that $\T$ has computable spectra. Since $\Testcfs$ is decidable and has computable finite spectra\reference, $\T \sqcup \Testcfs$ is decidable. Hence, given $k \in \No$, we can compute whether $k \in \spec_\T(\varphi)$ by checking if $\varphi \land P_k$ is $\T \sqcup \Testcfs$-satisfiable. Thus, $\T$ has computable finite spectra. Further, since $\T_{\infty}$ is decidable and has computable finite spectra\reference, $\T \sqcup \T_{\infty}$ is decidable. Thus, we can compute whether $\Inf \in \spec_\T(\varphi)$ by checking if $\varphi$ is $\T \sqcup \T_{\infty}$-satisfiable.

    Suppose that $\Inf \notin \spec_\T(\varphi)$. Then, by compactness, $\spec_\T(\varphi)$ is a finite subset of $\No$. It remains to show how to compute an explicit representation of this finite set. Since $\T$ has computable finite spectra, it suffices to compute an upper bound on $\spec_\T(\varphi)$. We can do this by searching for the smallest $n \in \No$ such that $\varphi \land \neq(x_1,\dots,x_n)$ is $\T$-unsatisfiable (where the variables $x_i$ are fresh).

    The other case is when $\Inf \in \spec_\T(\varphi)$. We first need to show that $\spec_\T(\varphi)$ is cofinite. Suppose not. Then, $S \coloneqq \No \setminus \spec_\T(\varphi)$ is infinite. Since $\T$ has computable finite spectra, $S$ is a decidable set. Then, $\Testg$ is decidable and has computable finite spectra\reference, so $\T \sqcup \Testg$ is decidable. But $\varphi \land P_n$ is $\T \sqcup \Testg$-satisfiable if and only if $F(n) = \Inf$, which is undecidable, a contradiction.

    We have shown that $\spec_\T(\varphi)$ is cofinite. It remains to actually compute the spectrum, by which we mean compute an explicit representation of its complement. Since $\T$ has computable finite spectra, it suffices to show that, given $\varphi$, we can compute the cardinality of $S \coloneqq \No \setminus \spec_\T(\varphi)$. Since $Th_{\T}$ is decidable and has computable finite spectra\reference, $\T \sqcup Th_{\T}$ is decidable. Now, $\varphi \land P_{\varphi,n}$ is $\T \sqcup Th_{\T}$-satisfiable if and only if $n > |S|$. Thus, by testing formulas of this form for satisfiability, we can compute the cardinality of $S$.
\end{proof}

In the above proof, some of the test theories with which we combine $\T$ (namely, $\Testg$ and $Th_{\T}$) depend on our choice of $\T$. This is no accident. Indeed, there is no $\class{CFS}$-complete theory:
\begin{restatable}{theorem}{thmgentlenotuniform} \label{thm-gentle-not-uniform}
    There is no decidable theory $\T'$ with computable finite spectra such that for every decidable theory $\T$, if $\T$ can be combined with $\T'$, then $\T$ is gentle.
\end{restatable}

\section{New Combination Theorems} \label{sec-new}

While \Cref{sec-known} studied the sharpness of existing combination theorems from the literature, in this section
we propose new combination theorems, allowing for new
ways to perform theory combination.
The new theorems are found by utilizing our Galois connection $\Gal$, its closure operator $Cl$, and its induced lattice of theory combination properties $\mathcal{C}$, as defined in \Cref{sec-galois}.
In \Cref{sec-sm-cs}, we provide a combination theorem that is based on smooth theories with computable spectra and show
that such theories can be combined with all infinitely 
decidable theories.
Then, in \Cref{sec-cs}, we prove a new symmetric combination theorem, based
on computable spectra, namely that every two decidable theories
with computable spectra can be combined.
For each of these combination methods, we also prove sharpness, in the spirit of \Cref{sec-known}.
In \Cref{sec-n-decidable} we take this technique a step further, and provide, at once, countably many distinct 
and sharp combination theorems, based on a
generalization of shiny theories.
Finally, in \Cref{sec-quasi-gentle}, we present a larger set of sharp combination
theorems, which is uncountable.
These combination theorems are based on a variant of gentleness.
All these new results are, in one way or another, motivated by the Galois connection $G$:
\Cref{sec-sm-cs,sec-n-decidable} rely on the closure operator $Cl$ being applied to some set of theories which, by itself, is not closed;
\Cref{sec-cs} on the algebraic structure of the lattice $\mathcal{C}$ induced by $G$, specifically taking meets and joints of known sets of theories;
and \Cref{sec-quasi-gentle} is motivated by a desire to demonstrate the structural richness of $\mathcal{C}$ by exhibiting uncountably many combination theorems.

\subsection{Combination of SM+CS Theories With Infinitely Decidable Theories} \label{sec-sm-cs}

In \cite{CADE30}, we proved that decidable smooth theories with a computable minimal model function (call this set of theories $\class{SM+MM}$) are combinable with every decidable theory that is infinitely decidable. 
It turns out that this combination theorem is not sharp, something our techniques helped us discover. Indeed, $\class{SM+MM}$ is not closed with respect to the Galois connection (although $\class{ID}$ turns out to be), which means that there is an improved combination theorem between $\class{ID}$ and $G(\class{ID}) = Cl(\class{SM+MM})$. We show that $Cl(\class{SM+MM}) = \class{SM+CS}$. To that end, we first prove a new combination theorem improving the one from \cite{CADE30} and then apply our test theories to show that the combination theorem is sharp.

\begin{theorem} \label{thm-sm-cs}
    Let $\T_1$ and $\T_2$ be decidable theories. Suppose that $\T_1$ is smooth and has computable spectra and that $\T_2$ is infinitely decidable. Then, $\T_1 \sqcup \T_2$ is decidable, meaning $\class{ID}\subseteq\Gal(\class{SM+CS})$ and $\class{SM+CS}\subseteq\Gal(\class{ID})$.
\end{theorem}

In light of Fontaine's lemma, the following implies \Cref{thm-sm-cs}:
\begin{lemma}
    Let $\T_1$ and $\T_2$ be decidable theories over signatures $\Sigma_1$ and $\Sigma_2$ respectively. Suppose that $\T_1$ is smooth and has computable spectra and that $\T_2$ is infinitely decidable. Then, it is decidable whether $\spec_{\T_1}(\varphi_1) \cap \spec_{\T_2}(\varphi_2) = \emptyset$, where $\varphi_1$ and $\varphi_2$ are conjunctions of literals over the signatures $\Sigma_1$ and $\Sigma_2$ respectively.
\end{lemma}
\begin{proof}
    Let $\varphi_1$ and $\varphi_2$ be conjunctions of literals over the signatures $\Sigma_1$ and $\Sigma_2$ respectively. We describe our decision procedure as follows. If $\varphi_1$ is $\T_1$-unsatisfiable, then $\spec_{\T_1}(\varphi_1) \cap \spec_{\T_2}(\varphi_2) = \emptyset$, so suppose $\varphi_1$ is $\T_1$-satisfiable.

    If $\Inf \in \spec_{\T_2}(\varphi_2)$ (which we can check since $\T_2$ is infinitely decidable), then $\Inf \in \spec_{\T_1}(\varphi_1) \cap \spec_{\T_2}(\varphi_2)$, since $\T_1$ is smooth. Otherwise, by compactness, there is a maximal $k \in \spec_{\T_2}(\varphi_2)$. We can compute this $k$ by checking the $\T_2$-satisfiability of $\varphi_2 \land \neq(x_1,\dots,x_n)$ (where the variables $x_i$ are fresh) for increasing values of $n$. Then, since $\T_1$ is smooth, we have $\spec_{\T_1}(\varphi_1) \cap \spec_{\T_2}(\varphi_2) = \emptyset$ if and only if $k \notin \spec_{\T_1}(\varphi_1)$, which we can check since $\T_1$ has computable spectra.
\end{proof}

Since every decidable smooth theory with a computable minimal model function has computable spectra (each nonempty spectrum is of the form $\{m \in \N : m \ge k\}$, where $k \in \N$ can be computed), \Cref{thm-sm-cs} is indeed a strengthening of the result from \cite{CADE30} mentioned at the start of this subsection. To show that it is a \emph{strict} improvement, the following example shows that there is a smooth theory with computable spectra whose minimal model function is not computable.
\begin{example} \label{ex-cm-cs}
    Let $\Tgr$ be the theory axiomatized by $\{P_{n}\rightarrow \psi_{\geq F(n)} : n\in\No,\; F(n)\in\No\} \cup \{P_n \rightarrow \psi_{\geq m} : m,n \in \No,\; F(n) = \aleph_0\}$, where $F$ is the function defined in \Cref{sec:sometests}. 
    Then, $\Tgr$ is smooth and has computable spectra, but its minimal model function is not computable.
\end{example}

Now that we've proved that $\class{ID} \subseteq \Gal(\class{SM+CS})$, we prove the converse.
\begin{theorem} \label{thm-f-sm-cs}
    If $\T$ is a decidable theory that can be combined with every smooth theory with computable spectra, then $\T$ is infinitely decidable, i.e. $\Gal(\class{SM+CS})\subseteq \class{ID}$.
\end{theorem}
\begin{proof}
    Suppose $\T$ can be combined with every smooth theory with computable spectra. Since $\T_\infty$ is a smooth theory with computable spectra (\Cref{theo:testtheories}), $\T \sqcup \T_\infty$ is decidable. But $\varphi$ is $\T \sqcup \T_\infty$-satisfiable if and only if $\Inf \in \spec_\T(\varphi)$, so $\T$ is infinitely decidable.
\end{proof}

The proof shows that a decidable theory $\T$ is infinitely decidable if and only if it can be combined with the test theory $\T_\infty$. Thus, $\T_\infty$ is $\class{SM+CS}$-complete.

\begin{theorem} \label{thm-f-id}
    If $\T$ is a decidable theory that can be combined with every decidable theory that is infinitely decidable, then $\T$ is smooth and has computable spectra, i.e. $\Gal(\class{ID})\subseteq \class{SM+CS}$.
    
\end{theorem}
\begin{proof}
    Suppose $\T$ can be combined with every decidable theory that is infinitely decidable.

    First, we show that $\T$ is smooth. If not, then for some formula $\varphi$, we have $m \in \spec_\T(\varphi)$ and $n \notin \spec_\T(\varphi)$ for some $m,n \in \N$ with $m < n$. In fact, by compactness, we can take $m,n \in \No$. Then, $\varphi \land P_{k}$ is $\T \sqcup \Testsm$-satisfiable if and only if $k \notin \unc$, which is undecidable since $\unc$ is undecidable. This contradicts our assumption that $\T$ can be combined with every decidable theory that is infinitely decidable, given that $\Testsm$ is infinitely decidable (\Cref{theo:testtheories}).

    Next, we show that $\T$ has computable spectra. Since $\T$ is decidable and smooth, it is infinitely decidable, so it remains to show that $\T$ has computable finite spectra. Let $\varphi$ be a quantifier-free formula and $k \in \No$. Then, we have $k \in \spec_\T(\varphi)$ if and only if $\varphi \land P_k$ is $\T \sqcup \Testcfs$-satisfiable. Since $\Testcfs$ is decidable and infinitely decidable (\Cref{theo:testtheories}), $\T \sqcup \Testcfs$ is decidable, so we have an algorithm for checking whether $k \in \spec_\T(\varphi)$, as desired.
\end{proof}

The proof shows that a decidable theory $\T$ is smooth and has computable spectra if and only if it can be combined with the test theories $\{\Testcfs\} \cup \{\Testsm : m,n \in \No,\; m<n\}$; that is, $\class{SM+CS} = \Gal(\{\Testcfs\} \cup \{\Testsm : m,n \in \No,\; m<n\})$. We can improve this characterization to only use a single test theory, exhibiting a $\class{ID}$-complete theory.
\begin{restatable}{theorem}{thmidcomplete} \label{thm-id-complete}
    There is a decidable theory $\T'$ that is infinitely decidable and such that for every decidable theory $\T$, if $\T$ can be combined with $\T'$, then $\T$ is smooth and has computable spectra.
\end{restatable}

We have shown that $\class{SM+CS}$ and $\class{ID}$ are closed. Since $\T_{\infty}$ is smooth and has a computable minimal model function  by \Cref{theo:testtheories}, the proof of \Cref{thm-f-sm-cs} shows that $\Gal(\class{SM+MM}) = \class{ID}$. Hence, $Cl(\class{SM+MM}) = \class{SM+CS}$ (but also $Cl(\class{ID})=\class{ID}$), giving a precise sense in which \Cref{thm-sm-cs} is the sharp version of the weaker combination theorem from \cite{CADE30}.

\subsection{Combination of Theories That Have Computable Spectra} \label{sec-cs}

We prove a new combination theorem stating that any two decidable theories with computable spectra are combinable; that is $\class{CS} \subseteq \Gal(\class{CS})$. Then, we prove that, in fact, $\Gal(\class{CS}) = \class{CS}$. Thus, in addition to the Nelson--Oppen theorem, we have another symmetric combination theorem (i.e., a combination theorem imposing the same requirement on both component theories).

We start by giving some motivation to show how one could have discovered this combination theorem using our abstract framework. The idea is that by taking meets and joins of closed sets of theories, we can discover new combination theorems. By definition, we have $\class{CS} = \class{CFS} \wedge \class{ID}$, so we know that $\class{CS}$ is closed. Thus, purely by abstract considerations, we know that there is a new combination theorem involving $\class{CS}$. We have
\[
    \Gal(\class{CS}) = \Gal(\class{CFS} \wedge \class{ID}) = \Gal(\class{CFS}) \vee \Gal(\class{ID}) = \class{gentle} \vee \class{SM+CS}.
\]
Determining the join of two sets of theories is not totally straightforward, but a modicum of guesswork can lead us to the right answer. We are looking for a closed set of theories containing both $\class{gentle}$ and $\class{SM+CS}$. Since $\class{CS}$ satisfies those criteria, it is only natural to conjecture that $\class{gentle} \vee \class{SM+CS} = \class{CS}$. And indeed, this conjecture is correct, as we now proceed to show.

\begin{theorem} \label{thm-cs}
    Let $\T_1$ and $\T_2$ be decidable theories with computable spectra. Then, $\T_1 \sqcup \T_2$ is decidable, meaning $\class{CS} \subseteq \Gal(\class{CS})$.
\end{theorem}

In light of Fontaine's lemma, the following implies \Cref{thm-cs}:
\begin{lemma}
    Let $\T_1$ and $\T_2$ be decidable theories over signatures $\Sigma_1$ and $\Sigma_2$ respectively. Suppose that $\T_1$ and $\T_2$ have computable spectra. Then, it is decidable whether $\spec_{\T_1}(\varphi_1) \cap \spec_{\T_2}(\varphi_2) = \emptyset$, where $\varphi_1$ and $\varphi_2$ are conjunctions of literals over the signatures $\Sigma_1$ and $\Sigma_2$ respectively.
\end{lemma}
\begin{proof}
    Let $\varphi_1$ and $\varphi_2$ be conjunctions of literals over the signatures $\Sigma_1$ and $\Sigma_2$ respectively. We describe our decision procedure as follows. First, check if $\Inf \in \spec_{\T_1}(\varphi_1) \cap \spec_{\T_2}(\varphi_2)$ (which we can do since both theories are infinitely decidable). If so, then $\spec_{\T_1}(\varphi_1) \cap \spec_{\T_2}(\varphi_2) \neq \emptyset$. Otherwise, without loss of generality, $\Inf \notin \spec_{\T_1}(\varphi_1)$. If $\varphi_1$ is $\T_1$-unsatisfiable, then $\spec_{\T_1}(\varphi_1) \cap \spec_{\T_2}(\varphi_2) = \emptyset$; otherwise, by compactness, there is a maximal $k \in \spec_{\T_1}(\varphi_1)$. We can compute this $k$ by checking the $\T_1$-satisfiability of $\varphi_1 \land \neq(x_1,\dots,x_n)$ (where the variables $x_i$ are fresh) for increasing values of $n$. Then, since $\T_1$ has computable spectra, we can compute an explicit representation of the finite set $\spec_{\T_1}(\varphi_1)$. And since $\T_2$ has computable spectra, we can check whether $m \in \spec_{\T_2}(\varphi_2)$ for any of the finitely many elements $m \in \spec_{\T_1}(\varphi_1)$.
\end{proof}

Now that we've proved that $\class{CS} \subseteq \Gal(\class{CS})$, we prove that, in fact, $\Gal(\class{CS}) = \class{CS}$.
\begin{theorem} \label{thm-f-cs}
If $\T$ is a decidable theory that can be combined with every decidable theory with computable spectra, then $\T$ has computable spectra, i.e. $\Gal(\class{CS})\subseteq \class{CS}$, and so $Cl(\class{CS})=\class{CS}$.
\end{theorem}
\begin{proof}
    Suppose $\T$ can be combined with every decidable theory with computable spectra. Since $\Testcfs$ is decidable and has computable spectra by \Cref{theo:testtheories}, $\T \sqcup \Testcfs$ is decidable. Hence, given $k \in \No$, we can compute whether $k \in \spec_\T(\varphi)$ by checking if $\varphi \land P_k$ is $\T \sqcup \Testcfs$-satisfiable. Thus, $\T$ has computable finite spectra. Further, since $\T_{\infty}$ is decidable and has computable spectra  by \Cref{theo:testtheories}, $\T \sqcup \T_{\infty}$ is decidable. Thus, we can compute whether $\Inf \in \spec_\T(\varphi)$ by checking if $\varphi$ is $\T \sqcup \T_{\infty}$-satisfiable.
\end{proof}

The proof shows that a decidable theory $\T$ has computable spectra if and only if it can be combined with the test theories $\Testcfs$ and $\T_{\infty}$. We can improve this characterization to only use a single test theory, exhibiting a $\class{CS}$-complete theory.
\begin{restatable}{theorem}{thmcscomplete} \label{thm-cs-complete}
    There is a decidable theory $\T'$ with computable spectra such that for every decidable theory $\T$, if $\T$ can be combined with $\T'$, then $\T$ has computable spectra.
\end{restatable}

In \cite{Bonacina}, the authors proved a similar combination theorem, which states that decidable theories $\T_1$ and $\T_2$ can be combined provided that, for each $i \in \{1,2\}$,
\begin{enumerate}
    \item $\T_i$ is over a finite signature $\Sigma_i$;
    \item $\T_i$ is infinitely decidable; and
    \item it is decidable whether a finite $\Sigma_i$-interpretation is a model of $\T_i$.
\end{enumerate}
These requirements imply that $\T_i$ has computable spectra. Indeed, we have infinite decidability by assumption, and we have computable finite spectra because, given a quantifier-free $\varphi$ and $k \in \No$, we can check if $k \in \spec_{\T_i}(\varphi)$ by enumerating the finitely many $\Sigma_i$-interpretations of size $k$ and checking whether some interpretation is a model of both $\T_i$ and $\varphi$. Thus, the computable spectra combination theorem is strictly more general than the one from \cite{Bonacina}.

\subsection{Combination of $n$-Shiny Theories With $n$-Decidable Theories} \label{sec-n-decidable}

In \cite[Theorem~7]{manna-zarba-survey}, Manna and Zarba proved that, for every fixed $n \in \No$, any two decidable theories all of whose models have size $n$ are combinable. This combination theorem is not sharp; for example, it is subsumed by the gentle combination theorem and the computable spectra combination theorem. Nevertheless, using our abstract framework, we can use this combination theorem as a means to discover a new, sharp combination theorem.

Let $\mathfrak{T}_n$ be the set of decidable theories all of whose models have size $n \in \No$. The idea is that we can discover a new combination theorem by determining $Cl(\mathfrak{T}_n)$. By Fontaine's lemma, $\Gal(\mathfrak{T}_n)$ is the set of decidable theories $\T$ such that there is an algorithm that takes a quantifier-free formula $\varphi$ and returns whether $n \in \spec_\T(\varphi)$; call such theories \emph{$n$-decidable} and denote the set of decidable theories that are $n$-decidable by $\class{$n$-decidable}$ (this property is a natural variant of infinite decidability). Thus, we know that there is a new combination theorem involving $\class{$n$-decidable}$. It remains to determine $Cl(\mathfrak{T}_n) = \Gal(\class{$n$-decidable})$ (notice that $\class{$n$-decidable}$ is closed, since $\class{$n$-decidable}=\Gal(\mathfrak{T}_{n})$). 
At this point, we consult our table of test theories, looking for ones that are $n$-decidable, and we use these test theories to deduce properties of $\Gal(\class{$n$-decidable})$; it may be helpful to take inspiration from the proof of \Cref{thm-shiny-tight}. It turns out that $\Gal(\class{$n$-decidable})$ is the set of decidable theories $\T$ such that there is an algorithm that takes a $\T$-satisfiable quantifier-free formula $\varphi$ and returns a pair $(t,k)$, where $t \in \{0,1,2\}$ and $k \in \No$, such that:
\begin{itemize}
    \item if $t = 0$, then $\spec_\T(\varphi) = \{n\}$;
    \item if $t = 1$, then $\spec_\T(\varphi) = \{n\} \cup \{m \in \N : m \ge k\}$; and
    \item if $t = 2$, then $\spec_\T(\varphi) = \{m \in \N : m \ge k\}$.
\end{itemize}
Call such theories \emph{$n$-shiny} and denote the set of decidable $n$-shiny theories by $\class{$n$-shiny}$.

\begin{theorem} \label{thm-n-decidable}
    Let $\T_1$ and $\T_2$ be decidable theories. Suppose that $\T_1$ is $n$-decidable and $\T_2$ is $n$-shiny. Then $\T_1 \sqcup \T_2$ is decidable, meaning  $\class{$n$-shiny}\subseteq\Gal(\class{$n$-decidable})$ and $\class{$n$-decidable}\subseteq\Gal(\class{$n$-shiny})$.
\end{theorem}

In light of Fontaine's lemma, the following implies \Cref{thm-n-decidable}:
\begin{lemma}
    Let $\T_1$ and $\T_2$ be decidable theories over signatures $\Sigma_1$ and $\Sigma_2$ respectively. Suppose that $\T_1$ is $n$-decidable and that $\T_2$ is $n$-shiny for some $n \in \No$. Then, it is decidable whether $\spec_{\T_1}(\varphi_1) \cap \spec_{\T_2}(\varphi_2) = \emptyset$, where $\varphi_1$ and $\varphi_2$ are conjunctions of literals over the signatures $\Sigma_1$ and $\Sigma_2$ respectively.
\end{lemma}
\begin{proof}
    Let $\varphi_1$ and $\varphi_2$ be conjunctions of literals over the signatures $\Sigma_1$ and $\Sigma_2$ respectively. We describe our decision procedure as follows. If $\varphi_2$ is $\T_2$-unsatisfiable, then $\spec_{\T_1}(\varphi_1) \cap \spec_{\T_2}(\varphi_2) = \emptyset$, so assume $\varphi_2$ is $\T_2$-satisfiable.

    We have $\spec_{\T_2}(\varphi_2) = \{n\}$, $\spec_{\T_2}(\varphi_2) = \{n\} \cup \{m \in \N : m \ge k\}$, or $\spec_{\T_2}(\varphi_2) = \{m \in \N : m \ge k\}$ for some $k \in \No$ (and we can compute $k$ and which case holds). If the first or second case holds, first check if $n \in \spec_{\T_1}(\varphi_1)$ (which we can do since $\T_1$ is $n$-decidable); if so, we have $\spec_{\T_1}(\varphi_1) \cap \spec_{\T_2}(\varphi_2) \neq \emptyset$. Otherwise, if the second or third case holds, check if $\varphi_1 \land \neq(x_1,\dots,x_k)$ is $\T_1$-satisfiable (where the variables $x_i$ are fresh); if so, we have $\spec_{\T_1}(\varphi_1) \cap \spec_{\T_2}(\varphi_2) \neq \emptyset$. If neither check passes, we have $\spec_{\T_1}(\varphi_1) \cap \spec_{\T_2}(\varphi_2) = \emptyset$.
\end{proof}

\begin{example}
 Let $n \in \No$, and denote the theory of fixed-width bit-vectors of length $n$ by $\T_n^{B}$.
 This is an important fragment of the SMT-LIB theory of bit-vectors~\cite{SMTLIB}. 
 The theory $\T_n^{B}$
 is $2^n$-shiny, as all models of all formulas that are satisfiable in it have size $2^n$.
 By \Cref{thm-n-decidable},
 its combination with any $2^n$-decidable theory is decidable. 
 This improves the decidability results for combinations of bit-vectors with other theories, obtainable from the existing shiny and gentle combination theorems.
\end{example}

For our sharpness results, we need to state two more properties of our test theories.
\begin{lemma} \label{lem-n-shiny-test}
    \begin{enumerate}
        \item If $k,k' \in \No$, $k < k'$, and $k \neq n$, then $\T^{k'}_{k}$ is $n$-decidable.
        \item $\Teqn$ is $n$-shiny.
    \end{enumerate}
\end{lemma}

\begin{theorem} \label{thm-f-n-decidable}
    If $\T$ is a decidable theory that can be combined with every decidable theory that is $n$-decidable, then $\T$ is $n$-shiny, meaning $\Gal(\class{$n$-decidable}) \subseteq \class{$n$-shiny}$.
\end{theorem}
\begin{proof}
    Suppose $\T$ can be combined with every $n$-decidable theory.

    First, we show that $\T$ has computable finite spectra. Since $\Testcfs$ is decidable and has computable finite spectra  (\Cref{theo:testtheories}), it is $n$-decidable, $\T \sqcup \Testcfs$ is decidable. Hence, given $k \in \No$, we can compute whether $k \in \spec_\T(\varphi)$ by checking if $\varphi \land P_k$ is $\T \sqcup \Testcfs$-satisfiable. Thus, $\T$ has computable finite spectra.

    Now, it suffices to show that for every $\T$-satisfiable quantifier-free formula $\varphi$, we either have $\spec_{\T}(\varphi) = \{n\}$, $\spec_{\T}(\varphi) = \{n\} \cup \{m \in \N : m \ge k\}$, or $\spec_{\T}(\varphi) = \{m \in \N : m \ge k\}$ for some $k \in \No$; the fact that $\T$ has computable finite spectra implies that we can actually compute $k$ and which case holds. To show that each spectrum has the desired form, it in turn suffices to show that for every quantifier-free formula $\varphi$: (i) if $\spec_\T(\varphi) \setminus \{n\} \neq \emptyset$, then $(\spec_\T(\varphi) \setminus \{n\}) \cap \No \neq \emptyset$; and (ii) if $k \in \spec_\T(\varphi) \setminus \{n\}$, then $k' \in \spec_\T(\varphi)$ for all $k' \ge k$.
    
    First, we show that if $\spec_\T(\varphi) \setminus \{n\} \neq \emptyset$, then $(\spec_\T(\varphi) \setminus \{n\}) \cap \No \neq \emptyset$. Suppose not. Then, $\spec_\T(\varphi \land \neq(x_1,\dots,x_{n+1})) = \{\Inf\}$ (where the variables $x_i$ are fresh). Hence, $\varphi \land \neq(x_1,\dots,x_{n+1}) \land P_m$ is $\T \sqcup \Tle^{\No}$-satisfiable if and only if $F(m) = \Inf$, which is undecidable by our assumption on $F$. This is a contradiction, given that $\Tle^{\No}$ has computable finite spectra  (\Cref{theo:testtheories}) and is therefore $n$-decidable.

    Second, we show that if $k \in \spec_\T(\varphi) \setminus \{n\}$, then $k' \in \spec_\T(\varphi)$ for all $k' \ge k$. If not, then we have $k \in \spec_\T(\varphi) \setminus \{n\}$ and $k' \notin \spec_\T(\varphi)$ for some $k,k' \in \N$ with $k < k'$. In fact, by compactness, we can take $k,k' \in \No$. Then, $\varphi \land P_m$ is $\T \sqcup \T^{k'}_{k}$-satisfiable if and only if $m \notin \unc$, which is undecidable since $\unc$ is undecidable. This is a contradiction, given that $\T^{k'}_{k}$ is $n$-decidable for $k \neq n$  (\Cref{lem-n-shiny-test}).
\end{proof}

The proof shows that a decidable theory $\T$ is $n$-shiny if and only if it can be combined with the test theories $\{\Tle^{\No},\Testcfs\} \cup \{\T^{k'}_{k} : k,k' \in \No,\; k<k',\; k \neq n\}$. We can improve this characterization to only use a single test theory, exhibiting a $\class{$n$-decidable}$-complete theory.
\begin{restatable}{theorem}{thmndecidablecomplete} \label{thm-n-decidable-complete}
    There is a decidable theory $\T'$ that is $n$-decidable and such that for every decidable theory $\T$, if $\T$ can be combined with $\T'$, then $\T$ is $n$-shiny.
\end{restatable}

Finally, we prove that $\Gal(\class{$n$-shiny}) = \class{$n$-decidable}$ (half of this was proven in \Cref{thm-n-decidable}).
\begin{theorem}
    If $\T$ is a decidable theory that can be combined with every decidable $n$-shiny theory, then $\T$ is $n$-decidable.
    This means $\Gal(\class{$n$-shiny})\subseteq \class{$n$-decidable}$, and thus $Cl(\class{$n$-shiny})=\class{$n$-shiny}$ and $Cl(\class{$n$-decidable})=\class{$n$-decidable}$.
\end{theorem}
\begin{proof}
    Suppose $\T$ can be combined with every decidable $n$-shiny theory. Since $\Teqn$ is a decidable $n$-shiny theory  (\Cref{lem-n-shiny-test}), $\T \sqcup \Teqn$ is decidable. But $\varphi$ is $\T \sqcup \Teqn$-satisfiable if and only if $n \in \spec_\T(\varphi)$, so $\T$ is $n$-decidable.
\end{proof}

The proof shows that a decidable theory $\T$ is $n$-decidable if and only if it can be combined with the test theory $\Teqn$. Thus, $\Teqn$ is $\class{$n$-shiny}$-complete.

\subsection{Combination of \quagen{} Theories With Co-$\mathfrak{F}$-QG Theories} \label{sec-quasi-gentle}

\Cref{sec-n-decidable} exhibited $\aleph_0$ many combination theorems; now we go further and exhibit $2^{\aleph_0}$ many combination theorems. This isn't just a theoretical curiosity: the corresponding sets of theories are natural generalizations of gentle theories, and for that reason, we call them \emph{quasi-gentle}.
To motivate them, recall gentle theories can be combined with themselves, although 
this was later improved to allow gentle theories to be combined with theories having computable finite spectra.
Taking into account the fact that $\class{gentle}\subseteq\class{CFS}$, one may wonder whether there is a property $X$ such that $\class{gentle}\subseteq\class{X}\subseteq \class{CFS}$ and $\Gal(\class{X})=\class{X}$.
Looking for such a solution leads to the methods to be described below, some of them being the fixed points we were looking for.

We start with some terminology on filters, see \cite{Royal} for a general reference. 
A \emph{filter} $\mathfrak{F}$ on $\mathbb{N}$ is a subset of the power set $\mathbb{P}(\mathbb{N})$ such that:
$(i)$ $\emptyset\notin \mathfrak{F}$ and $\mathbb{N}\in\mathfrak{F}$; 
$(ii)$ if $A,B\in \mathfrak{F}$, then $A\cap B\in \mathfrak{F}$;
and $(iii)$ if $A\in \mathfrak{F}$ and $A\subseteq B$, then $B\in \mathfrak{F}$.
We say that $\mathfrak{F}$ is \emph{free} if the intersection of all of its elements is empty;
this is equivalent to $\mathfrak{F}$ containing all cofinite sets.
All sets in a free filter are infinite.
An ultrafilter $\mathfrak{F}$ is a maximal filter; equivalently, $\mathfrak{F}$ is an ultrafilter if for every $A\in\mathbb{P}(\mathbb{N})$, either $A\in \mathfrak{F}$ or $\mathbb{N}\setminus A\in\mathfrak{F}$.
Free ultrafilters exist assuming the axiom of choice. More generally, we can define filters and ultrafilters on any lattice by substituting $\wedge$ for $\cap$ and $\le$ for $\subseteq$.

\begin{definition}
    Let $\mathfrak{F}$ be a free filter.
    A decidable theory is said to be $\mathfrak{F}$-quasi-gentle, or simply \quagen, when it has computable finite spectra, and the spectrum of every quantifier-free formula is either a finite set of finite cardinalities, or $A\cup\{\aleph_{0}\}$ for an $A\in\mathfrak{F}$. The set of decidable \quagen{} theories is denoted $\class{\quagen}$.

    A decidable theory is said to be co-$\mathfrak{F}$-quasi-gentle, or simply \coquagen, when it has computable finite spectra, and the spectrum of every quantifier-free formula is either a finite set of finite cardinalities, or $A\cup\{\aleph_{0}\}$ for an $A\subseteq\mathbb{N}$ such that $\mathbb{N}\setminus A\notin \mathfrak{F}$. The set of decidable \coquagen{} theories is denoted $\class{\coquagen}$.
\end{definition}

Notice that the spectrum of quantifier-free formulas in a \quagen{} theory can be $A\cup\{\aleph_{0}\}$ for any cofinite set $A$, since $\mathfrak{F}$ contains every cofinite set.
Notice too that if $\mathfrak{F}$ is an ultrafilter then a theory is \quagen{} if and only if it is \coquagen{}:
indeed, $A \in \mathfrak{F}$ if and only if $\mathbb{N}\setminus A \notin \mathfrak{F}$.
This will give us a collection of fixed points for $\Gal$ between $\class{CFS}$ and $\class{gentle}$.

As a concrete example, let $\mathbb{F}$ be the free filter of cofinite sets, usually called the \emph{Fr{\'e}chet filter}~\cite{Jech}. Then, $\mathbb{F}$-QG theories are those with computable finite spectra such that the spectrum of every quantifier-free formula is either a finite set of finite cardinalities or a cofinite set of cardinalities; co-$\mathbb{F}$-QG theories are those with computable finite spectra such that the spectrum of every quantifier-free formula is either a finite set of finite cardinalities or an infinite set of cardinalities. Observe that being $\mathbb{F}$-QG is a natural weakening of being gentle: we require the spectra to be finite or cofinite, but we do not require a computable method to determine which case holds or to output an explicit representation of the spectrum (or its complement). This justifies the name \emph{quasi-gentleness}. Also notice that being co-$\mathbb{F}$-QG is a slight strengthening of having computable finite spectra.

\begin{theorem} \label{thm-qg-combinable}
    Let $\T_1$ and $\T_2$ be decidable theories. Suppose that $\T_1$ is \quagen{} and $\T_2$ is \coquagen{}. Then, $\T_1 \sqcup \T_2$ is decidable, i.e. $\class{\coquagen{}} \subseteq \Gal(\class{\quagen{}})$ and $\class{\quagen}\subseteq\Gal(\class{\coquagen})$.
\end{theorem}

In light of Fontaine's lemma, the following implies \Cref{thm-qg-combinable}:
\begin{lemma}
    Let $\T_1$ and $\T_2$ be decidable theories over signatures $\Sigma_1$ and $\Sigma_2$ respectively. Suppose that $\T_1$ is \quagen{} and that $\T_2$ is \coquagen{} for some free filter $\mathfrak{F}$. Then, it is decidable whether $\spec_{\T_1}(\varphi_1) \cap \spec_{\T_2}(\varphi_2) = \emptyset$, where $\varphi_1$ and $\varphi_2$ are conjunctions of literals over the signatures $\Sigma_1$ and $\Sigma_2$ respectively.
\end{lemma}
\begin{proof}
    Take conjunctions of literals $\varphi_{1}$ and $\varphi_{2}$ in, respectively, $\Sigma_{1}$ and $\Sigma_{2}$. 
    For each $n$, starting with $1$, we test whether $\varphi_{1}$ and $\varphi_{2}$ both have models of cardinality $n$:
    if so, we return $ \spec_{\T_{1}}(\varphi_{1})\cap \spec_{\T_{2}}(\varphi_{2})$ is not empty.
    If not, we test whether $\varphi_{1}\wedge\NNNEQ{x}{n}$ and $\varphi_{2}\wedge\NNNEQ{x}{n}$ are, respectively, $\T_{1}$- and $\T_{2}$-satisfiable:
    if not, we return $ \spec_{\T_{1}}(\varphi_{1})\cap \spec_{\T_{2}}(\varphi_{2})$ is empty;
    if so, we proceed to $n+1$.
    This procedure always terminates, for if either $ \spec_{\T_{1}}(\varphi_{1})$ or $ \spec_{\T_{2}}(\varphi_{2})$ is finite there will come a point where $\varphi_{1}\wedge\NNNEQ{x}{n}$ or $\varphi_{2}\wedge\NNNEQ{x}{n}$ are no longer satisfiable in their respective theories;
    and if both $ \spec_{\T_{1}}(\varphi_{1})=A\cup\{\aleph_{0}\}$ for an $A\in\mathfrak{F}$, and $ \spec_{\T_{2}}(\varphi_{2})=B\cup\{\aleph_{0}\}$ for a $B\subseteq \mathbb{N}$ such that $\mathbb{N}\setminus B\notin \mathfrak{F}$, $ \spec_{\T_{1}}(\varphi_{1})\cap \spec_{\T_{2}}(\varphi_{2})\cap\mathbb{N}$ is not empty and the algorithm will eventually find the first shared element.
    Indeed, had we $A\cap B=\emptyset$, this would imply $A\subseteq \mathbb{N}\setminus B$, and given $A\in\mathfrak{F}$ we would find $\mathbb{N}\setminus B\in\mathfrak{F}$, a contradiction.
\end{proof}

For our sharpness results, we need to state more properties of our test theories.
\begin{lemma} \label{lem-qg-test}
    \begin{enumerate}
        \item $\Testcfs$ is $\mathfrak{F}$-QG and co-$\mathfrak{F}$-QG for all free filters $\mathfrak{F}$.
        \item If $S \in \mathfrak{F}$, then $\Testg$ is \quagen{}.
        \item If $\mathbb{N} \setminus S \notin \mathfrak{F}$, then $\Testg$ is \coquagen{}.
    \end{enumerate}
    In particular, $\Tle^{\No}$ is \quagen{} and \coquagen{} for all free filters $\mathfrak{F}$.
\end{lemma}

We now prove one direction of  $\Gal(\class{\quagen{}}) = \class{\coquagen{}}$, the other having been proven in \Cref{thm-qg-combinable}.

\begin{theorem}
    If $\T$ is a decidable theory that can be combined with every \quagen{} theory, then $\T$ is \coquagen, meaning $\Gal(\class{\quagen{}}) \subseteq \class{\coquagen{}}$.
\end{theorem}

\begin{proof}
    Suppose $\T$ can be combined with every \quagen{} decidable theory, and take a quantifier-free formula $\varphi$.

    We first show that $\T$ has computable finite spectra.
    Indeed, $\varphi\wedge P_{n}$ is $\T\sqcup\Testcfs$-satisfiable if and only if $n\in \spec_{\T}(\varphi)$, and $\Testcfs$ is \quagen{}  by \Cref{lem-qg-test}.

    We then show by contradiction that, if $\spec_{\T}(\varphi)$ is finite, then it cannot contain $\aleph_{0}$.
    Were that not the case, let $n$ be the maximum of $\spec_{\T}(\varphi)\cap\mathbb{N}$, and consider the \quagen{} theory $\Tle^{\No}$ (\Cref{lem-qg-test}):
    $\varphi\wedge\NEQ{(x_{1},\ldots,x_{n+1})}\wedge P_{m}$, for fresh variables $x_{i}$, is then $\T\sqcup\Tle^{\No}$-satisfiable if and only if $F(m)=\aleph_{0}$, which is undecidable by our assumption on $F$. This contradicts the fact that $\T\sqcup\Tle^{\No}$ is decidable.

    Now we prove, again by contradiction, that if $\spec_{\T}(\varphi)$ is infinite, then $S \coloneqq \mathbb{N}\setminus\spec_{\T}(\varphi) \notin \mathfrak{F}$. Suppose $S \in \mathfrak{F}$.
    Then, $S$ must be decidable, since it is the complement of the decidable set $\spec_{\T}(\varphi)$ (given that $\T$ has computable finite spectra).
    Then $\varphi\wedge P_{m}$ is $\T\sqcup\Testg$-satisfiable if, and only if, $F(m)=\aleph_{0}$;
    notice that, of course, $\Testg$ is \quagen{} (\Cref{lem-qg-test}).   
\end{proof}

Next, we prove that $\Gal(\class{\coquagen{}}) = \class{\quagen{}}$, half of which was proven in \Cref{thm-qg-combinable}.

\begin{theorem}
    If $\T$ is a decidable theory that can be combined with every \coquagen{} theory, then $\T$ is \quagen, i.e. $\Gal(\class{\coquagen{}}) \subseteq \class{\quagen{}}$.
\end{theorem}

\begin{proof}
    The proof is more or less the same as the one for the previous theorem;
    let $\varphi$ be a quantifier-free formula.
    $\spec_{\T}(\varphi)\cap\mathbb{N}$ is decidable, and so $\T$ has computable finite spectra, as $\varphi\wedge P_{n}$ is $\T\sqcup\Testcfs$-satisfiable if and only if $n\in \spec_{\T}(\varphi)$, and $\Testcfs$ is both decidable and \coquagen{} (\Cref{lem-qg-test}).

    By contradiction, if $\spec_{\T}(\varphi)$ is finite and contains $\aleph_{0}$, let $n$ be the maximum of $\spec_{\T}(\varphi)\cap\mathbb{N}$, and consider the \coquagen{} theory $\Tle^{\No}$ (\Cref{lem-qg-test}):
    $\varphi\wedge\NEQ{(x_{1},\ldots,x_{n+1})}\wedge P_{m}$ is then $\T\sqcup\Tle^{\No}$-satisfiable if and only if $F(m)=\aleph_{0}$, leading to a contradiction.

    Similarly, if $\spec_{\T}(\varphi)$ is infinite but not in $\mathfrak{F}$, then $S\coloneqq \mathbb{N}\setminus\spec_{\T}(\varphi)$ is decidable and $\mathbb{N}\setminus S \notin \mathfrak{F}$:
    $\Testg$ is then \coquagen{} by \Cref{lem-qg-test}, and 
    since $\varphi\wedge P_{m}$ is $\T\sqcup\Testg$-satisfiable if and only if $F(m)=\aleph_{0}$, we again reach a contradiction.
\end{proof}

We leave open the question of whether there are any $\class{\quagen{}}$-complete and $\class{\coquagen{}}$-complete theories (for all we know, it is possible that the answer depends on the choice of filter $\mathfrak{F}$).

\subsubsection{The Structure of \quagen{} and Co-$\mathfrak{F}$-QG Theories}

We now prove a few propositions that tell us how the sets of \quagen{} and \coquagen{} theories relate to the lattice $\mathcal{C}$ induced by our Galois connection.

Let $REC$ be the lattice, indeed a Boolean algebra, of decidable subsets of $\mathbb{N}$. Every filter $\mathfrak{F}$ on $\mathbb{N}$ induces a filter on $REC$ by restricting to $\mathfrak{F} \cap REC$. We first show that $\class{(co-)$\mathfrak{F}$-QG}$ is determined by $\mathfrak{F} \cap REC$.

\begin{proposition} \label{prop-computable-equiv}
    If $\mathfrak{F} \cap REC = \mathfrak{G} \cap REC$, then $\class{(co-)$\mathfrak{F}$-QG} = \class{(co-)$\mathfrak{G}$-QG}$.
\end{proposition}
\begin{proof}
    We prove the proposition for quasi-gentleness, with the proof for co-quasi-gentleness being similar. For any \quagen{} theory $\T$ and quantifier-free formula $\varphi$, we have $\spec_\T(\varphi) \cap \mathbb{N} \in \mathfrak{F}$ if and only if $\spec_\T(\varphi) \cap \mathbb{N} \in \mathfrak{G}$, since $\spec_\T(\varphi) \cap \mathbb{N}$ is decidable. Thus, $\T$ is $\mathfrak{F}$-QG if and only if it is $\mathfrak{G}$-QG.
\end{proof}

Thus, we are justified in restricting our attention to free filters on $REC$. We abuse notation by referring to \quagen{} and \coquagen{} theories even when $\mathfrak{F}$ is a free filter on $REC$.

\begin{proposition} \label{prop-refine}
    Let $\mathfrak{F}$ and $\mathfrak{G}$ be free filters on $REC$. We have $\mathfrak{F} \subseteq \mathfrak{G}$ if and only if $\class{$\mathfrak{F}$-QG} \subseteq \class{$\mathfrak{G}$-QG}$ if and only if $\class{co-$\mathfrak{G}$-QG} \subseteq \class{co-$\mathfrak{F}$-QG}$.
\end{proposition}
\begin{proof}
    We prove the proposition for quasi-gentleness, with the proof for co-quasi-gentleness being similar.

    Suppose $\mathfrak{F} \subseteq \mathfrak{G}$ and let $\T$ be a $\mathfrak{F}$-QG theory:
    then the spectra of its quantifier-free formulas are either finite sets of finite cardinalities, or $A\cup\{\aleph_{0}\}$ for an $A\in\mathfrak{F}$.
    Since $\mathfrak{F}\subseteq\mathfrak{G}$, the spectra of the quantifier-free formulas of $\T$ are either finite sets of finite cardinalities, or $A\cup\{\aleph_{0}\}$ for an $A\in\mathfrak{G}$, meaning it is a $\mathfrak{G}$-QG theory.

    Conversely, suppose $\mathfrak{F} \not\subseteq \mathfrak{G}$, and let $S \in \mathfrak{F} \setminus \mathfrak{G}$. Then, let $\T$ be the $\Sigma_{1}$-theory axiomatized by $\{\lnot \psi_{=n} : n \notin S\}$. We have $\spec_\T(\top) = S \cup \{\aleph_0\}$, and $\T$ is $\mathfrak{F}$-QG but not $\mathfrak{G}$-QG. Thus, $\class{$\mathfrak{F}$-QG} \not\subseteq \class{$\mathfrak{G}$-QG}$.
\end{proof}

Given a family of sets $\{A_i : i \in I\}$, we say that it has the \emph{strong finite intersection property} if
\[
    \bigcap_{i \in I'} A_i
\]
is infinite for any finite $I' \subseteq I$. If $\{A_i : i \in I\}$ has the strong finite intersection property, we say that the filter \emph{generated} by $\{A_i : i \in I\}$ is the smallest filter containing every $A_i$; this filter is free. If each $A_i$ is decidable, we can similarly generate a free filter on $REC$.

\begin{proposition}
    There is set $\mathcal{F}$ of $2^{\aleph_0}$ free filters on $REC$ such that the sets of theories $\{\class{(co-)$\mathfrak{F}$-QG} : \mathfrak{F} \in \mathcal{F}\}$ form a chain in the lattice $\mathcal{C}$ induced by the Galois connection. Additionally, there is set $\mathcal{F}$ of $2^{\aleph_0}$ free filters on $REC$ such that the sets of theories $\{\class{(co-)$\mathfrak{F}$-QG} : \mathfrak{F} \in \mathcal{F}\}$ form an antichain in $\mathcal{C}$.
\end{proposition}
\begin{proof}
    By \Cref{prop-refine}, it suffices to construct a set of $2^{\aleph_0}$ free filters on $REC$ that form a (anti)chain when ordered by inclusion.

    For each $i \in \No$, let $A_i$ be the natural numbers whose $i$th bit is a 0 (starting at the least significant bit); for example, $A_1$ is the even numbers, and $A_2 = \{0,1,4,5,\dots\}$. Then, $\{A_i : i \in \No\}$ has the strong finite intersection property, and each $A_i$ is decidable. Given a set $S \subseteq \mathbb{P}(\No)$, let $\mathfrak{F}_S$ be the free filter on $REC$ generated by $\{A_i : i \in S\}$. Note that if $S, S' \subseteq \mathbb{P}(\No)$, then $S \subseteq S'$ if and only if $\mathfrak{F}_S \subseteq \mathfrak{F}_{S'}$. Thus, given a (anti)chain of size $2^{\aleph_0}$ in $\mathbb{P}(\No)$, we have a (anti)chain of $2^{\aleph_0}$ free filters on $REC$.
\end{proof}

We have shown that $\mathcal{C}$ is not only infinite, but quite large as well: it has height and width at least $2^{\aleph_0}$.

\subsection{Pseudocode} \label{sec-pseudocode}

We provide pseudocode for each of the combination methods introduced in this section. \Cref{alg-fontaine} gives pseudocode implementing Fontaine's lemma, from which all of our combination methods can be derived. Then, \Cref{alg-new-methods} shows how each combination method allows us to fulfill the requirement of Fontaine's lemma.

\begin{algorithm}[ht]
    \caption{Pseudocode for Fontaine's lemma. The function returns whether $\varphi_1 \land \varphi_2$ is $\T_1 \sqcup \T_2$-satisfiable.} \label{alg-fontaine}
    \begin{algorithmic}[1]
        \Require An algorithm to decide whether $\spec_{\T_1}(\varphi_1) \cap \spec_{\T_2}(\varphi_2) = \emptyset$, where $\varphi_1$ and $\varphi_2$ are conjunctions of literals over $\Sigma_1$ and $\Sigma_2$ respectively
        \Function{Fontaine}{$\varphi_1, \varphi_2$}
        \ForAll{arrangements $\delta$ on $\vars(\varphi_1) \cap \vars(\varphi_2)$}
            \If{$\spec_{\T_1}(\varphi_1 \land \delta) \cap \spec_{\T_2}(\varphi_2 \land \delta) \neq \emptyset$}
                \State \Return \TRUE
            \EndIf
        \EndFor
        \State \Return \FALSE
        \EndFunction
    \end{algorithmic}
\end{algorithm}

\begin{algorithm}[p]
    \caption{Pseudocode for the new combination methods. Each function returns whether $\spec_{\T_1}(\varphi_1) \cap \spec_{\T_2}(\varphi_2) \neq \emptyset$, where $\varphi_1$ and $\varphi_2$ are conjunctions of literals over $\Sigma_1$ and $\Sigma_2$ respectively. We assume throughout that $\T_1$ and $\T_2$ are decidable.} \label{alg-new-methods}
    \begin{algorithmic}[1]
        \Require $\T_1$ is smooth and has computable spectra and $\T_2$ is infinitely decidable
        \Function{SM+CS}{$\varphi_1, \varphi_2$}
            \If{$\aleph_0 \in \spec_{\T_2}(\varphi_2)$}
                \State \Return $\varphi_1$ is $\T_1$-satisfiable
            \EndIf
            \State $k \gets 0$
            \While{$\varphi_2 \land \neq(x_{1},\ldots,x_{k+1})$ is $\T_2$-satisfiable} \Comment{$x_{1},\ldots,x_{k+1}$ are fresh}
                \State $k \gets k+1$
            \EndWhile
            \State \Return $k \in \spec_{\T_1}(\varphi_1)$
        \EndFunction

        \Statex

        \Require $\T_1$ and $\T_2$ have computable spectra
        \Function{CS}{$\varphi_1, \varphi_2$}
            \If{$\aleph_0 \in \spec_{\T_1}(\varphi_1) \cap \spec_{\T_2}(\varphi_2)$}
                \State \Return \TRUE
            \ElsIf{$\aleph_0 \notin \spec_{\T_1}(\varphi_1)$}
                \State $i \gets 1$
            \Else \Comment{$\aleph_0 \notin \spec_{\T_2}(\varphi_2)$}
                \State $i \gets 2$
            \EndIf
            \State $k \gets 0$
            \While{$\varphi_i \land \neq(x_{1},\ldots,x_{k+1})$ is $\T_i$-satisfiable} \Comment{$x_{1},\ldots,x_{k+1}$ are fresh}
                \State $k \gets k+1$
            \EndWhile
            \ForAll{$n \in [k]$}
                \If{$n \in \spec_{\T_1}(\varphi_1) \cap \spec_{\T_2}(\varphi_2)$}
                    \State \Return \TRUE
                \EndIf
            \EndFor
            \State \Return \FALSE
        \EndFunction

        \Statex

        \Require $\T_1$ is $n$-shiny and $\T_2$ is $n$-decidable
        \Function{$n$-shiny}{$\varphi_1, \varphi_2$}
            \If{$n \in \spec_{\T_1}(\varphi_1) \cap \spec_{\T_2}(\varphi_2)$}
                \State \Return \TRUE
            \ElsIf{$\varphi_1$ is $\T_1$-satisfiable}
                \State $k \gets \max(\mathbb{N} \setminus \spec_{\T_1}(\varphi_1)) + 1$
                \State \Return $\varphi_2 \land \neq(x_{1},\ldots,x_{k})$ is $\T_2$-satisfiable \Comment{$x_{1},\ldots,x_{k}$ are fresh}
            \Else
                \State \Return \FALSE
            \EndIf
        \EndFunction

        \Statex

        \Require $\T_1$ is $\mathfrak{F}$-quasi-gentle and $\T_2$ is co-$\mathfrak{F}$-quasi-gentle \Comment{$\mathfrak{F}$ is a free filter}
        \Function{quasi-gentle}{$\varphi_1, \varphi_2$}
            \State $n \gets 1$
            \While{$\varphi_i \land \neq(x_{1},\ldots,x_{n})$ is $\T_i$-satisfiable for both $i \in \{1,2\}$} \Comment{$x_{1},\ldots,x_{n}$ are fresh}
                \If{$n \in \spec_{\T_1}(\varphi_1) \cap \spec_{\T_2}(\varphi_2)$}
                    \State \Return \TRUE
                \EndIf
                \State $n \gets n+1$
            \EndWhile
            \State \Return \FALSE
        \EndFunction
    \end{algorithmic}
\end{algorithm}

\section{The Relationships Between Combination Theorems} \label{sec-lattice}

We wrap up by plotting in \Cref{fig:lattice} the Hasse diagram for the lattice of theory combination properties we have examined. We claim that the diagram is correct in the following sense: we have $\class{X} \subseteq \class{Y}$ if and only if there is an upward path from $\class{X}$ to $\class{Y}$. Furthermore, while we sometimes have $\class{\quagen{}} = \class{\coquagen{}}$ (e.g., when $\mathfrak{F}$ is an ultrafilter), for every other pair of sets, the inclusion is strict. The Galois connection $\Gal$ gives rise to a symmetry about the dashed line; for example: $\Gal(\class{gentle}) = \Gal(\class{CFS})$, $\Gal(\class{shiny})=\Gal(\class{})$, and so on.

We start by showing that the lines in the Hasse diagram correspond to inclusions.
\begin{restatable}{proposition}{propinclusion} \label{prop-inclusion}
    For every pair of sets $\class{X}$ and $\class{Y}$ in \Cref{fig:lattice}, if there is an upward path from $\class{X}$ to $\class{Y}$, then $\class{X} \subseteq \class{Y}$.
\end{restatable}

\begin{table}
\centering
\caption{Example theories. $F:\No\rightarrow\No\cup\{\Inf\}$ is a function such that:
$(i)$ $\{(m,n) \in {\No}^2 : F(m)\geq n\}$ is decidable; and $(ii)$ $\{n : F(n)=\aleph_{0}\}$ is undecidable.
$\unc \subset \No$ is an undecidable set.}
\renewcommand{\arraystretch}{1.5}
\centering
\begin{tabular}{c|c|c}
Sig. & Name & Axiomatization\\\hline
$\Spn$ & $\Th{d}^{n}$ & $\{P_{2k+1}\rightarrow\psi_{\geq n+1} : k\in \unc\}\cup\{P_{2k}\rightarrow \psi_{\leq F(k)} : F(k)\in\mathbb{N}\}\cup\{P_{i}\rightarrow\neg P_{j} : i\neq j\}$\\
$\Spn$ & $\Th{cfs}$ & $\{P_{1}\rightarrow\psi_{\geq k} : k\in\mathbb{N}\}\cup\{P_{k}\rightarrow\psi_{\leq F(k)} : F(k)\in\mathbb{N}, k\geq 2\}\cup\{P_{i}\rightarrow\neg P_{j} : i\neq j\}$\\
$\Spn$ & $\Th{si}$ & $\{P_{n}\rightarrow \psi_{\geq m} : n\in \unc, m\in\No\}\cup\{P_{i}\rightarrow\neg P_{j} : i\neq j\}$\\
$\Sp$ & $\Th{cs}$ & $\{P\rightarrow \psi_{=1}\}\cup\{\neg P\rightarrow \psi_{\geq m} : m\in\mathbb{N}\}$\\
$\Sp$ & $\Th{n-s}^{n}$ & $\{P\rightarrow\psi_{=n}\}\cup\{\neg P\rightarrow\psi_{\geq n}\}$\\
$\Sigma_{1}$ & $\Tleqn$ & $\{\psi_{\leq n}\}$\\
$\Sigma_{1}$ & $\Tgeqn$ & $\{\psi_{\geq n}\}$
\end{tabular}
\label{tab:extheories}
\centering
\end{table}

Next, we demonstrate that the inclusions indicated by \Cref{fig:lattice} are strict, with the exception of $\class{\quagen{}}$ and $\class{\coquagen{}}$, and that there are no more inclusions.

\begin{proposition} \label{prop-strict-inclusion}
    For every pair of sets $\class{X}$ and $\class{Y}$ in \Cref{fig:lattice}, if there is no upward path from $\class{X}$ to $\class{Y}$ and $(\class{X},\class{Y}) \neq (\class{\coquagen{}},\class{\quagen{}})$, then $\class{X} \not\subseteq \class{Y}$.
\end{proposition}

For each set $\class{X}$ in \Cref{fig:lattice}, we prove \Cref{prop-strict-inclusion} by constructing a theory $\T$ such that $\T \in \class{X} \setminus \class{Y}$ for all sets $\class{Y}$ such that there is no upward path from $\class{X}$ to $\class{Y}$ and $(\class{X},\class{Y}) \neq (\class{\coquagen{}},\class{\quagen{}})$.

Let $\Sp$ be a signature with no functions symbols and a single nullary predicate $P$. We use the theories given in \Cref{tab:extheories} ($\Tleqn$ and $\Tgeqn$ were defined in \cite{CADE}). 

While we sometimes have $\class{\quagen{}} = \class{\coquagen{}}$, we show a separation for some choices of free filter $\mathfrak{F}$, in particular those for which there is a decidable set $A \subset \mathbb{N}$ such that $A \notin \mathfrak{F}$ and $\mathbb{N} \setminus A \notin \mathfrak{F}$ (i.e., the induced filter on $REC$ is not an ultrafilter). For example, if $\mathfrak{F} = \mathbb{F}$ is the Fr{\'e}chet filter, then we can take $A$ to be the even numbers.

We place our example theories in \Cref{fig:lattice-ex-theories}, where each theory occupies the position from \Cref{fig:lattice} of the strongest property it satisfies.
For example, $\Th{d}^{n}$ is decidable, but lacks
the other properties, $\Tgeqn$ has all the properties, and so on.

$\Th{d}^{n}$ is a $\Spn$-theory, dependent on an $n\in\No$, where the validity of a $P_{k}$ implies:
that the interpretation has at least $n+1$ elements if $k$ is odd;
and at most $F(k/2)$ elements if $k$ is even.
$\Th{cfs}$ is also a $\Spn$-theory, where the validity of $P_{1}$ implies the interpretation is infinite, and the validity of $P_{k}$, for $k\geq 2$, implies the interpretation has at most $F(k)$ elements;
this makes $\Th{cfs}$ very similar to $\Tle^{\No}$, except for the behavior of $P_{1}$.
The last of the new $\Spn$-theories is $\Th{si}$, and in this theory the validity of $P_{k}$ implies the interpretation is infinite if and only if $k\in\unc$;
notice that, if $k\notin\unc$, then $P_{k}$ has models of all finite cardinalities too.

$\Th{cs}$ is a $\Sp$-theory easy to define: if $P$ is true, the interpretation has exactly one element, and if it is false then the interpretation must be infinite.
$\Th{n-s}^{n}$ is a $\Sp$-theory as well, but this time dependent on an $n\in\No$, where the truth of $P$ implies an interpretation has cardinality exactly $n$, and its falsity implies an interpretation has at least $n$ elements.
$\Tleqn$ and $\Tgeqn$, both dependent on $n\in\No$, are $\Sigma_{1}$-theories where an interpretation must have, respectively, at most and at least $n$ elements.

\begin{figure}[h!]
    \centering
    \adjustbox{scale=1.0,center}{
    \begin{tikzcd}[sep=small,column sep=3.0em,row sep=0.35em]
& & \Th{d}^{n} \arrow[dash]{ddr}{} \arrow[dash]{dl}{} & & \\%
& \Th{d}^{n-1} \arrow[dash]{dl}{} & & & \\
\Th{cfs}\arrow[dash]{ddrr}{}\arrow[dash]{d}{} & & & \T_{n}^{n+1}\arrow[dash]{ddr}{}\arrow[dash]{ddl}{} & \\
\Tle^{A}\arrow[dash]{dd}{} & & & &  \\
 & & \Th{cs}\arrow[dash]{ddr}{}\arrow[dash]{ddll}{} & & \Th{si}\arrow[dash]{ddl}{} \\
\Tle^{\No}\arrow[dash]{d}{} & & & & \\
\Tleqn\arrow[dash]{dr}{} & & & \Tinfty\arrow[dash]{ddl}{} & \\
& \Th{n-s}^{n}\arrow[dash]{dr}{} & & & \\
& & \Tgeqn & & \\%
\end{tikzcd}}
    \caption{The example theories used to separate the theory combination properties in \Cref{fig:lattice}}
    \label{fig:lattice-ex-theories}
\end{figure}
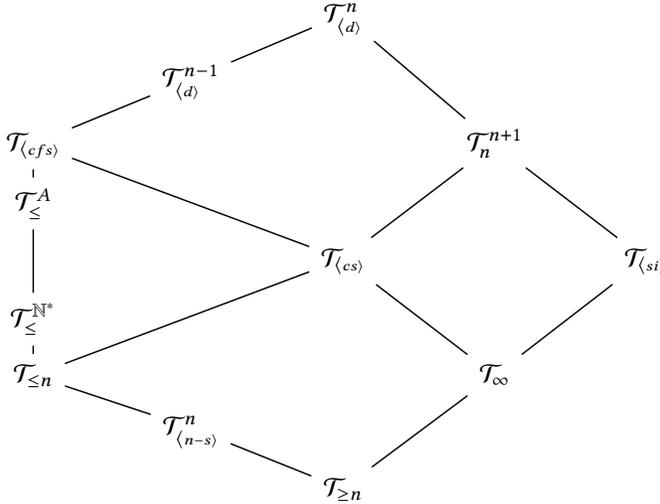

\subsection{Lattice-Theoretic Properties of $\mathcal{C}$}

Recall from \Cref{sec-galois} that $\mathcal{C}$ is the lattice of closed elements of $\mathbb{P}(\mathfrak{T})$ (equivalently, the image of $\Gal$). We have made some basic observations about $\mathcal{C}$: it is complete, it is isomorphic to its dual, and it has height and width at least $2^{\aleph_0}$. Other properties of the lattice remain unknown. For but one example, we do not know if $\mathcal{C}$ is distributive. Given the demonstrated utility of computing meets and joins in $\mathcal{C}$, such algebraic questions seem worthy of study.

\section{Conclusion and Further Directions}

Ever since Nelson and Oppen's seminal paper~\cite{NelsonOppen}, the central theoretical question in theory combination has been when sets of theories can be disjointly combined. We provided a systematic answer to this question by developing a general technique for identifying the largest set of theories that can be disjointly combined with a given set. We demonstrated this technique by characterizing the theories that can be disjointly combined with sets of theories to which existing theory combination methods apply, in each case either showing that the method is sharp or improving the method to make it sharp. Moreover, we showed how to use our framework to easily generate new combination methods.

The potential impact of our work is twofold. First, existing as well as future combination methods 
can be better understood by viewing them through the lens of our Galois connection. Specifically, given a new combination method, we can determine whether it is sharp and compare it to previous combination methods based on its position in the lattice.
Second, this work paves the way to discover new combination methods that could be implemented in SMT solvers for combined theories. Thus, 
when a new ``container theory'' (such as that of uninterpreted functions, arrays, etc.) 
is introduced, its various properties can be evaluated in order to produce
a combination method that allows it to be combined with the largest possible
set of decidable theories.

Our work suggests several avenues for further research. One is to generalize our results to many-sorted logic, which is particularly important for applications to SMT. We are also interested in whether analogous results can be proved for non-disjoint theory combination. There is some work on non-disjoint theory combination (e.g., \cite{unions-non-disjoint,ghilardi-non-disjoint,gentle-non-disjoint,polite-non-disjoint}), although it has been less thoroughly explored than the disjoint case.

Another direction is to develop more efficient combination methods inspired by the methods we have considered, several of which are computationally expensive. Once a problem is known to be decidable, it is often possible to develop tractable algorithms. Indeed, one motivation for (strongly) polite theory combination was to be a more efficient alternative to shiny theory combination. On the theoretical side, one could study which pairs of theories are \emph{efficiently} combinable, meaning roughly that they can be combined without an increase in computational complexity. Oppen~\cite{OppenSI} proved some results along these lines, demonstrating the role that convexity plays in the efficiency of the Nelson--Oppen method.

\section*{Acknowledgments}

Przybocki was supported by the NSF Graduate Research Fellowship Program under Grant No. DGE-2140739.
Zohar and Toledo were supported by BSF grant number 2020704, DARPA 
contract FA875024-2-1001,
and a gift from AWS.

\bibliographystyle{ACM-Reference-Format}
\bibliography{bib}

@book{Royal,
  title={A Royal Road to Topology: Convergence of Filters},
  author={Szymon Dolecki},
  isbn={9789811232107},
  lccn={2022013688},
  series={G - Reference,Information and Interdisciplinary Subjects Series},
  url={https://books.google.com.br/books?id=F4oZzgEACAAJ},
  year={2024},
  publisher={World Scientific}
}

@inproceedings{CADE30,
  author       = {Guilherme Vicentin de Toledo and
                  Benjamin Przybocki and
                  Yoni Zohar},
  editor       = {Clark W. Barrett and
                  Uwe Waldmann},
  title        = {Being Polite Is Not Enough (and Other Limits of Theory Combination)},
  booktitle    = {Automated Deduction - {CADE} 30 - 30th International Conference on
                  Automated Deduction, Stuttgart, Germany, July 28-31, 2025, Proceedings},
  series       = {Lecture Notes in Computer Science},
  volume       = {15943},
  pages        = {17--34},
  publisher    = {Springer},
  year         = {2025},
  url          = {https://doi.org/10.1007/978-3-031-99984-0\_2},
  doi          = {10.1007/978-3-031-99984-0\_2},
  bibsource    = {dblp computer science bibliography, https://dblp.org}
}

@inproceedings{FroCoS2025,
  author       = {Benjamin Przybocki and
                  Guilherme Vicentin de Toledo and
                  Yoni Zohar},
  editor       = {Ren{\'{e}} Thiemann and
                  Christoph Weidenbach},
  title        = {Shininess, Strong Politeness, and Unicorns},
  booktitle    = {Frontiers of Combining Systems - 15th International Symposium, FroCoS
                  2025, Reykjavik, Iceland, September 29 - October 1, 2025, Proceedings},
  series       = {Lecture Notes in Computer Science},
  volume       = {15979},
  pages        = {135--152},
  publisher    = {Springer},
  year         = {2025},
  url          = {https://doi.org/10.1007/978-3-032-04167-8\_8},
  doi          = {10.1007/978-3-032-04167-8\_8},
  bibsource    = {dblp computer science bibliography, https://dblp.org}
}

@inproceedings{de2009generalized,
  title={Generalized, efficient array decision procedures},
  author={De Moura, Leonardo and Bj{\o}rner, Nikolaj},
  booktitle={2009 Formal Methods in Computer-Aided Design},
  pages={45--52},
  year={2009},
  organization={IEEE}
}

@inproceedings{LPAR,
  author    = {Guilherme V. Toledo and Yoni Zohar},
  title     = {Combining Combination Properties: Minimal Models},
  booktitle = {Proceedings of 25th Conference on Logic for Programming, Artificial Intelligence and Reasoning},
  editor    = {Nikolaj Bjorner and Marijn Heule and Andrei Voronkov},
  series    = {EPiC Series in Computing},
  volume    = {100},
  pages     = {19--35},
  year      = {2024},
  publisher = {EasyChair},
  bibsource = {EasyChair, https://easychair.org},
  issn      = {2398-7340},
  url       = {https://easychair.org/publications/paper/9KKC},
  doi       = {10.29007/6qkh}}

@article{OppenSI,
title = {Complexity, convexity and combinations of theories},
journal = {Theoretical Computer Science},
volume = {12},
number = {3},
pages = {291-302},
year = {1980},
issn = {0304-3975},
doi = {https://doi.org/10.1016/0304-3975(80)90059-6},
author = {Derek C. Oppen}
}

@InProceedings{CADE,
author="Toledo, Guilherme V.
and Zohar, Yoni
and Barrett, Clark",
editor="Pientka, Brigitte
and Tinelli, Cesare",
title="Combining Combination Properties: An Analysis of Stable Infiniteness, Convexity, and Politeness",
booktitle="Automated Deduction -- CADE 29",
year="2023",
publisher="Springer Nature Switzerland",
address="Cham",
pages="522--541"
}

@inproceedings{CasalRasga,
  author    = {Filipe Casal and
               Jo{\~{a}}o Rasga},
  title     = {Revisiting the Equivalence of Shininess and Politeness},
  booktitle = {{LPAR}},
  series    = {Lecture Notes in Computer Science},
  volume    = {8312},
  pages     = {198--212},
  publisher = {Springer},
  year      = {2013},
  doi={10.1007/978-3-642-45221-5_15}
}

@Article{CasalRasga2,
author={Casal, Filipe
and Rasga, Jo{\~a}o},
title={Many-Sorted Equivalence of Shiny and Strongly Polite Theories},
journal={Journal of Automated Reasoning},
year={2018},
month=Feb,
day={01},
volume={60},
number={2},
pages={221-236},
issn={1573-0670},
doi={10.1007/s10817-017-9411-y},
}

@book {marker2002,
    AUTHOR = {Marker, David},
     TITLE = {Model theory: {An} introduction},
    SERIES = {Graduate Texts in Mathematics},
    VOLUME = {217},
 PUBLISHER = {Springer-Verlag, New York},
      YEAR = {2002},
     PAGES = {viii+342},
      ISBN = {0-387-98760-6},
}

@inproceedings{nounicorns,
  author       = {Benjamin Przybocki and
                  Guilherme Vicentin de Toledo and
                  Yoni Zohar and
                  Clark W. Barrett},
  editor       = {Andr{\'{e}} Platzer and
                  Kristin Yvonne Rozier and
                  Matteo Pradella and
                  Matteo Rossi},
  title        = {The Nonexistence of Unicorns and Many-Sorted L{\"{o}}wenheim-Skolem
                  Theorems},
  booktitle    = {Formal Methods - 26th International Symposium, {FM} 2024, Milan, Italy,
                  September 9-13, 2024, Proceedings, Part {I}},
  series       = {Lecture Notes in Computer Science},
  volume       = {14933},
  pages        = {658--675},
  publisher    = {Springer},
  year         = {2024},
  url          = {https://doi.org/10.1007/978-3-031-71162-6_34},
  doi          = {10.1007/978-3-031-71162-6_34},
  bibsource    = {dblp computer science bibliography, https://dblp.org}
}

@InProceedings{Bonacina,
author="Bonacina, Maria Paola
and Ghilardi, Silvio
and Nicolini, Enrica
and Ranise, Silvio
and Zucchelli, Daniele",
editor="Furbach, Ulrich
and Shankar, Natarajan",
title="Decidability and Undecidability Results for Nelson-Oppen and Rewrite-Based Decision Procedures",
booktitle="Automated Reasoning",
year="2006",
publisher="Springer Berlin Heidelberg",
address="Berlin, Heidelberg",
pages="513--527",
isbn="978-3-540-37188-5"
}

@InProceedings{gentle,
author="Fontaine, Pascal",
editor="Ghilardi, Silvio
and Sebastiani, Roberto",
title="Combinations of Theories for Decidable Fragments of First-Order Logic",
booktitle="Frontiers of Combining Systems",
year="2009",
publisher="Springer Berlin Heidelberg",
address="Berlin, Heidelberg",
pages="263--278",
isbn="978-3-642-04222-5"
}

@InProceedings{polite,
    TITLE = {Combining data structures with nonstably infinite theories using many-sorted logic},
    AUTHOR = {Ranise, Silvio and Ringeissen, Christophe and Zarba, Calogero G.},
    URL = {https://hal.inria.fr/inria-00000570},
    BOOKTITLE = {5th International Workshop on Frontiers of Combining Systems - FroCoS'05},
    ADDRESS = {Vienna},
    EDITOR = {Bernard Gramlich},
    PUBLISHER = {Springer},
    SERIES = {Lecture Notes in Artificial Intelligence},
    VOLUME = {3717},
    PAGES = {48--64},
    YEAR = {2005},
    MONTH = Sep,
    DOI = {10.1007/11559306}
}

@MISC{SMTLIB,
  author =	 {Clark Barrett and Pascal Fontaine and Cesare Tinelli},
  title =	 {{The Satisfiability Modulo Theories Library (SMT-LIB)}},
  howpublished = {{\tt www.SMT-LIB.org}},
  year =	 2016,
}

@inproceedings{cvc5,
  author    = {Haniel Barbosa and
               Clark W. Barrett and
               Martin Brain and
               Gereon Kremer and
               Hanna Lachnitt and
               Makai Mann and
               Abdalrhman Mohamed and
               Mudathir Mohamed and
               Aina Niemetz and
               Andres N{\"{o}}tzli and
               Alex Ozdemir and
               Mathias Preiner and
               Andrew Reynolds and
               Ying Sheng and
               Cesare Tinelli and
               Yoni Zohar},
  title     = {cvc5: {A} Versatile and Industrial-Strength {SMT} Solver},
  booktitle = {{TACAS} {(1)}},
  series    = {Lecture Notes in Computer Science},
  volume    = {13243},
  pages     = {415--442},
  publisher = {Springer},
  year      = {2022}
}

@inproceedings{kovacs2017coming,
  title={Coming to terms with quantified reasoning},
  author={Kov{\'a}cs, Laura and Robillard, Simon and Voronkov, Andrei},
  booktitle={Proceedings of the 44th ACM SIGPLAN Symposium on Principles of Programming Languages},
  pages={260--270},
  year={2017}
}

@article{chen2022solving,
  title={Solving string constraints with regex-dependent functions through transducers with priorities and variables},
  author={Chen, Taolue and Flores-Lamas, Alejandro and Hague, Matthew and Han, Zhilei and Hu, Denghang and Kan, Shuanglong and Lin, Anthony W and R{\"u}mmer, Philipp and Wu, Zhilin},
  journal={Proceedings of the ACM on Programming Languages},
  volume={6},
  number={POPL},
  pages={1--31},
  year={2022},
  publisher={ACM New York, NY, USA}
}

@inproceedings{brain2019building,
  title={Building better bit-blasting for floating-point problems},
  author={Brain, Martin and Schanda, Florian and Sun, Youcheng},
  booktitle={Tools and Algorithms for the Construction and Analysis of Systems: 25th International Conference, TACAS 2019, Held as Part of the European Joint Conferences on Theory and Practice of Software, ETAPS 2019, Prague, Czech Republic, April 6--11, 2019, Proceedings, Part I 25},
  pages={79--98},
  year={2019},
  organization={Springer}
}

@book{DBLP:series/txtcs/KroeningS16,
  author       = {Daniel Kroening and
                  Ofer Strichman},
  title        = {Decision Procedures - An Algorithmic Point of View, Second Edition},
  series       = {Texts in Theoretical Computer Science. An {EATCS} Series},
  publisher    = {Springer},
  year         = {2016}
}

@article{NelsonOppen,
author = {Nelson, Greg and Oppen, Derek C.},
title = {Simplification by Cooperating Decision Procedures},
year = {1979},
issue_date = {Oct. 1979},
publisher = {Association for Computing Machinery},
address = {New York, NY, USA},
volume = {1},
number = {2},
issn = {0164-0925},
url = {https://doi.org/10.1145/357073.357079},
doi = {10.1145/357073.357079},
journal = {ACM Trans. Program. Lang. Syst.},
month = oct,
pages = {245–257},
numpages = {13}
}

@incollection{BSST21,
   author = {Clark Barrett and Roberto Sebastiani and Sanjit Seshia and
	Cesare Tinelli},
   editor = {Armin Biere and Marijn J. H. Heule and Hans van Maaren and
	Toby Walsh},
   title = {Satisfiability Modulo Theories},
   booktitle = {Handbook of Satisfiability, Second Edition},
   series = {Frontiers in Artificial Intelligence and Applications},
   volume = {336},
   chapter = {33},
   pages = {825--885},
   publisher = {IOS Press},
   month = feb,
   year = {2021},
   url = {http://www.cs.stanford.edu/~barrett/pubs/BSST21.pdf}
}

@inproceedings{JB10-LPAR,
  author       = {Dejan Jovanovi{\'c} and
                  Clark Barrett},
  title        = {Polite Theories Revisited},
  booktitle    = {{LPAR} (Yogyakarta)},
  series       = {Lecture Notes in Computer Science},
  volume       = {6397},
  pages        = {402--416},
  publisher    = {Springer},
  year         = {2010}
}

@book{Jech,
  doi = {10.1007/3-540-44761-x},
  url = {https://doi.org/10.1007/3-540-44761-x},
  year = {2003},
  publisher = {Springer Berlin Heidelberg},
  title = {Set Theory},
author={Thomas Jech}
}

@inproceedings{DBLP:conf/frocos/BarrettDS02,
  author       = {Clark W. Barrett and
                  David L. Dill and
                  Aaron Stump},
  title        = {A Generalization of Shostak's Method for Combining Decision Procedures},
  booktitle    = {FroCoS},
  series       = {Lecture Notes in Computer Science},
  volume       = {2309},
  pages        = {132--146},
  publisher    = {Springer},
  year         = {2002}
}

@inproceedings{DBLP:conf/birthday/BonacinaFRT19,
  author       = {Maria Paola Bonacina and
                  Pascal Fontaine and
                  Christophe Ringeissen and
                  Cesare Tinelli},
  editor       = {Carsten Lutz and
                  Uli Sattler and
                  Cesare Tinelli and
                  Anni{-}Yasmin Turhan and
                  Frank Wolter},
  title        = {Theory Combination: Beyond Equality Sharing},
  booktitle    = {Description Logic, Theory Combination, and All That - Essays Dedicated
                  to Franz Baader on the Occasion of His 60th Birthday},
  series       = {Lecture Notes in Computer Science},
  volume       = {11560},
  pages        = {57--89},
  publisher    = {Springer},
  year         = {2019},
  url          = {https://doi.org/10.1007/978-3-030-22102-7_3},
  doi          = {10.1007/978-3-030-22102-7_3},
  bibsource    = {dblp computer science bibliography, https://dblp.org}
}

@article{shiny,
  author    = {Cesare Tinelli and
               Calogero G. Zarba},
  title     = {Combining Nonstably Infinite Theories},
  journal   = {J. Autom. Reason.},
  volume    = {34},
  number    = {3},
  pages     = {209--238},
  year      = {2005}
}

@article{barrett2007abstract,
  title={An abstract decision procedure for a theory of inductive data types},
  author={Barrett, Clark and Shikanian, Igor and Tinelli, Cesare},
  journal={Journal on Satisfiability, Boolean Modelling and Computation},
  volume={3},
  number={1-2},
  pages={21--46},
  year={2007},
  publisher={SAGE Publications Sage UK: London, England}
}

@article {polite-alg,
    AUTHOR = {Sheng, Ying and Zohar, Yoni and Ringeissen, Christophe and
              Lange, Jane and Fontaine, Pascal and Barrett, Clark},
     TITLE = {Polite combination of algebraic datatypes},
   JOURNAL = {J. Automat. Reason.},
  FJOURNAL = {Journal of Automated Reasoning},
    VOLUME = {66},
      YEAR = {2022},
    NUMBER = {3},
     PAGES = {331--355},
      ISSN = {0168-7433,1573-0670},
   MRCLASS = {68Q65 (68N18 68T20 68V15)},
  MRNUMBER = {4449707},
MRREVIEWER = {Rub\'en\ Rubio},
       DOI = {10.1007/s10817-022-09625-3},
       URL = {https://doi.org/10.1007/s10817-022-09625-3},
}

@article {sheng-si-polite,
    AUTHOR = {Sheng, Ying and Zohar, Yoni and Ringeissen, Christophe and
              Reynolds, Andrew and Barrett, Clark and Tinelli, Cesare},
     TITLE = {Combining stable infiniteness and (strong) politeness},
   JOURNAL = {J. Automat. Reason.},
  FJOURNAL = {Journal of Automated Reasoning},
    VOLUME = {67},
      YEAR = {2023},
    NUMBER = {4},
     PAGES = {Paper No. 34, 22},
      ISSN = {0168-7433,1573-0670},
   MRCLASS = {68V15 (03F25)},
  MRNUMBER = {4646413},
       DOI = {10.1007/s10817-023-09684-0},
       URL = {https://doi.org/10.1007/s10817-023-09684-0},
}

@Inbook{tinelli-new,
author="Tinelli, Cesare
and Harandi, Mehdi",
editor="Baader, Frans
and Schulz, Klaus U.",
title="A New Correctness Proof of the Nelson-Oppen Combination Procedure",
bookTitle="Frontiers of Combining Systems: First International Workshop, Munich, March 1996",
year="1996",
publisher="Springer Netherlands",
address="Dordrecht",
pages="103--119",
isbn="978-94-009-0349-4",
doi="10.1007/978-94-009-0349-4_5",
url="https://doi.org/10.1007/978-94-009-0349-4_5"
}

@article {galois,
    AUTHOR = {Ore, Oystein},
     TITLE = {Galois connexions},
   JOURNAL = {Trans. Amer. Math. Soc.},
  FJOURNAL = {Transactions of the American Mathematical Society},
    VOLUME = {55},
      YEAR = {1944},
     PAGES = {493--513},
      ISSN = {0002-9947,1088-6850},
       DOI = {10.2307/1990305},
       URL = {https://doi.org/10.2307/1990305},
}

@book {lattice-thy,
    AUTHOR = {Birkhoff, Garrett},
     TITLE = {Lattice theory},
    SERIES = {American Mathematical Society Colloquium Publications},
    VOLUME = {Vol. 25},
   EDITION = {third},
 PUBLISHER = {American Mathematical Society, Providence, RI},
      YEAR = {1979},
     PAGES = {vi+418},
      ISBN = {0-8218-1025-1},
}

@article {tarski-fixed-point,
    AUTHOR = {Tarski, Alfred},
     TITLE = {A lattice-theoretical fixpoint theorem and its applications},
   JOURNAL = {Pacific J. Math.},
  FJOURNAL = {Pacific Journal of Mathematics},
    VOLUME = {5},
      YEAR = {1955},
     PAGES = {285--309},
      ISSN = {0030-8730,1945-5844},
       URL = {http://projecteuclid.org/euclid.pjm/1103044538},
}

@article{knaster-fixed-point,
 author = {Knaster, Bronis{\l}aw},
 title = {Un th{\'e}or{\`e}me sur les fonctions d'ensembles.},
 fjournal = {Annales de la Soci{\'e}t{\'e} Polonaise de Math{\'e}matique},
 journal = {Ann. Soc. Polon. Math.},
 volume = {6},
 pages = {133--134},
 year = {1928},
}

@article {robinson-consistency,
    AUTHOR = {Robinson, Abraham},
     TITLE = {A result on consistency and its application to the theory of
              definition},
      NOTE = {Nederl. Akad. Wetensch. Proc. Ser. A {\bf 59}},
   JOURNAL = {Indag. Math.},
  FJOURNAL = {},
    VOLUME = {18},
      YEAR = {1956},
     PAGES = {47--58},
}

@inproceedings{manna-zarba-survey,
  author       = {Zohar Manna and
                  Calogero G. Zarba},
  editor       = {Bernhard K. Aichernig and
                  T. S. E. Maibaum},
  title        = {Combining Decision Procedures},
  booktitle    = {Formal Methods at the Crossroads. From Panacea to Foundational Support,
                  10th Anniversary Colloquium of UNU/IIST, the International Institute
                  for Software Technology of The United Nations University, Lisbon,
                  Portugal, March 18-20, 2002, Revised Papers},
  series       = {Lecture Notes in Computer Science},
  volume       = {2757},
  pages        = {381--422},
  publisher    = {Springer},
  year         = {2002},
  url          = {https://doi.org/10.1007/978-3-540-40007-3\_24},
  doi          = {10.1007/978-3-540-40007-3\_24},
  bibsource    = {dblp computer science bibliography, https://dblp.org}
}

@article {shostak,
    AUTHOR = {Shostak, Robert E.},
     TITLE = {Deciding combinations of theories},
   JOURNAL = {J. Assoc. Comput. Mach.},
  FJOURNAL = {Journal of the Association for Computing Machinery},
    VOLUME = {31},
      YEAR = {1984},
    NUMBER = {1},
     PAGES = {1--12},
      ISSN = {0004-5411,1557-735X},
       DOI = {10.1145/2422.322411},
       URL = {https://doi.org/10.1145/2422.322411},
}

@inproceedings{shostak-light,
  author       = {Harald Ganzinger},
  editor       = {Andrei Voronkov},
  title        = {Shostak Light},
  booktitle    = {Automated Deduction - CADE-18, 18th International Conference on Automated
                  Deduction, Copenhagen, Denmark, July 27-30, 2002, Proceedings},
  series       = {Lecture Notes in Computer Science},
  volume       = {2392},
  pages        = {332--346},
  publisher    = {Springer},
  year         = {2002},
  url          = {https://doi.org/10.1007/3-540-45620-1\_28},
  doi          = {10.1007/3-540-45620-1\_28},
  bibsource    = {dblp computer science bibliography, https://dblp.org}
}

@article {unions-non-disjoint,
    AUTHOR = {Tinelli, Cesare and Ringeissen, Christophe},
     TITLE = {Unions of non-disjoint theories and combinations of
              satisfiability procedures},
   JOURNAL = {Theoret. Comput. Sci.},
  FJOURNAL = {Theoretical Computer Science},
    VOLUME = {290},
      YEAR = {2003},
    NUMBER = {1},
     PAGES = {291--353},
      ISSN = {0304-3975,1879-2294},
       DOI = {10.1016/S0304-3975(01)00332-2},
       URL = {https://doi.org/10.1016/S0304-3975(01)00332-2},
}

@article {ghilardi-non-disjoint,
    AUTHOR = {Ghilardi, Silvio},
     TITLE = {Model-theoretic methods in combined constraint satisfiability},
   JOURNAL = {J. Automat. Reason.},
  FJOURNAL = {Journal of Automated Reasoning},
    VOLUME = {33},
      YEAR = {2004},
    NUMBER = {3-4},
     PAGES = {221--249},
      ISSN = {0168-7433,1573-0670},
       DOI = {10.1007/s10817-004-6241-5},
       URL = {https://doi.org/10.1007/s10817-004-6241-5},
}

@inproceedings{gentle-non-disjoint,
  author       = {Paula Daniela Chocron and
                  Pascal Fontaine and
                  Christophe Ringeissen},
  editor       = {St{\'{e}}phane Demri and
                  Deepak Kapur and
                  Christoph Weidenbach},
  title        = {A Gentle Non-disjoint Combination of Satisfiability Procedures},
  booktitle    = {Automated Reasoning - 7th International Joint Conference, {IJCAR}
                  2014, Held as Part of the Vienna Summer of Logic, {VSL} 2014, Vienna,
                  Austria, July 19-22, 2014. Proceedings},
  series       = {Lecture Notes in Computer Science},
  volume       = {8562},
  pages        = {122--136},
  publisher    = {Springer},
  year         = {2014},
  url          = {https://doi.org/10.1007/978-3-319-08587-6\_9},
  doi          = {10.1007/978-3-319-08587-6\_9},
  bibsource    = {dblp computer science bibliography, https://dblp.org}
}

@inproceedings{polite-non-disjoint,
  author       = {Paula Daniela Chocron and
                  Pascal Fontaine and
                  Christophe Ringeissen},
  editor       = {Amy P. Felty and
                  Aart Middeldorp},
  title        = {A Polite Non-Disjoint Combination Method: Theories with Bridging Functions
                  Revisited},
  booktitle    = {Automated Deduction - {CADE-25} - 25th International Conference on
                  Automated Deduction, Berlin, Germany, August 1-7, 2015, Proceedings},
  series       = {Lecture Notes in Computer Science},
  volume       = {9195},
  pages        = {419--433},
  publisher    = {Springer},
  year         = {2015},
  url          = {https://doi.org/10.1007/978-3-319-21401-6\_29},
  doi          = {10.1007/978-3-319-21401-6\_29},
  bibsource    = {dblp computer science bibliography, https://dblp.org}
}

\appendix

\section{Further Definitions}
$\T$ is \emph{finitely witnessable} when there exists a computable function $\wit:QF(\Sigma)\rightarrow QF(\Sigma)$, called a \emph{witness}, satisfying for every quantifier-free formula $\varphi$:
that $\varphi$ and $\Exists{\overarrow{x}}\wit(\varphi)$ are $\T$-equivalent, where $\overarrow{x}=\vars(\wit(\varphi))\setminus\vars(\varphi)$;
and, if $\wit(\varphi)$ is $\T$-satisfiable, there exists a $\T$-interpretation $\A$ with $\A\vDash\wit(\varphi)$ and $\dom{\A}=\vars(\wit(\varphi))^{\A}$ \cite{polite}.
$\T$ is \emph{strongly finitely witnessable} when it has a witness that, in addition to the above properties, satisfies that, for any finite set of variables $V$, arrangement $\delta_{V}$ on $V$ and quantifier-free formula $\varphi$, if $\wit(\varphi)\wedge\delta_{V}$ is $\T$-satisfiable, then there exists a $\T$-interpretation $\A$ with $\A\vDash\wit(\varphi)\wedge\delta_{V}$ and $\dom{\A}=\vars(\wit(\varphi)\wedge\delta_{V})^{\A}$ \cite{JB10-LPAR}.
A theory is \emph{polite} if it is finitely witnessable and smooth;
it is \emph{strongly polite} if, in addition, it is strongly finitely witnessable.\footnote{Both the original notion from \cite{polite} and its revision in \cite{JB10-LPAR} were called {\em polite} where they were studied. However, since \cite{CasalRasga2} the term {\em polite} refers to the original while the term {\em strongly polite} refers to the revision. We follow this convention in this paper.}

$\Teq$ is the $\Sigma_{1}$-theory with no axioms, which is equivalent to it being axiomatized by all tautologies:
it is the logic of only equalities and disequalities.
\cite{shiny} proves that it is shiny, and of course it is decidable.

\section{The ``Only if'' Direction of Fontaine's Lemma} \label{appendix-fontaine}

We recall the statement of Fontaine's lemma:
\lemfontaine*

We prove the ``only if'' direction of this lemma, since Fontaine~\cite{gentle} only proved the ``if'' direction. We use the following result:
\begin{lemma}[{\cite[Proposition~3.1]{tinelli-new}}] \label{lem-tinelli}
    Let $\T_1$ and $\T_2$ be theories. Then, their disjoint union $\T_1 \sqcup \T_2$ is consistent if and only if $\spec_{\T_1}(\top) \cap \spec_{\T_2}(\top) \neq \emptyset$.
\end{lemma}

Now, we can prove the ``only if'' direction of Fontaine's lemma.
\begin{lemma}
    Let $\T_1$ and $\T_2$ be theories over the signatures $\Sigma_1$ and $\Sigma_2$ respectively. If the disjoint combination $\T_1 \sqcup \T_2$ is decidable, then the following problem is decidable: given conjunctions of literals $\varphi_1$ and $\varphi_2$ over $\Sigma_1$ and $\Sigma_2$ respectively, determine whether $\spec_{\T_1}(\varphi_1) \cap \spec_{\T_2}(\varphi_2) = \emptyset$.
\end{lemma}
\begin{proof}
    Let $\varphi_1$ and $\varphi_2$ be conjunctions of literals. We may assume that the variables in $\varphi_1$ and $\varphi_2$ are disjoint by renaming variables if necessary. For $i \in \{1,2\}$, let $\T'_i$ be the theory axiomatized by $\T_i \cup \{\Exists{\vec{w_i}} \varphi_i\}$, where $\vec{w_i} = \vars(\varphi_i)$. Then, the following are equivalent:
    \begin{itemize}
        \item $\spec_{\T_1}(\varphi_1) \cap \spec_{\T_2}(\varphi_2) \neq \emptyset$;
        \item $\spec_{\T'_1}(\top) \cap \spec_{\T'_2}(\top) \neq \emptyset$;
        \item $\T'_1 \sqcup \T'_2$ is consistent (this equivalence is by \Cref{lem-tinelli});
        \item $\Exists{\vec{w_1}} \varphi_1 \land \Exists{\vec{w_2}} \varphi_2$ is $\T_1 \sqcup \T_2$-satisfiable; and
        \item $\varphi_1 \land \varphi_2$ is $\T_1 \sqcup \T_2$-satisfiable.
    \end{itemize}
    Thus, we can effectively decide whether $\spec_{\T_1}(\varphi_1) \cap \spec_{\T_2}(\varphi_2) \neq \emptyset$ by checking if $\varphi_1 \land \varphi_2$ is $\T_1 \sqcup \T_2$-satisfiable.
\end{proof}

\section{Test Theories}

Here, we prove the assertions about properties of test theories, namely those stated in \Cref{theo:testtheories,lem-n-shiny-test,lem-qg-test}.

\subsection{\tp{Proofs for $\Testsm$}{Proofs for the Theories That Test Smoothness}}\label{sec:testsm}

\begin{lemma}
    $\Testsm$ is decidable and infinitely decidable.
\end{lemma}

\begin{proof}
\begin{enumerate} 
\item Let $\varphi$ be a conjunction of literals;
    a $P$-literal is a literal of the form $P_{k}$ or $\neg P_{k}$;
    $\varphi^{\prime}$ is the conjunction of all literals in $\varphi$ that are not $P$-literals.
    Of course, if $\varphi$ contains both the literals $P_{k}$ and $\neg P_{k}$, or $P_{i}$ and $P_{j}$ for $i\neq j$, it is not $\Testsm$-satisfiable;
    if that does not happen, we state that $\varphi$ is $\Testsm$-satisfiable if and only if $\minmod_{\Teq}(\varphi^{\prime})\leq n$, which implies the theory is decidable.

    One direction is quite clear, so assume $\minmod_{\Teq}(\varphi^{\prime})\leq n$, and take a $\Teq$-interpretation $\A$ that satisfies $\varphi^{\prime}$ with $|\dom{\A}|=n$ (if $\minmod_{\Teq}(\varphi^{\prime})<n$ we must use the fact $\Teq$ is smooth).
    We define a $\Spn$-interpretation $\B$ starting from $\A$ by making $P_{k}^{\A}$ true if and only if $P_{k}$ is a literal of $\varphi$ (we are assuming there is at most one), meaning $\B$ satisfies $\varphi$;
    furthermore, $|\dom{\B}|=n$, meaning both $\psi_{=m}\vee\psi_{=n}$ and $P_{k}\rightarrow\psi_{=n}$, for $k\in \unc$, are satisfied in $\B$, and thus is a $\Testsm$-interpretation.

    \item $\Testsm$ is clearly infinitely decidable as it contains no infinite interpretations. \qedhere
    \end{enumerate}
\end{proof}

\begin{lemma}
    If $k,k' \in \No$, $k < k'$, and $k \neq n$, then $\T^{k'}_{k}$ is $n$-decidable.
\end{lemma}
\begin{proof}
    If $n \neq k'$, then $\T^{k'}_{k}$ has no models of size $n$, so it is trivially $n$-decidable. Otherwise, if $n = k'$, then every $\T^{k'}_{k}$-satisfiable quantifier-free formula has a model of cardinality $n$, and since the theory is decidable we get that it is $n$-decidable.
\end{proof}

\subsection{\tp{Proofs for $\Testg$}{Proofs for the Theory That Tests Gentleness}}\label{sec:prooftestg}

\begin{lemma}
    Let $S$ be an infinite decidable set. Then, $\Testg$ is decidable, has computable finite spectra, and is $n$-decidable.
\end{lemma}

\begin{proof}
\begin{enumerate}
    \item Let $\varphi$ be a conjunction of literals, and $\varphi^{\prime}$ be the result of removing all $P$-literals from $\varphi$.
    Of course, if $\varphi$ contains both $P_{n}$ and $\neg P_{n}$, or $P_{m}$ and $P_{n}$ for $m\neq n$, it is not $\Testg$-satisfiable, so we may assume $\varphi$ has at most one positive $P$-literal.

    If $\varphi$ contains no positive $P$-literal at all, we state $\varphi$ is $\Testg$-satisfiable if and only if $\varphi^{\prime}$ is $\Teq$-satisfiable:
    indeed, given any $\Teq$-interpretation $\A$ with cardinality in $S$ satisfying $|\dom{\A}|\geq \minmod_{\Teq}(\varphi^{\prime})$ (we can always find one as $\Teq$ is smooth), we make it into a $\Testg$-interpretation by making all $P_{n}$ false, and of course it still satisfies $\varphi^{\prime}$.
    Notice this also implies 
    \[\spec_{\Testg}(\varphi)=[\minmod_{\Teq}(\varphi^{\prime}),\aleph_{0}]\cap  (S \cup \{\aleph_0\}):\]
    of course the problem of whether $i\in \spec_{\Testg}(\varphi)$, for any $i\in \No$, is then decidable.

    If $\varphi$ contains the literal $P_{n}$, we state that $\varphi$ is $\Testg$-satisfiable if and only if \[[\minmod_{\Teq}(\varphi^{\prime}), F(n)]\cap (S \cup \{\aleph_0\})\neq\emptyset:\]
    we can check if $i \in \No$ is in that set by first checking if $\minmod_{\Teq}(\varphi^{\prime})\leq i$, if $i\leq F(n)$ (what can be done algorithmically by the definition of $F$), and then if $i\in S$ (which is a decidable set).
    Indeed, if $\A$ is a $\Teq$-interpretation that satisfies $\varphi$ with $|\dom{\A}|\in [\minmod_{\Teq}(\varphi^{\prime}), F(n)]\cap (S \cup \{\aleph_0\})$ (remember $\Teq$ is smooth, so we have interpretations satisfying $\varphi$ for all cardinalities greater than or equal to $\minmod_{\Teq}(\varphi^{\prime})$), we turn it into a $\Testg$-interpretation, that still satisfies $\varphi$, by making $P_{n}$ true in it, and all other $P_{m}$, for $m\neq n$, false.
    It is clear that in this case $\spec_{\Testg}(\varphi) = [\minmod_{\Teq}(\varphi^{\prime}), F(n)]\cap (S \cup \{\aleph_0\})$. Since $S$ is infinite, $\varphi$ is $\Testg$-satisfiable if and only if it has a finite model, so we can decide whether $\varphi$ is $\Testg$-satisfiable by checking whether $i \in \No$ is in $\spec_{\Testg}(\varphi)$ for increasing values of $i$ until either we find an $i \in \spec_{\Testg}(\varphi)$ (in which case $\varphi$ is $\Testg$-satisfiable) or we find that $i > F(n)$ (in which case $\varphi$ is $\Testg$-unsatisfiable).
    Thus we have proved not only that $\Testg$ is decidable, but also that it has computable finite spectra.

    \item $\Testg$ was proved to have computable finite spectra in the previous item.

    \item That $\Testg$ is $n$-decidable follows from the fact that it has computable finite spectra. \qedhere
    \end{enumerate}
\end{proof}

\begin{lemma}
    Let $S$ be an infinite decidable set. If $S \in \mathfrak{F}$, then $\Testg$ is \quagen{}. And if $\mathbb{N} \setminus S \notin \mathfrak{F}$, then $\Testg$ is \coquagen{}.
\end{lemma}
\begin{proof}
    \begin{enumerate}
        \item Suppose $S \in \mathfrak{F}$. As proved above, the spectra of cubes in $\Testg$ is of the form $[\minmod_{\Teq}(\varphi^{\prime}), F(n)]\cap (S \cup \{\aleph_0\})$, and thus either a finite set of finite cardinalities or $[\minmod_{\Teq}(\varphi^{\prime}), \aleph_0]\cap (S \cup \{\aleph_0\})$, whose finite part is in $\mathfrak{F}$ as this is a free filter. Since $\Testg$ has computable finite spectra we reach the conclusion that it is \quagen{}.

        \item Essentially identical reasoning shows that if $\mathbb{N} \setminus S \notin \mathfrak{F}$, then $\Testg$ is \coquagen{}. \qedhere
    \end{enumerate}
\end{proof}

\subsection{\tp{Proofs for $\Testcfs$}{Proofs for the Theory That Tests Computable Finite Spectra}}

\begin{lemma}
    $\Testcfs$ is decidable, gentle (and thus $\mathfrak{F}$-QG and co-$\mathfrak{F}$-QG for every free filter $\mathfrak{F}$), has computable finite spectra, is infinitely decidable (and thus has computable spectra), and  $n$-decidable.
\end{lemma}

\begin{proof}
\begin{enumerate}
    \item Decidability will follow from gentleness.
    \item To prove $\Testcfs$ is gentle we will show that the spectra of cubes is either finite and can be output computably, or cofinite and its complement can be output computably:
    since any quantifier-free formula is equivalent to a disjunction of cubes, and the spectrum of a disjunction is the union of the spectra of the components, we will be then done.
        
    Let $\varphi$ be a conjunction of literals, and $\varphi^{\prime}$ be the conjunction of all literals in $\varphi$ that are not $P$-literals.
    If $\varphi$ contains both $P_{n}$ and $\neg P_{n}$, or $P_{m}$ and $P_{n}$ for $m\neq n$, we already know it is not $\Testcfs$-satisfiable, so assume $\varphi$ has at most one positive $P$-literal.
        
    If it contains none we state $\spec_{\Testcfs}(\varphi)=[\minmod_{\Teq}(\varphi^{\prime}),\aleph_{0}]$ (if $\varphi^{\prime}$ is $\Teq$-satisfiable, otherwise of course $\varphi$ is not $\Testcfs$-satisfiable), which is cofinite and whose complement is computable:
    one direction is obvious as any $\Testcfs$-interpretation is a $\Teq$-interpretation.
    Conversely, if $n\in[\minmod_{\Teq}(\varphi^{\prime}),\aleph_{0}]$, as $\Teq$ is smooth we can find a $\Teq$-interpretation $\A$ that satisfies $\varphi^{\prime}$;
    we make it into a $\Testcfs$-interpretation by making all $P_{n}$ false, so all axioms are vacuously satisfied.
        
    If $\varphi$ contains as its only positive $P$-literal $P_{n}$ we state that $\spec_{\Testcfs}(\varphi)=\{n\}$ if $\minmod_{\Teq}(\varphi^{\prime})\leq n$ (otherwise $\varphi$ is clearly not $\Testcfs$-satisfiable).
    Indeed, if $\minmod_{\Teq}(\varphi^{\prime})\leq n$, by using the smoothness of $\Teq$ we can get a $\Teq$-interpretation $\B$ that satisfies $\varphi^{\prime}$ with $|\dom{\B}|=n$;
    we then turn it into a $\Testcfs$-interpretation by making $P_{n}$ true and all $P_{i}$ for $i\neq n$ false, and of course the resulting interpretation will satisfy $\varphi$.
    \item $\Testcfs$ is $\mathfrak{F}$-QG since it is gentle (see \Cref{prop-inclusion}).
    \item $\Testcfs$ is co-$\mathfrak{F}$-QG since it is gentle.
    \item $\Testcfs$ has computable finite spectra since it is gentle.
    \item $\Testcfs$ is infinitely decidable since it is gentle.
    \item $\Testcfs$ is $n$-decidable since it has computable finite spectra. \qedhere
    \end{enumerate}
\end{proof}

\subsection{\tp{Proofs for $\Tinfty$}{Proofs for a Theory That Tests for Infinite Models}}

\begin{lemma}
    $\Tinfty$ is decidable, has computable finite spectra, is infinitely decidable (and thus has computable spectra), and is smooth.
\end{lemma}

\begin{proof}
    \begin{enumerate}
        \item \cite{FroCoS2025} proves $\Tinfty$, as a theory over an empty signature with finitely many sorts, is decidable.
        \item $\Tinfty$ has computable finite spectra as it contains no finite interpretations.
        \item $\Tinfty$ is infinitely decidable, and therefore has computable spectra, as $\spec_{\Tinfty}(\varphi)=\{\aleph_{0}\}$ for all quantifier-free $\Tinfty$-satisfiable formula $\varphi$ by \Cref{LowenheimSkolem}.
        \item \cite{CADE} proved that $\Tinfty$ is smooth. \qedhere
    \end{enumerate}
\end{proof}

\subsection{\tp{Proofs for $\Testsi$}{Proofs for the Theory That Tests Stable Infiniteness}}

\begin{lemma}
    $\Testsi$ is decidable, smooth (and thus stably infinite), and finitely witnessable (and thus polite).
\end{lemma}

\begin{proof}
\begin{enumerate}
    \item $\varphi^{\prime}$ is again obtained from a quantifier-free formula $\varphi$ by removing all of its $P$-literals.
    We can then state that $\varphi$ is $\Testsi$-satisfiable if and only if $\varphi^{\prime}$ is $\Teq$-satisfiable (and $\varphi$ does not contain $P_{i}$ and $\neg P_{i}$, or $P_{i}$ and $P_{j}$ for $i\neq j$, of course);
    one direction of this equivalence is obvious, so focus on the non-trivial one.
    Given $\Teq$ is smooth, we take an infinite $\Teq$-interpretation $\A$ that satisfies $\varphi^{\prime}$;
    we make it into a $\Testsi$-interpretation $\B$ by making $P_{i}^{\B}$ true only if $P_{i}$ is a positive literal in $\varphi$;
    $\B$ then indeed satisfies $\psi_{\geq n+1}$, of course, and $\varphi$.
    \item Now we prove $\Testsi$ is smooth, which also implies it is stably infinite.
    Take a quantifier-free formula $\varphi$, a $\Testsi$-interpretation $\A$ that satisfies $\varphi$, and a $\kappa\geq |\dom{\A}|$;
    let $X$ be a set of cardinality $\kappa-|\dom{\A}|$ disjoint from $\dom{\A}$.
    We define an interpretation $\B$ by making $\dom{\B}=\dom{\A}\cup X$ (so we have $|\dom{\B}|=\kappa$);
    $P_{i}^{\B}=P_{i}^{\A}$ for all $i$ (so if $P_{m}$ is true for an $m\in \unc$, $\A$ satisfies $\psi_{\geq n+1}$, and so does $\B$ as $|\dom{\B}|\geq|\dom{\A}|$, meaning $\B$ is a $\Testsi$-interpretation);
    and $x^{\B}=x^{\A}$ for all variables $x$, meaning $\B$ satisfies $\varphi$.
    \item To show that $\Testsi$ is finitely witnessable, consider the function $\wit$ such that, given a quantifier-free formula $\varphi$ with a single positive $P$-literal $P_n$, 
    \[\wit(\varphi)=\varphi\wedge\bigwedge_{i=1}^{n+1}x_{i}=x_{i},\]
    where $x_{1}$ through $x_{n+1}$ are fresh variables.
    Of course $\varphi$ and $\Exists{\overarrow{x}}\wit(\varphi)$ are $\Testsi$-equivalent, as $\varphi$ and $\wit(\varphi)$ are themselves already $\Teq$-equivalent.
    Now take a $\Testsi$-interpretation $\A$ that satisfies $\varphi$:
    if $|\vars(\varphi)^{\A}|<n+1$, take a set $X$ disjoint from $\dom{\A}$ with $|X|=n+1-|\vars(\varphi)^{\A}|$;
    otherwise, let $X=\emptyset$.
    We define an interpretation $\B$ by making:
    $\dom{\B}=\vars(\varphi)^{\A}\cup X$ (so $|\dom{\B}|\geq n+1$);
    $P_{m}^{\B}=P_{m}^{\A}$ (so in particular $\B$ is a $\Testsi$-interpretation);
    $x^{\B}=x^{\A}$ for every $x\in\vars(\varphi)$ (so $\B$ satisfies $\varphi$); and let $x_i^\B$ be such that
    $\{x_{1},\ldots,x_{n+1}\}\to \vars(\varphi)^{\A} \cup X$ is a surjective map if $X$ is not empty (so $\vars(\wit(\varphi))^{\B}=\dom{\B}$),
    and let $x_i^\B$ be chosen arbitrarily otherwise. \qedhere
\end{enumerate}
\end{proof}

\subsection{\tp{Proofs for $\Teqn$}{Proofs for a Theory That Tests n-Decidability}}

\begin{lemma}
    $\Teqn$ is decidable and $n$-shiny.
\end{lemma}

\begin{proof}
    \begin{enumerate}
        \item That $\Teqn$ is decidable follows from the fact it is defined over an empty signature~\cite{FroCoS2025}.
        \item That $\Teqn$ is $n$-shiny follows from the facts that this theory is decidable, and has only interpretations of cardinality $n$. \qedhere
    \end{enumerate}
\end{proof}

\subsection{\tp{Proofs for $Th_{\T}$}{Proofs for the Theory That, With the Minimal Model Function, Calculates the Spectra}}

\begin{lemma}
    Let $\T$ be a theory with computable finite spectra. Then, $Th_{\T}$ is decidable and has computable finite spectra.
\end{lemma}

\begin{proof}
    We define $P$-literals, in this case, as those of the form $P_{\varphi,n}$ or $\neg P_{\varphi,n}$.
    As before, since $P_{\varphi,n}\rightarrow\neg P_{\varphi^{\prime},n^{\prime}}$ for $(\varphi,n)\neq (\varphi^{\prime},n^{\prime})$, every conjunction of literals containing $P_{\varphi,n}$ and $\neg P_{\varphi,n}$, or $P_{\varphi,n}$ and $P_{\varphi^{\prime},n^{\prime}}$ for $(\varphi,n)\neq (\varphi^{\prime},n^{\prime})$, will immediately be $Th_{\T}$-unsatisfiable.

    So assume $\phi$ is a conjunction of literals containing at most one positive $P$-literal $P_{\varphi,n}$, and let $\phi'$ be the formula without this $P$-literal (if there is one).
    If $\phi$ actually contains no positive $P$-literals, it is $Th_{\T}$-satisfiable if and only if $\phi^{\prime}$ is $\Teq$-satisfiable, and $\spec_{Th_{\T}}(\phi)=[\minmod_{\Teq}(\phi^{\prime}),\aleph_{0}]$.
    Indeed, take a $\Teq$-interpretation that satisfies $\phi^{\prime}$ of any cardinality in this interval, and turn it into a $Th_{\T}$-interpretation by making all $P_{\varphi,n}$ false.

    If $\phi$ contains the literal $P_{\varphi,n}$, we recursively check the elements of $\spec_{\T}(\varphi) \cap \No$ by using that this set is decidable:
    while $m<\minmod_{\T}(\phi^{\prime})$ (we may just as well assume $\phi^{\prime}$ is $\Teq$-satisfiable) we test whether $m$ is the $n$th element of $\No\setminus\spec_{\T}(\varphi)$, and if it is then $\phi$ is $Th_{\T}$-unsatisfiable;
    if not, we increase $m$.
    When $m=\minmod_{\T}(\phi^{\prime})$ we know $s_{\varphi,n}\geq \minmod_{\T}(\phi^{\prime})$, regardless of whether $\spec_{\T}(\varphi) \cap \No$ has elements greater than or equal to $m$, and thus $\phi$ is $Th_{\T}$-satisfiable.
    We of course then have $\spec_{Th_{\T}}(\phi)=\{s_{\varphi,n}\}$:
    we can then decide whether an $i \in \No$ is in this spectrum by checking whether $s_{\varphi,n}=i$;
    that is, we test every cardinality in $[1,i]$ and list the ones that are not in $\spec_{\T}(\varphi)$, seeing if $i$ happens to be the $n$th one.
\end{proof}

\section{Complete Test Theories}

\subsection{Proof of Theorem~\ref{thm-shiny-uniform}}

\thmshinyuniform*
\begin{proof}
    The idea is to create a single decidable theory $\T'$ that plays the role of each of the test theories $\{\Tle^{\No},\Testcfs\} \cup \{\Testsm : m,n \in \No,\; m<n\}$ simultaneously. Let $\T'$ have predicates $P_n$ and $Q_n$ for each $n \in \No$ and predicates $R_{m,n,k}$ for each $m,n,k \in \No$ with $m < n$. Then, let $\T'$ be axiomatized by
    \begin{align*}
        &\{P_{n}\rightarrow\psi_{=n} : n\in\No\}\cup \\
        &\{Q_{n}\rightarrow\psi_{\leq F(n)} : F(n)\in\mathbb{N}\}\cup\{Q_{i}\rightarrow\neg Q_{j} : i\neq j\} \cup \\
        &\{R_{m,n,k} \rightarrow \psi_{=m} \lor \psi_{=n} : m,n,k \in \No,\; m<n\} \cup \{R_{m,n,k}\rightarrow \psi_{=n} : m,n,k \in \No,\; m<n,\; k \in \unc\} \cup \\
        &\{R_{i,j,k} \rightarrow  \lnot R_{i',j',k'} : (i,j,k) \neq (i',j',k')\} \cup \\
        &\{P_m \rightarrow \lnot Q_n : m,n \in \No\} \cup \\
        &\{P_\ell \rightarrow \lnot R_{m,n,k} : \ell,m,n,k \in \No,\; m < n\} \cup \\
        &\{Q_\ell \rightarrow \lnot R_{m,n,k} : \ell,m,n,k \in \No,\; m < n\},
    \end{align*}
    where $F$ and $\unc$ are as in the definition of $\Tle^{\No}$ and $\Testsm$ respectively (in particular, $\unc$ is undecidable).

    We claim that $\T'$ is decidable. Let $\varphi$ be a conjunction of literals, and suppose first that $\varphi$ contains a positive literal of the form $P_n$. Then, either $\varphi$ contains a positive literal of the form $Q_m$ or $R_{i,j,k}$, in which case it is $\T'$-unsatisfiable, or it does not, in which case we reduce to the decision problem for $\Testcfs$. Similarly, if $\varphi$ contains a positive literal of the form $Q_m$ or $R_{i,j,k}$, then we reduce to the decision problems for $\Tle^{\No}$ and $\Testsm$ respectively.

    At this point, the rest of the argument is as in the proof of \Cref{thm-shiny-tight}.
\end{proof}

\subsection{Proof of Theorem~\ref{thm-si-complete}}

\thmsicomplete*
\begin{proof}
    The idea is to create a single decidable theory $\T'$ that plays the role of each of the test theories $\{\Testsi : n \in \No\}$ simultaneously. Let $\T'$ have predicates $P_{n,k}$ for each $n,k \in \No$, and let it be axiomatized by
    \[
        \{P_{n,k}\rightarrow \psi_{\geq n+1} : k\in \unc\}\cup\{P_{i,j}\rightarrow\neg P_{i',j'} : (i,j)\neq (i',j')\}.
    \]
    The point is that, rather than having a different theory $\Testsi$ for each $n$, we can choose $n$ using our predicate $P_{n,k}$.

    The theory $\T'$ is decidable and stably infinite in virtue of the corresponding results for $\Testsi$, and we now proceed as in the proof of \Cref{thm-f-si}.
\end{proof}

\subsection{Proof of Theorem~\ref{thm-id-complete}}

\thmidcomplete*
\begin{proof}
    The idea is to create a single decidable theory $\T'$ that plays the role of each of the test theories $\{\Testcfs\} \cup \{\Testsm : m,n \in \No,\; m<n\}$ simultaneously. We proceed similarly to the proof of \Cref{thm-shiny-uniform}. Let $\T'$ have predicates $P_n$ for each $n \in \No$ and $Q_{m,n,k}$ for each $m,n,k \in \No$ with $m<n$. Then, let $\T'$ be axiomatized by
    \begin{align*}
        &\{P_{n}\rightarrow\psi_{=n} : n\in\No\}\cup \\
        &\{Q_{m,n,k} \rightarrow \psi_{=m} \lor \psi_{=n} : m,n,k \in \No,\; m<n\} \cup \{Q_{m,n,k}\rightarrow \psi_{=n} : m,n,k \in \No,\; m<n,\; k \in \unc\} \cup \\
        &\{Q_{i,j,k} \rightarrow  \lnot Q_{i',j',k'} : (i,j,k) \neq (i',j',k')\} \cup \\
        &\{P_\ell \rightarrow \lnot Q_{m,n,k} : \ell,m,n,k \in \No,\; m < n\},
    \end{align*}
    where $\unc$ is as in the definition of $\Testsm$ (namely, $\unc$ is undecidable).

    We claim that $\T'$ is decidable and infinitely decidable. Let $\varphi$ be a conjunction of literals, and suppose first that $\varphi$ contains a positive literal of the form $P_n$. Then, either $\varphi$ contains a positive literal of the form $Q_{i,j,k}$, in which case it is $\T'$-unsatisfiable, or it does not, in which case deciding whether $\varphi$ is $\T'$-satisfiable reduces to the decision problem for $\Testcfs$-satisfiability. Also, deciding whether $\aleph_0 \in \spec_{\T'}(\varphi)$ reduces to the corresponding problem for $\Testcfs$ in this case. Similarly, if $\varphi$ contains a positive literal of the form $R_{i,j,k}$, we appeal to the decidability and infinite decidability results for $\Testsm$.

    At this point, the rest of the argument is as in the proof of \Cref{thm-f-id}.
\end{proof}

\subsection{Proof of Theorem~\ref{thm-cs-complete}}

\thmcscomplete*
\begin{proof}
    The idea is to create a single decidable theory $\T'$ that plays the role of each of the test theories $\Testcfs$ and $\T_{\infty}$ simultaneously. Let $\T'$ have predicates $P_n$ for each $n \in \N$ (rather than merely each $n \in \No$). Then, let $\T'$ be axiomatized by
    \[
        \{P_{n}\rightarrow\psi_{=n} : n\in\No\}\cup\{P_{\aleph_0} \rightarrow \psi_{\geq n} : n\in\No\}.
    \]
    The point is that, compared to $\Testcfs$, we now have an extra predicate to assert that the domain is infinite.

    The theory $\T'$ is decidable and has computable spectra in virtue of the corresponding results for $\Testcfs$ and $\T_{\infty}$, and we now proceed as in the proof of \Cref{thm-f-cs}.
\end{proof}

\subsection{Proof of Theorem~\ref{thm-n-decidable-complete}}

\thmndecidablecomplete*
\begin{proof}
    The idea is to create a single decidable theory $\T'$ that plays the role of each of the test theories $\{\Tle^{\No},\Testcfs\} \cup \{\T^{k'}_{k} : k,k' \in \No,\; k<k',\; k \neq n\}$ simultaneously. We proceed similarly to the proof of \Cref{thm-shiny-uniform}. Let $\T'$ have predicates $P_m$ and $Q_m$ for each $m \in \No$ and predicates $R_{i,j,k}$ for each $i,j,k \in \No$ with $i < j$ and $i \neq n$. Then, let $\T'$ be axiomatized by
    \begin{align*}
        &\{P_{m}\rightarrow\psi_{=m} : m\in\No\}\cup \\
        &\{Q_{m}\rightarrow\psi_{\leq F(m)} : F(m)\in\mathbb{N}\}\cup\{Q_{i}\rightarrow\neg Q_{j} : i\neq j\} \cup \\
        &\{R_{i,j,k} \rightarrow \psi_{=i} \lor \psi_{=j} : i,j,k \in \No,\; i<j,\; i \neq n\} \cup \{R_{i,j,k}\rightarrow \psi_{=j} : i,j,k \in \No,\; i<j,\; k \in \unc\} \cup \\
        &\{R_{i,j,k} \rightarrow  \lnot R_{i',j',k'} : (i,j,k) \neq (i',j',k')\} \cup \\
        &\{P_i \rightarrow \lnot Q_j : i,j \in \No\} \cup \\
        &\{P_\ell \rightarrow \lnot R_{i,j,k} : \ell,i,j,k \in \No,\; i<j,\; i \neq n\} \cup \\
        &\{Q_\ell \rightarrow \lnot R_{i,j,k} : \ell,i,j,k \in \No,\; i<j,\; i \neq n\},
    \end{align*}
    where $F$ and $\unc$ are as in the definition of $\Tle^{\No}$ and $\T^{k'}_{k}$ respectively (in particular, $\unc$ is undecidable).

    We claim that $\T'$ is decidable and $n$-decidable. Let $\varphi$ be a conjunction of literals, and suppose first that $\varphi$ contains a positive literal of the form $P_m$. Then, either $\varphi$ contains a positive literal of the form $Q_\ell$ or $R_{i,j,k}$, in which case it is $\T'$-unsatisfiable, or it does not, in which case deciding whether $\varphi$ is $\T'$-satisfiable reduces to the decision problem for $\Testcfs$-satisfiability. Also, deciding whether $n \in \spec_{\T'}(\varphi)$ reduces to the corresponding problem for $\Testcfs$ in this case. Similarly, if $\varphi$ contains a positive literal of the form $Q_\ell$ or $R_{i,j,k}$, we appeal to the decidability and $n$-decidability results for $\Tle^{\No}$ and $\T^{j}_{i}$ (when $i \neq n$) respectively.
    
    At this point, the rest of the argument is as in the proof of \Cref{thm-f-n-decidable}.
\end{proof}

\section{Proof of Theorem~\ref{thm-gentle-not-uniform}}

\thmgentlenotuniform*

Throughout the following, fix a decidable $\Sigma$-theory $\T'$ with computable finite spectra. Our goal is to construct a decidable theory $\T$ such that $\T \sqcup \T'$ is decidable but $\T$ is not gentle.

Given a set $S \subseteq \No$, let $\T_S$ be the theory over the empty signature axiomatized by $\{\lnot \psi_{=n} : n \notin S\}$. We will carefully construct a set $S$ so that $\T_S$ and $\T_S \sqcup \T'$ are decidable, but $\T_S$ is not gentle. Let us first list the properties we want $S$ to satisfy; then we will actually construct a set with those properties.

\begin{lemma} \label{lem-s-gentle}
    Suppose $S \subseteq \No$ satisfies the following conditions:
    \begin{itemize}
        \item $S$ is not cofinite; and
        \item it is decidable whether a $\Sigma$-formula $\varphi$ satisfies $\spec_{\T'}(\varphi) \cap (S \cup \{\Inf\}) = \emptyset$.
    \end{itemize}
    Then, $\T_S$ and $\T_S \sqcup \T'$ are decidable, but $\T_S$ is not gentle.
\end{lemma}
\begin{proof}
    First, we have $\spec_{\T_S}(\top) = S \cup \{\Inf\}$, which is neither a finite set of finite cardinalities nor a cofinite set. Thus, $\T_S$ is not gentle.

    Second, we claim that $\T_S$ is decidable. If $\varphi$ is a quantifier-free formula satisfiable in the pure theory of equality, then it has an infinite model, which is also a $\T_S$-interpretation. Thus, $\varphi$ is $\T_S$-satisfiable if and only if it is satisfiable in the pure theory of equality, which is decidable.

    It remains to show that $\T_S \sqcup \T'$ is decidable. Let $\varphi_1$ be a quantifier-free $\Sigma$-formula and $\varphi_2$ be a quantifier-free formula over the empty signature. Our goal is to decide whether $\varphi_1 \land \varphi_2$ is $\T_S \sqcup \T'$-satisfiable. This happens if and only if $\spec_{\T'}(\varphi_1 \land \varphi_2) \cap (S \cup \{\Inf\}) \neq \emptyset$, which is decidable by assumption.
\end{proof}

Now we describe the construction of $S$. Let $\varphi_1, \varphi_2, \dots$ be an enumeration of $\Sigma$-formulas. Then, our set is constructed iteratively by the algorithm in \Cref{alg-main}. We give an intuitive description of the algorithm before proving that it constructs a set $S$ satisfying the properties from \Cref{lem-s-gentle}. The algorithm alternates between processing formulas and processing numbers. The state of the algorithm includes $i \in \No$, which is the index of the next unprocessed formula, and $j \in \No$, which is the next unprocessed number. We also keep track of
\begin{itemize}
    \item $SAT \subseteq \No$, which are the indices of $\Sigma$-formulas $\varphi$ such that our partial construction of $S$ has already guaranteed that $\varphi$ is $\T_S \sqcup \T'$-satisfiable;
    \item $UNSAT \subseteq \No$, which are the indices of $\Sigma$-formulas $\varphi$ such that our partial construction of $S$ has already guaranteed that $\varphi$ is $\T_S \sqcup \T'$-unsatisfiable; and
    \item $PROM \subseteq \No$, which are the indices of $\Sigma$-formulas $\varphi$ such that we promise to make $\varphi$ $\T_S \sqcup \T'$-satisfiable, even though the decisions made by the algorithm so far have not yet fulfilled that promise.
\end{itemize}

\begin{algorithm}[ht]
    \caption{Main loop constructing $S$} \label{alg-main}
    \begin{algorithmic}[1]
        \State let $i,j \gets 1$ \Comment{$i,j \in \No$}
        \State let $S, SAT, UNSAT, PROM \gets \emptyset$ \Comment{$S, SAT, UNSAT, PROM \in \mathbb{P}(\No)$}
        \While{\TRUE}
            \State let $i, SAT, UNSAT, PROM \gets \procform(i,S,SAT,UNSAT,PROM)$
            \State let $j, S, SAT, PROM \gets \procnum(j,S,SAT,PROM)$
        \EndWhile
    \end{algorithmic}
    \end{algorithm}

\begin{algorithm}[ht]
    \caption{$\procform(i,S,SAT,UNSAT,PROM)$} \label{alg-form}
    \begin{algorithmic}[1]
        \If{$\varphi_i$ has a $\T'$-interpretation whose size is in $S$}
            \State let $SAT \gets SAT \cup \{i\}$
            \State let $i \gets i+1$
            \State \Return $i, SAT, UNSAT, PROM$
        \ElsIf{$\varphi_i$ has no $\T'$-interpretations of size at least $j$}
            \State let $UNSAT \gets UNSAT \cup \{i\}$
            \State let $i \gets i+1$
            \State \Return $i, SAT, UNSAT, PROM$
        \Else
            \State let $PROM \gets PROM \cup \{i\}$
            \State let $i \gets i+1$
            \State \Return $i, SAT, UNSAT, PROM$
        \EndIf
    \end{algorithmic}
\end{algorithm}

\begin{algorithm}[ht]
    \caption{$\procnum(j,S,SAT,PROM)$} \label{alg-num}
    \begin{algorithmic}[1]
        \While{$\varphi_k$ has a $\T'$-interpretation of size $j$ for some $k \in PROM$}
            \State let $S \gets S \cup \{j\}$
            \State let $SAT \gets SAT \cup \{k\}$
            \State let $PROM \gets PROM \setminus \{k\}$
            \State let $j \gets j+1$
        \EndWhile
        \State let $j \gets j+1$
        \State \Return $j, S, SAT, PROM$
    \end{algorithmic}
\end{algorithm}

When we process the formula $\varphi_i$, we first check whether we have already made $\varphi_i$ $\T_S \sqcup \T'$-satisfiable, meaning that $\varphi_i$ has a $\T'$-interpretation whose size is in $S$; if so, we add $i$ to $SAT$. Otherwise, we check whether we have already made $\varphi_i$ $\T_S \sqcup \T'$-unsatisfiable, meaning that $\varphi_i$ has no $\T'$-interpretations whose size is in $S$ nor does it have any $\T'$ models size at least $j$ (the first unprocessed number); if so, we add $i$ to $UNSAT$. Otherwise, we promise to make $\varphi_i$ $\T_S \sqcup \T'$-satisfiable, and we add $i$ to $PROM$. Since $\varphi_i$ has now been processed, we increment $i$.

When we process the number $j$, we first check whether some formula that we have promised to satisfy has a $\T'$-interpretation of size $j$; if so, we fulfill that promise by adding $j$ to $S$, updating $SAT$ and $PROM$ accordingly. We increment $j$ and repeat this procedure until we find a $j$ whose inclusion in $S$ is not necessary to fulfill a promise. Since we are not obliged to add $j$ to $S$, we skip it by incrementing $j$ again.

We first justify that our procedure is indeed an algorithm.
\begin{lemma} \label{lem-computable}
    Every step of \Cref{alg-main} is computable.
\end{lemma}
\begin{proof}
    We only need to check that the conditions of the if statements in \Cref{alg-form} are decidable. We can decide whether $\varphi_i$ has a $\T'$-interpretation whose size is in $S$ because $\T'$ has computable finite spectra and $S$ is a finite set at each stage of the procedure. We can decide whether $\varphi_i$ has a $\T'$-interpretation of size at least $j$ by checking the $\T'$-satisfiability of the formula $\varphi_i \land \neq(x_1,\dots,x_j)$ (where the variables $x_k$ are fresh).
\end{proof}

Next, we show that \Cref{alg-form,alg-num} terminate.
\begin{lemma} \label{lem-terminate}
    \Cref{alg-form,alg-num} terminate for all choices of inputs.
\end{lemma}
\begin{proof}
    For \Cref{alg-form}, this is trivial. For \Cref{alg-num}, we need to show that the while loop terminates. This is because $PROM$ is a finite set whose size decreases with each iteration of the loop.
\end{proof}

Now, we show that the set $S$ constructed by \Cref{alg-main} satisfies the conditions from \Cref{lem-s-gentle}.

\begin{lemma} \label{lem-not-cofinite}
    The set $S$ constructed by \Cref{alg-main} is not cofinite.
\end{lemma}
\begin{proof}
    Every iteration of the main loop of \Cref{alg-main} calls $\procnum$. For each call to $\procnum$, we omit some $j$ from the set $S$. Thus, $S$ is not cofinite.
\end{proof}

\begin{lemma} \label{lem-decidable-intersection}
    For the set $S$ constructed by \Cref{alg-main}, it is decidable whether a $\Sigma$-formula $\varphi$ satisfies $\spec_{\T'}(\varphi) \cap (S \cup \{\Inf\}) = \emptyset$.
\end{lemma}
\begin{proof}
    The decision procedure is to run \Cref{alg-main} until the formula $\varphi$ is processed by $\procform$. If the index of $\varphi$ is added to $SAT$ or $PROM$, then we claim that $\spec_{\T'}(\varphi) \cap (S \cup \{\Inf\}) \neq \emptyset$; otherwise, if the index of $\varphi$ is added to UNSAT, we claim that $\spec_{\T'}(\varphi) \cap (S \cup \{\Inf\}) = \emptyset$. By \Cref{lem-computable,lem-terminate}, this constitutes an effective procedure. It remains to justify its correctness.

    If $i \in SAT$, then $\varphi_i$ has a $\T'$-interpretation whose size is in $S$ (elements of $S$ are never removed during the algorithm), so $\spec_{\T'}(\varphi_i) \cap (S \cup \{\Inf\}) \neq \emptyset$.
    
    If $i \in PROM$, there are two cases to consider. First, if $\Inf \in \spec_{\T'}(\varphi_i)$, then $\spec_{\T'}(\varphi_i) \cap (S \cup \{\Inf\}) \neq \emptyset$. Otherwise, observe that since $i$ was added to $PROM$, $\varphi_i$ has a $\T'$-interpretation of size at least $j$ (the first unprocessed number at the point $\varphi_i$ was processed). In particular, in this case, $\varphi_i$ has a \emph{finite} $\T'$-interpretation of size at least $j$. Let $j'$ be the size of such a $\T'$-interpretation. Then, when $j'$ is processed by $\procnum$ at a future stage, it will be added to $S$. Thus, $\spec_{\T'}(\varphi_i) \cap (S \cup \{\Inf\}) \neq \emptyset$.

    Finally, if $i \in UNSAT$, then $\varphi_i$ does not have a $\T'$-interpretation whose size is in $S \cap [j-1]$ (where $j$ is the first unprocessed number at the point $\varphi_i$ was processed). Further, $\varphi_i$ does not have a $\T'$-interpretation of size at least $j$. Therefore, $\spec_{\T'}(\varphi_i) \cap (S \cup \{\Inf\}) = \emptyset$.
\end{proof}

\begin{proof}[Proof of \Cref{thm-gentle-not-uniform}]
    This is immediate from \Cref{lem-s-gentle,lem-not-cofinite,lem-decidable-intersection}.
\end{proof}

\section{Proofs of Claims in Example~\ref{ex-cm-cs}} \label{appendix-cm-cs-example}

Here we prove the claimed facts about $\Tgr$, which recall is the theory axiomatized by $\{P_{n}\rightarrow \psi_{\geq F(n)} : n\in\No,\; F(n)\in\No\} \cup \{P_n \rightarrow \psi_{\geq m} : m,n \in \No,\; F(n) = \aleph_0\}$, where $F:\No\rightarrow\No\cup\{\Inf\}$ is such that $\{(m,n) \in {\No}^2 : F(m)\geq n\}$ is decidable, $\{n : F(n)=\aleph_{0}\}$ is undecidable.

\begin{proposition}
    The theory $\Tgr$ is smooth.
\end{proposition}
\begin{proof}
    This is immediate from the fact that $\Tgr$ has an existential axiomatization.
\end{proof}

\begin{proposition}
    The theory $\Tgr$ has computable spectra.
\end{proposition}
\begin{proof}
    We want to compute whether $k \in \spec_{\Tgr}(\varphi)$, where $\varphi$ is a quantifier-free formula and $k \in \N$. Without loss of generality, $\varphi$ is a conjunction of literals. We can then write $\varphi = \varphi_1 \land \varphi_2$, where $\varphi_1$ is a conjunction of literals of the form $P_n$ or $\lnot P_n$, and $\varphi_2$ is a conjunction of equalities and disequalties. If $\varphi_2$ is unsatisfiable in pure first-order logic, then clearly $k \notin \spec_{\Tgr}(\varphi)$. So suppose $\varphi_2$ is satisfiable in pure first-order logic, and let $m$ be the size of its smallest model (which is finite and can be effectively computed). Then, we have $k \in \spec_{\Tgr}(\varphi)$ if and only if $k \ge m$ and $k \ge F(n)$ for each $n \in \No$ such that $P_n$ occurs positively in $\varphi_1$. If $k = \aleph_0$, then this condition is trivially satisfied, so $k \in \spec_{\Tgr}(\varphi)$. Otherwise, if $k \in \No$, our assumptions on $F$ imply that this condition can be effectively checked.
\end{proof}

\begin{proposition}
    The theory $\Tgr$ does not have a computable minimal model function.
\end{proposition}
\begin{proof}
    Computing $\minmod_{\Tgr}(P_n)$ would allow us to decide whether $F(n) = \Inf$, which is undecidable by our assumption on $F$.
\end{proof}

\section{Pseudocode for the Old Combination Methods}

In \Cref{sec-pseudocode}, we presented pseudocode for each of the new combination methods we introduced. For completeness, here we give pseudocode for the previously known combination methods, which were discussed in \Cref{sec-known}. Each function in \Cref{alg-old-methods} shows how each combination method allows us to fulfill the requirement of Fontaine's lemma (\Cref{lem-fontaine,alg-fontaine}).

\begin{algorithm}[ht]
    \caption{Pseudocode for the old combination methods. Each function returns whether $\spec_{\T_1}(\varphi_1) \cap \spec_{\T_2}(\varphi_2) \neq \emptyset$, where $\varphi_1$ and $\varphi_2$ are conjunctions of literals over $\Sigma_1$ and $\Sigma_2$ respectively. We assume throughout that $\T_1$ and $\T_2$ are decidable.} \label{alg-old-methods}
    \begin{algorithmic}[1]
        \Require $\T_1$ is shiny
        \Function{shiny}{$\varphi_1, \varphi_2$}
            \If{$\varphi_1$ is $\T_1$-unsatisfiable}
                \State \Return \FALSE
            \EndIf
            \State $k \gets \minmod_{\T_1}(\varphi_1)$
            \State \Return $\varphi_2 \land \neq(x_{1},\ldots,x_{k})$ is $\T_2$-satisfiable \Comment{$x_{1},\ldots,x_{k}$ are fresh}
        \EndFunction

        \Statex

        \Require $\T_1$ and $\T_2$ are stably infinite
        \Function{Nelson--Oppen}{$\varphi_1, \varphi_2$}
            \State \Return $\varphi_1$ is $\T_1$-satisfiable and $\varphi_2$ is $\T_2$-satisfiable
        \EndFunction

        \Statex

        \Require $\T_1$ is gentle and $\T_2$ has computable finite spectra
        \Function{gentle}{$\varphi_1, \varphi_2$}
            \If{$\spec_{\T_1}(\varphi_1)$ is finite}
                \ForAll{$n \in \spec_{\T_1}(\varphi_1)$}
                    \If{$n \in \spec_{\T_2}(\varphi_2)$}
                        \State \Return \TRUE
                    \EndIf
                \EndFor
                \State \Return \FALSE
            \Else \Comment{$\spec_{\T_1}(\varphi_1)$ is cofinite}
                \State $k \gets \max(\mathbb{N} \setminus \spec_{\T_1}(\varphi_1))$
                \ForAll{$n \in \spec_{\T_1}(\varphi_1) \cap [k]$}
                    \If{$n \in \spec_{\T_2}(\varphi_2)$}
                        \State \Return \TRUE
                    \EndIf
                \EndFor
                \State \Return $\varphi_2 \land \neq(x_{1},\ldots,x_{k+1})$ is $\T_2$-satisfiable \Comment{$x_{1},\ldots,x_{k+1}$ are fresh}
            \EndIf
        \EndFunction
    \end{algorithmic}
\end{algorithm}

\section{Proof of Proposition~\ref{prop-inclusion}}

\propinclusion*

\begin{proof}
    We prove each of the claimed inclusions.

    $\class{$n$-decidable} \subseteq \class{}$ and $\class{ID} \subseteq \class{}$: These are clear, since all of the theories are assumed to be decidable.

    $\class{CFS} \subseteq \class{$n$-decidable}$: If $\T \in \class{CFS}$, then we can decide whether $k \in \spec_\T(\varphi)$ for any given $k \in \No$. In particular, we can decide whether $n \in \spec_\T(\varphi)$, so $\T \in \class{$n$-decidable}$

    $\class{\coquagen} \subseteq \class{CFS}$, $\class{CS} \subseteq \class{CFS}$, and $\class{CS} \subseteq \class{ID}$: These are immediate from the definitions.

    $\class{SI} \subseteq \class{ID}$: If $\T \in \class{SI}$, then we have $\aleph_0 \in \spec_\T(\varphi)$ if and only if $\varphi$ is $\T$-satisfiable. Since $\T$ is decidable, we can decide whether $\aleph_0 \in \spec_\T(\varphi)$. Thus, $\T \in \class{ID}$.

    $\class{\quagen} \subseteq \class{\coquagen}$: Suppose $\T \in \class{\quagen}$. Then, $\T$ has computable finite spectra. Also, the spectrum of every quantifier-free formula is either a finite set of finite cardinalities, or $A\cup\{\aleph_{0}\}$ for an $A\in\mathfrak{F}$. But $A\in\mathfrak{F}$ implies $\mathbb{N} \setminus A \notin \mathfrak{F}$. Thus, $\T \in \class{\coquagen}$.

    $\class{gentle} \subseteq \class{\quagen}$: If $\T \in \class{gentle}$, then $\T$ has computable finite spectra, and the spectrum of every quantifier-free formula is either a finite set of finite cardinalities or a cofinite set. Since $\mathfrak{F}$ is a free filter, it contains all cofinite sets. Thus, $\T \in \class{\quagen}$.

    $\class{gentle} \subseteq \class{CS}$ and $\class{SM+CS} \subseteq \class{CS}$: These are immediate from the definitions.

    $\class{SM+CS} \subseteq \class{SI}$: This is because every smooth theory is stably infinite.

    $\class{$n$-shiny} \subseteq \class{gentle}$: If $\T \in \class{$n$-shiny}$, then we can compute an explicit representation of the spectrum of a given quantifier-free formula, and the definition of $n$-shininess implies that this spectrum is either a finite set of finite cardinalities or a cofinite set. Thus, $\T \in \class{gentle}$.

    $\class{shiny} \subseteq \class{$n$-shiny}$: If $\T \in \class{shiny}$, the spectrum of every $\T$-satisfiable quantifier-free formula is of the form $\{m \in \N : m \ge k\}$ for some $k \in \No$ that can be computed. Thus, $\T \in \class{$n$-shiny}$.

    $\class{shiny} \subseteq \class{SM+CS}$: If $\T \in \class{shiny}$, then $\T$ is smooth by definition and the spectrum of every $\T$-satisfiable quantifier-free formula is of the form $\{m \in \N : m \ge k\}$ for some $k \in \No$ that can be computed, which allows us to compute whether a given $n \in \N$ is in the spectrum. Thus, $\T \in \class{SM+CS}$.
\end{proof}

\section{Proof of Proposition~\ref{prop-strict-inclusion}}

We here prove \Cref{prop-strict-inclusion}.
For each set \class{X}, we find a theory $\T$ such that $\T \in \class{X} \setminus \class{Y}$ for all sets $\class{Y}$ such that there is no upward path in \Cref{fig:lattice} from $\class{X}$ to $\class{Y}$ and $(\class{X},\class{Y}) \neq (\class{\coquagen{}},\class{\quagen{}})$.

\subsection{\tp{Examples of Decidable and $n$-Decidable Theories}{Examples of Decidable and n-Decidable Theories}}\label{sec:exdec}

\begin{lemma}
    For all $k \in \No$, $\Th{d}^{k}$ is decidable but not infinitely decidable. It is $n$-decidable if $k=n-1$ and not $n$-decidable for $k=n$; in particular, $\Th{d}^{k}$ does not have computable finite spectra for any $k$.
\end{lemma}

\begin{proof}
\begin{enumerate}
    \item To show that $\Th{d}^{n}$ is decidable we can assume, as many times before, that $\varphi$ is a conjunction of literals with at most one positive $P$-literal.
    If $\varphi$ contains no positive $P$-literals we state $\varphi$ is $\Th{d}^{n}$-satisfiable if and only if $\varphi^{\prime}$ is $\Teq$-satisfiable, and $n+1$ is in its spectra iff $\minmod_{\Teq}(\varphi^{\prime})\leq n+1$:
    indeed, given any $\Teq$-interpretation that satisfies $\varphi^{\prime}$, we just make all $P_{k}$ false.
    If $\varphi$ contains the literal $P_{2k}$ we just need to check, as in the case of $\Tle$, whether $\minmod_{\Teq}(\varphi^{\prime})\leq F(k)$, and $n+1$ will be in this spectrum iff $\minmod_{\Teq}(\varphi^{\prime})\leq n+1\leq F(k)$.
    If, instead, $\varphi$ contains the literal $P_{2k+1}$, then $\varphi$ is $\Th{d}^{n}$-satisfiable if and only if $\varphi^{\prime}$ is $\Teq$-satisfiable, and $n+1$ is in this spectrum iff $\minmod_{\Teq}(\varphi^{\prime})\leq n+1$:
    just take a $\Teq$-interpretation that satisfies $\varphi^{\prime}$ with cardinality the maximum of $\minmod_{\Teq}(\varphi^{\prime})$ and $n+1$, make $P_{2k+1}$ true and all other $P_{i}$ false.

    \item The proof above, in particular, shows $\Th{d}^{n-1}$ is $n$-decidable.

    \item $n\in\spec_{\Th{d}^{n}}(P_{2k+1})$ if and only if $k\notin \unc$, meaning $\Th{d}^{n}$ is not $n$-decidable nor has computable finite spectra.

    \item $\aleph_{0}\in\spec_{\Th{d}^{n}}(P_{2k})$ if and only if $F(k)=\aleph_{0}$, an undecidable problem, meaning the theory is not infinitely decidable. \qedhere
    \end{enumerate}
\end{proof}

\subsection{Example of an Infinitely Decidable Theory}

\begin{lemma}
    $\T_{n}^{n+1}$ is decidable and infinitely decidable, but is neither $n$-decidable nor stably infinite.
\end{lemma}

\begin{proof}
\begin{enumerate}
    \item See \Cref{sec:testsm} for the proof of decidability.
    \item $\T_{n}^{n+1}$ is infinitely decidable as it contains no infinite interpretations.
    \item This theory is not $n$-decidable as $n\in\spec_{\T_{n}^{n+1}}(P_{k})$ if and only if $k\notin \unc$.
    \item It is also not stably infinite as it contains no infinite interpretations. \qedhere
\end{enumerate}
\end{proof}

\subsection{Example of a Theory With Computable Finite Spectra}

\begin{lemma}
    $\Th{cfs}$ is decidable and has computable finite spectra, but is neither infinitely decidable nor co-$\mathfrak{F}$-QG for any free filter $\mathfrak{F}$.
\end{lemma}

\begin{proof}
\begin{enumerate}
    \item Let $\varphi$ be a conjunction of literals with at most one positive $P$-literal.
    If $\varphi$ has no positive $P$-literals we proceed as usual, $\varphi$ being $\Th{cfs}$-satisfiable iff $\varphi^{\prime}$ is $\Teq$-satisfiable;
    and, in that case, $\spec_{\Th{cfs}}(\varphi) = [\minmod_{\Teq}(\varphi^{\prime}),\aleph_{0}]$.

    If $\varphi$ contains $P_{1}$, again it is $\Th{cfs}$-satisfiable iff $\varphi^{\prime}$ is $\Teq$-satisfiable, but this time we have $\spec_{\Th{cfs}}(\varphi)$ equal to $\{\aleph_{0}\}$.
    Finally, if $\varphi$ contains $P_{k}$ for $k\geq 2$, then it is $\Th{cfs}$-satisfiable iff $\varphi^{\prime}$ is $\Teq$-satisfiable and $\minmod_{\Teq}(\varphi^{\prime})\leq F(k)$, what can be tested algorithmically;
    not only that, but $i \in \No$ is in $\spec_{\Th{cfs}}(\varphi)$ iff $\minmod_{\Teq}(\varphi^{\prime})\leq i\leq F(k)$, what is decidable.
    Indeed, given a $\Teq$-interpretation of cardinality $i$ that satisfies $\varphi^{\prime}$, we turn it into a $\Th{cfs}$-interpretation by making $P_{k}$ true, and all other $P_{j}$ for $k\neq j$ false.

    \item The proof above also shows that $\Th{cfs}$ has computable finite spectra.

    \item $\Th{cfs}$ is not infinitely decidable as the spectrum of $P_{k}$, for $k\geq 2$, contains $\aleph_{0}$ if and only if $F(k)=\aleph_{0}$.

    \item $\Th{cfs}$ is also not \coquagen{} as $\spec_{\Th{cfs}}(P_{1})=\{\aleph_{0}\}$. \qedhere
\end{enumerate}
\end{proof}

\subsection{Example of a Theory With Computable Spectra}

\begin{lemma}
    $\Th{cs}$ is decidable and has computable spectra, but is neither stably infinite nor co-$\mathfrak{F}$-QG for any free filter $\mathfrak{F}$.
\end{lemma}

\begin{proof}
\begin{enumerate}
    \item For the signature $\Sp$ the $P$-literals will be $P$ and $\neg P$.
    Let $\varphi$ be a conjunction of literals. If $\varphi$ contains neither $P$ not $\neg P$, then it is $\Th{cs}$-satisfiable iff $\phi^{\prime}$ is $\Teq$-satisfiable, and its spectrum is $\{1,\aleph_{0}\}$ if $\minmod_{\Teq}(\varphi^{\prime})=1$ and $\{\aleph_{0}\}$ otherwise; indeed, take a $\Teq$-interpretation of the appropriate cardinality satisfying $\phi^{\prime}$ and set $P$ to true if the cardinality is 1 and false if the cardinality is $\aleph_0$. Similarly, if $\varphi$ contains $\neg P$, it is $\Th{cs}$-satisfiable iff $\phi^{\prime}$ is $\Teq$-satisfiable, and then its spectrum is just $\{\aleph_{0}\}$.
    If $\varphi$ contains $P$, then it is $\Th{cs}$-satisfiable iff $\minmod_{\Teq}(\varphi^{\prime})=1$, and in that case its spectrum is just $\{1\}$;
    indeed, given a $\Teq$-interpretation satisfying $\varphi^{\prime}$ with cardinality $1$, we turn it into a $\Th{cs}$-interpretation by making $P$ true.

    \item As any quantifier-free formula is a conjunction of literals, its spectrum is the union of the spectra of its cubes, which as we saw above are computable.
    We thus get $\Th{cs}$ has computable spectra.

    \item The spectrum of $P$ is $\{1\}$, so $\Th{cs}$ is not stably infinite.

    \item The spectrum of $\neg P$ is $\{\aleph_{0}\}$, so $\Th{cs}$ is not \coquagen. \qedhere
\end{enumerate}
\end{proof}

\subsection{Example of a Stably Infinite Theory}

\begin{lemma}
    $\Th{si}$ is decidable and stably infinite, but it is not $n$-decidable for any $n \in \No$.
\end{lemma}

\begin{proof}
    \begin{enumerate}
        \item To see that $\Th{si}$ is decidable, again we can look at conjunctions of literals $\varphi$ with at most one positive $P$-literal.
        We then state that $\varphi$ is $\Th{si}$-satisfiable iff $\varphi^{\prime}$ is $\Teq$-satisfiable.
        Indeed, given an infinite interpretation of $\varphi^{\prime}$, we transform it into a $\Th{si}$-interpretation by making $P_{m}$ true if and only if it occurs as a positive literal in $\varphi$.
        \item Our construction in the proof of the decidability of $\Th{si}$ shows this theory is stably infinite.
        \item $\Th{si}$ is not $n$-decidable, for any $n\in\No$, as $n\in\spec_{\Th{si}}(P_{m})$ iff $m\notin \unc$.\qedhere
    \end{enumerate}
\end{proof}

\subsection{Examples for Quasi-Gentleness and Co-Quasi-Gentleness}

\begin{lemma}
    $\Tle^{\No}$ is decidable and \quagen{} but not infinitely decidable.
\end{lemma}
\begin{proof}
\begin{enumerate}
    \item As proven in \Cref{sec:prooftestg}, the theory $\Tle^{S}$, for a decidable $S \in \mathfrak{F}$, is decidable and \quagen{}. In particular, $\Tle^{\No}$ is decidable and \quagen{}.
    \item We have $\aleph_0 \in \spec_{\Tle^{\No}}(P_n)$ if and only if $F(n) = \aleph_0$, so $\Tle^{\No}$ is not infinitely decidable. \qedhere
\end{enumerate}
\end{proof}

\begin{lemma}
    Suppose $A \subset \mathbb{N}$ is a decidable set such that $A \notin \mathfrak{F}$ and $\mathbb{N} \setminus A \notin \mathfrak{F}$. Then, $\Tle^A$ is decidable and \coquagen{} but neither \quagen{} nor infinitely decidable.
\end{lemma}
\begin{proof}
\begin{enumerate}
    \item As proven in \Cref{sec:prooftestg}, the theory $\Tle^{S}$, for a decidable $S$ such that $\mathbb{N} \setminus S \notin \mathfrak{F}$, is decidable and \coquagen{}. In particular, $\Tle^A$ is decidable and \coquagen{}.
    \item We have $\spec_{\Tle^A}(\top) = A \cup \{\aleph_0\}$, and $A \notin \mathfrak{F}$, so $\Tle^A$ is not \quagen{}.
    \item We have $\aleph_0 \in \spec_{\Tle^A}(P_n)$ if and only if $F(n) = \aleph_0$, so $\Tle^A$ is not infinitely decidable. \qedhere
\end{enumerate}
\end{proof}

\subsection{Example of a Gentle Theory}

\begin{lemma}
    If $n \ge 2$, then $\Tleqn$ is decidable and gentle, but neither stably infinite nor $m$-shiny for any $m \in \No$.
\end{lemma}

\begin{proof}
\begin{enumerate}
    \item $\Tleqn$ is a $\Sigma_{1}$-theory, an empty signature with finitely many sorts, so it is decidable by \cite{FroCoS2025}.
    \item We state that, if $\varphi$ is a $\Tleqn$-satisfiable quantifier-free formula, then $\spec_{\Tleqn}(\varphi)$ equals the interval $[\minmod_{\Teq}(\varphi), n]$.
    Of course any $i\in \spec_{\Tleqn}(\varphi)$ must satisfy $i\leq n$;
    furthermore, for any $\minmod_{\Teq}(\varphi)\leq i\leq n$, using that $\Teq$ is smooth we can find a $\Teq$-interpretation $\A$ that satisfies $\varphi$ with $|\dom{\A}|=i$, which is therefore a $\Tleqn$-interpretation as well. Since $[\minmod_{\Teq}(\varphi), n]$ is a a finite set of finite cardinalities for which we can compute an explicit representation, $\Tleqn$ is gentle.
    \item $\Tleqn$ is not stably infinite, since it has no infinite models.
    \item We have $\spec_{\Tleqn}(\top) = [1, n]$, which is a finite set but not a singleton, which is impossible for an $m$-shiny theory. \qedhere
\end{enumerate}
\end{proof}

\subsection{Example of a Smooth Theory With Computable Spectra}

\begin{lemma}
    $\Tinfty$ is decidable, smooth and has computable spectra, but it is not co-$\mathfrak{F}$-QG for any free filter $\mathfrak{F}$.
\end{lemma}

\begin{proof}
    \begin{enumerate}
        \item \cite{FroCoS2025} proves $\Tinfty$ is decidable, since it is over an empty signature.
        \item \cite{CADE} proves $\Tinfty$ is smooth.
        \item Clearly the spectrum of any $\Tinfty$-satisfiable quantifier-free formula is exactly $\{\aleph_{0}\}$, so it has computable spectra.
        \item Because the spectrum of a quantifier-free $\Tinfty$-satisfiable formula is $\{\aleph_{0}\}$, this theory cannot be \coquagen. \qedhere
    \end{enumerate}
\end{proof}

\subsection{\tp{Example of an $n$-Shiny Theory}{Example of an n-Shiny Theory}}

\begin{lemma}
    $\Th{n-s}^{n}$ is decidable and $n$-shiny, but it is not stably infinite.
\end{lemma}

\begin{proof}
    \begin{enumerate}
        \item Take a conjunction of literals $\phi$.
        If it doesn't contain $P$, then $\phi$ is $\Th{n-s}^{n}$-satisfiable iff $\phi^{\prime}$ is $\Teq$-satisfiable (where $\phi^{\prime}$ is the formula without its $P$-literal if there is one), and its spectrum equals 
        \[[\max(\minmod_{\Teq}(\phi^{\prime}),n),\aleph_{0}].\]
        
        If $\phi$ contains $P$ (and does not contain $\neg P$), then it is $\Th{n-s}^{n}$-satisfiable iff $\phi^{\prime}$ is $\Teq$-satisfiable and $\minmod_{\Teq}(\phi^{\prime})\leq n$, and its spectrum equals $\{n\}$.
        From this we conclude that $\Th{n-s}^{n}$ is $n$-shiny.

        \item $\Th{n-s}^{n}$ is, however, not stably infinite, as $P$ has no infinite $\Th{n-s}^{n}$-interpretations.
    \end{enumerate}
\end{proof}

\subsection{Example of a Shiny Theory}

\begin{lemma}
    $\Tgeqn$ is decidable and shiny.
\end{lemma}

\begin{proof}
    \begin{enumerate}
        \item \cite{FroCoS2025} proves $\Tgeqn$ is decidable, since it is over an empty signature.
        \item \cite{LPAR} proves it is also shiny.
    \end{enumerate}
\end{proof}

\end{document}